\documentclass[sigconf, nonacm, pdfa]{acmart}

\usepackage{threeparttable}

\usepackage{url}
\usepackage{amsmath,amsfonts}
\usepackage[ruled,linesnumbered, vlined]{algorithm2e}
\usepackage{algorithmic}
\usepackage{graphicx}
\usepackage{booktabs}
\usepackage{textcomp}
\usepackage{mathrsfs}
\usepackage{blindtext}
\usepackage{amsthm}
\usepackage{multirow}
\usepackage{bigstrut}
\usepackage{array}
\usepackage{amsthm}
\usepackage{tabularx}
\usepackage{subfig}
\usepackage{enumitem}
\usepackage{multirow}

\usepackage{mathrsfs}
\usepackage{verbatim}
\usepackage{color}
\usepackage{multicol}

\usepackage[normalem]{ulem}
\newtheorem{myDef}{Definition}
\newtheorem{myPro}{Problem}
\newtheorem{myLem}{Lemma}
\newtheorem{theorem}{Theorem}
\newtheorem{definition}{Definition}

\usepackage[nice]{nicefrac}
\usepackage{diagbox}

\usepackage{bm}

\newcommand\vldbdoi{10.14778/3746405.3746433}
\newcommand\vldbpages{3134 - 3148}
\newcommand\vldbvolume{18}
\newcommand\vldbissue{9}
\newcommand\vldbyear{2025}
\newcommand\vldbauthors{Fei TENG, Haoyang LI and Lei CHEN}
\newcommand\vldbtitle{\shorttitle} 
\newcommand\vldbavailabilityurl{https://github.com/XinTT/LLMLog}
\newcommand\vldbpagestyle{empty}

\def\BibTeX{{\rm B\kern-.05em{\sc i\kern-.025em b}\kern-.08em
		T\kern-.1667em\lower.7ex\hbox{E}\kern-.125emX}}

\renewcommand{\proofname}{Proof Sketch}
\definecolor{mygray}{gray}{0.6}

\begin{document}
\title{LLMLog: Advanced Log Template Generation via LLM-driven Multi-Round Annotation}

\author{Fei TENG}
\affiliation{
  \institution{CSE, HKUST}
}
\email{fteng@connect.ust.hk}

\author{Haoyang LI$^*$}\thanks{*Corresponding Author}
\affiliation{
  \institution{Computing, PolyU}
}
\email{haoyang-comp.li@polyu.edu.hk}

\author{Lei CHEN}
\affiliation{
	\institution{DSA, HKUST \& HKUST (GZ)}
}
\email{leichen@cse.ust.hk}

\begin{abstract}
Modern computing systems, such as HDFS and Spark, produce vast quantities of logs that developers use for tasks like anomaly detection and error analysis. To simplify log analysis, template generation methods have been proposed to standardize log formats, transforming unstructured data into structured templates. Existing heuristic-based methods and neural network-based methods suffer from low accuracy problems due to the reliance on handcrafted heuristics or specific log patterns in training sets.
Recently, large language models (LLMs) have shown great potential in log template generation. However, they often struggle with ambiguous, complex, or highly specific log content, which can lead to errors in generating accurate templates. To address these challenges, we propose LLMLog, a multi-round annotation framework with adaptive in-context learning. We first propose an edit-distance-based similarity metric to evaluate log similarity. Then, we introduce a method to select the most informative $k$ unlabeled logs for annotation by considering both the representativeness of the logs and the confidence of LLM predictions. Additionally, we design an adaptive context selection strategy that adaptively selects labeled logs to ensure comprehensive keyword coverage for unlabeled logs. These labeled logs serve as the context for LLMs to better understand the unlabeled logs, thereby enhancing the accuracy of template generation. Extensive experiments on sixteen datasets demonstrate that LLMLog outperforms the state-of-the-art approaches.
\end{abstract}

\twocolumn

\maketitle

\pagestyle{\vldbpagestyle}
\begingroup\small\noindent\raggedright\textbf{PVLDB Reference Format:}\\
\vldbauthors. \vldbtitle. PVLDB, \vldbvolume(\vldbissue): \vldbpages, \vldbyear.\\
\href{https://doi.org/\vldbdoi}{doi:\vldbdoi}
\endgroup
\begingroup
\renewcommand\thefootnote{}\footnote{\noindent
This work is licensed under the Creative Commons BY-NC-ND 4.0 International License. Visit \url{https://creativecommons.org/licenses/by-nc-nd/4.0/} to view a copy of this license. For any use beyond those covered by this license, obtain permission by emailing \href{mailto:info@vldb.org}{info@vldb.org}. Copyright is held by the owner/author(s). Publication rights licensed to the VLDB Endowment. \\
\raggedright Proceedings of the VLDB Endowment, Vol. \vldbvolume, No. \vldbissue\ %
ISSN 2150-8097. \\
\href{https://doi.org/\vldbdoi}{doi:\vldbdoi} \\
}\addtocounter{footnote}{-1}\endgroup

\ifdefempty{\vldbavailabilityurl}{}{
\vspace{.3cm}
\begingroup\small\noindent\raggedright\textbf{PVLDB Artifact Availability:}\\
The source code, data, and/or other artifacts have been made available at \url{\vldbavailabilityurl}.
\endgroup
}

\section{Introduction} \label{sec:intro}
Modern computing systems, such  as HDFS and Spark, generate vast quantities of logs, which offer a wealth of information about system runtime behavior.
Developers can use the recordings to debug and maintain the computer system, including anomaly detection~\cite{vldbanomolydetection,vldblog2,anamoly1,anomaly2,anomaly3,wwwlog} 
and error analysis~\cite{analysis1,analysis2,analysis3}. 
However, the massive volume of logs causes the difficulty of system developers and maintain staffs to detect the anomalies and track errors.
To facilitate easier understanding and analysis, motivated by pattern mining approaches~\cite{vldbpattern1,vldbpattern2,sigmodpattern1}, researchers have proposed generating templates for these massive logs.
For example, as shown in Figure 1, template generation creates structured formats to consistently standardize the logs, making them easier to parse, analyze, and maintain.
This process extracts structured patterns from unstructured log data, which transforms raw logs into more organized formats. 
As a result, developers can quickly identify fields with anomalies or errors, making it easier to debug issues and maintain system health.
The diversity of logs within a system, such as configuration, user, and error logs, often maps to multiple templates rather than a single one, complicating their effective template generation. This challenge is further exacerbated by the ever-growing volume of logs and their variability over time due to system updates.

\begin{figure}[]
	\color{black}
	\centering
	\includegraphics[width=1\columnwidth]{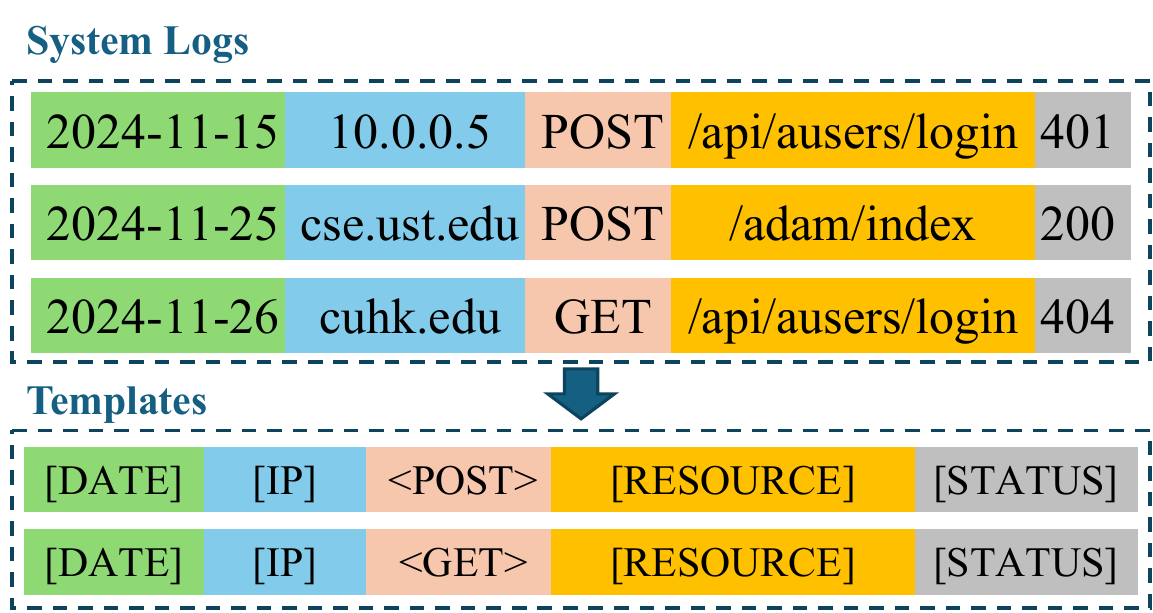}
	\caption{Template generation from system logs.
		As shown in the figure,we can generate two templates for the three logs from HTTP. Templates help structure the logs by assigning types or categories to each word in the logs.}
	\label{fig:template_generation}
\end{figure}

Depending on the technique of generating templates for logs, current approaches can be categorized into three types, i.e., heuristic-based methods~\cite{drain,swisslog,logmine,shio,logsig}, neural network-based methods	~\cite{logppt,nulog}, and LLM-based methods~\cite{Divlog}.
Firstly, heuristic-based methods~\cite{drain,swisslog,logmine,shio,logsig} separate user logs into different clusters using handcrafted rules or heuristics, such as assuming all logs are of equal length if they have the same template. However, such predefined heuristics or rules cannot completely match patterns of logs. 
Consequently, they lack the flexibility and adaptability to standardize diverse logs.
Secondly, neural-network-based approaches train neural networks to predict the type of each word in the logs, 
such as \texttt{[DATE]} and \texttt{[IP]} in Figure 1. 
However, the success of supervised approaches relies on numerous annotated labels to train the neural networks, which requires expensive human effort~\cite{Divlog,logppt}.

Recently, large language models (LLMs)~\cite{llm1,llm2,llm3,llm4,llm5,llm6} pre-trained on diverse corpora, such as text and code, have shown exceptional abilities in understanding text, achieving significant success in natural language processing tasks~\cite{nl1,nl2,llm4}, such as question answering~\cite{qa1,qa2,qa3}. 
Therefore, researchers have explored using LLMs to comprehend text in logs and generate templates for them.
Typically, they provide each user log along with candidate word labels to the LLMs, instructing them to label each word in the log to create its template. However, due to the presence of complex, specific, and ambiguous words in logs (e.g., a username  \texttt{asjaks001}), LLMs may struggle to interpret these words accurately, leading to challenges in generating correct templates.

To enhance the understanding of words in logs,
researchers~\cite{annotation1,annotation2,annotation3,annotation4} have proposed a multi-round annotation framework for large language models (LLMs), which typically consists of  an annotation component and an in-context learning (ICL) component.
Firstly,
in the annotation component, each round involves labeling a subset of logs that lack corresponding templates. Researchers typically define a similarity score among logs, such as cosine similarity between log embeddings~\cite{embedding1,embedding2,embedding3,embedding4}. Based on this score, they select a set of logs within a budget $k_a$ to represent the other unlabeled logs.
Second,
the in-context learning component involves selecting top-$k_c$ similar logs with their templates to serve as context for unlabeled logs~\cite{Divlog,icl1,icl4}. This contextual information helps the LLMs understand the correlation between words and labels, enabling them to predict labels for the unlabeled logs.
After each round, the framework filters logs with lower prediction confidence  as unlabeled. It then iteratively repeats the annotation and ICL steps until the LLMs generate templates for all logs.

Nevertheless, despite the success of existing frameworks, there are still three limitations. 
Firstly, they define log similarity based on log embeddings in both annotation and in-context learning components.
However, this embedding similarity overlooks crucial aspects such as important words (e.g., POST), and more emphasizes useless words (e.g., time stamps or IP addresses in figure 1).
Thus, in annotation or in-context learning, the similar logs may share duplicate time stamps while lack crucial words, offering insufficient information.
Second, in the annotation component, existing approaches typically select the {\color{black}top-$k_a$} unlabeled logs for annotation where selected logs are similar to as many logs as possible. However, this approach neglects the importance of LLM prediction confidence, where logs with low prediction confidence often require more annotations.
Third, in the in-context learning component, existing approaches~\cite{adaicl, vldbicl1,icl,vldbicl2} select a fixed number of top-$k_c$ similar labeled logs to serve as context for each unlabeled log. However, using a fixed number of labeled logs may fail to provide sufficient context for the LLM, potentially leading to the generation of incorrect templates.

To address the aforementioned limitations, we propose an LLM-driven multi-round annotation framework with adaptive in-context learning, called LLMLog.
Specifically, this framework defines an edit-distance based log similarity metric which emphasize the important keywords by measuring word insertion/deletion/replacement operations. 
The most useful logs in annotation and ICL can be respectively identified by adaptively varying the wordset.
Benefit from similarity metric, we define activated logs for representativeness of annotated logs.

In each round of the annotation component, we first adaptively determine the budget for that round. 
Then, we propose a greedy algorithm to annotate logs by jointly optimizing the confidence of LLM predictions and the total number of activated logs, taking into account logs with low confidence and high representativeness.
In adaptive ICL component, we select the minimum number of demonstrative logs from annotation set for each input log by ensuring all keywords are covered from input log in a greedy manner.
Compared to the fixed top-$k_c$ strategy, our adaptive selection can precisely guide LLM by ensuring the informativeness of contextual information from each demonstrative log.
The contributions of this paper are summarized as follows.

\begin{itemize}[leftmargin=*]
	\item We propose LLMLog, a novel LLM-driven multi-round annotation framework with adaptive in-context learning for log template generation. The framework addresses limitations in existing methods by iteratively improving annotation and template generation processes.

	\item {\color{black} We propose an semantic edit-distance-based metric for keyword coverage and an LLM feedback metric considering word consistency and confidence. Additionally, we introduce an adaptive strategy to dynamically allocate the annotation budget. Finally, we formulate the NP-hard multi-round log annotation problem and develop a greedy algorithm with theoretical guarantees.}

	\item  We develop an adaptive strategy for selecting the minimum number of demonstrative logs as context for each unlabeled log. This approach ensures that all keywords from the input log are covered, enhancing the LLM's understanding and improving template generation accuracy compared to fixed top-k strategies.

\item 
Extensive experiments on 16 datasets demonstrate that LLMLog achieves higher accuracy than state-of-the-art baselines while reducing computational and API costs through adaptive demonstration selection and efficient log annotation.

\end{itemize}

\section{Preliminary and Related Works}\label{sec:related}
In this section, we first present the preliminaries on logs and log template generation. Then, we discuss the related works. Important notations of  this paper are summarized in Table~\ref{tab:notation}.

\subsection{Log Template Generation Problem} \label{ssec:templategeneration}
Logs are textual sequences that document the behaviors and events of a system. 
Formally, a log $s$ is one system message recording a system event, composed of a set of words, denoted as $s = (w_1^s, w_2^s,..., w_{|s|}^s)$, i.e., \texttt{2024-11-14 192.168.1.1 GET /index.html 200 123ms} is a log.
Logs serve as an essential tool for system monitoring, debugging, and performance analysis~\cite{vldbanomolydetection}. 
By capturing a detailed account of system activities, logs allow developers to trace events, identify issues, and ensure optimal functionality~\cite{vldblog3,wwwlog}. 
Recent research in database and data management focuses on various aspects of log analysis, including log anomaly detection~\cite{vldbanomolydetection,vldblog2,tkdelog13,kddtemplate2}, log-based retrieval~\cite{sigmodlogretrieval10,sigmodlogretrieval12,icdelog18,icdelog17}, and root cause analysis~\cite{vldblogrootcause6,vldblogrootcause8,tkdelog14,tkdelog15}.

In modern systems, logs are produced in massive volumes daily.
Each system generates multiple system events, producing logs in various formats corresponding to different templates.
However, the large size of  logs makes it challenging and time-consuming for  developers and operators to trace events and identify issues efficiently.
Recent studies have explored knowledge base template generation for user queries~\cite{vldbpattern1,sigmodpattern1,sigmodpattern2,sigmodtemplate4}, while several researchers have proposed methods for generating templates for SQL queries~\cite{vldbpattern2,vldbpattern3,icdetemplate3,icdetemplate5}. There are also several works proposed to generate templates for Web pages~\cite{tkdetemplate1,tkdetemplate2,vldbwebtemplate1,tkdewebtemplate2}. 
Motivated by template generation works in database area~\cite{sigmodpattern1,vldbpattern2}, we can generate templates for logs, converting unstructured logs into structured templates, enabling efficient storage, analysis, and debugging~\cite{vldbanomolydetection,vldblog2}.
 Formally, 
 given a log $s = (w_1^s, \dots, w_{|s|}^s)$ and  word type candidates $\mathcal{T}=\{\tau_1,\cdots, \tau_{|\mathcal{T}|}\}$ and system keyword candidates $\mathcal{K}=\{k_1,\cdots, k_{|\mathcal{K}|}\}$, 
 the target of log template generation is to generate a template $t=(w_1^t, \cdots, w_{|t|}^t)$  for $s$, which maps each word $w_i^s \in s$ to a corresponding word type $w_i^t \in \mathcal{T}$ or a system keyword $w_i^t \in  \mathcal{K}$.
 
 For example, the template of the log  \texttt{2024-11-14 192.168.1.1 GET /index.html 200 123ms} is 
 \texttt{{[DATE]} [IP] $<$GET$>$ [RESOURCE] [STATUS] [LATENCY]},
 where $[\cdot]$ is a word type and $<\cdot>$ is a keyword. The template of log  \texttt{2024-11-14 192.168.1.2 DELETE /test.html 200 123ms} is  \texttt{{[DATE]} [IP] $<$DELETE$>$ [RESOURCE] [STATUS] [LATENCY]}.
By identifying structured templates from raw logs, log template generation streamlines the management and analysis of large-scale log data from systems such as Apache~\cite{hadoop,logpai}, HDFS~\cite{hadoop,logpai}, and Spark~\cite{spark,logpai}. 
These templates make it easier to track IP addresses, enabling the detection of anomalies or errors in database systems~\cite{vldbanomolydetection,vldblog2}.

\subsection{Log Template Generation Approaches} \label{ssec:log_models}
Depending on the techniques used for generating log templates, existing works can be classified into three categories: heuristic-based approaches~\cite{drain,swisslog,logmine,shio,logsig}, neural network-based approaches~\cite{nulog,logppt}, and LLM-based approaches~\cite{Divlog}.

\subsubsection{Heuristic-based Approaches}  \label{ssec:heuristic}

Heuristic-based models~\cite{drain,swisslog,logmine,shio,logsig} depend on handcrafted rules or heuristics to group logs into multiple clusters. 
In practice, heuristics cannot match the characteristics of logs from multiple domains well, resulting in low flexibility and adaptability in labelling diverse logs.
For instance, Drain~\cite{drain} assumes that the first word of a log must always be a keyword. However, this assumption is not consistent across logs from different systems. Moreover, logs from specific datasets may follow unique patterns. For example, in the HDFS dataset~\cite{logpai}, the pattern \texttt{blk\_3651} represents a block with ID \texttt{3651}, which can be interpreted using a specific regular expression like $blk\setminus d+$. 
However, such patterns do not exist in other datasets, such as Mac~\cite{logpai} and Android~\cite{logpai}, requiring the re-writing of regular expressions for each dataset, which incurs significant human effort.
As each system have logs in diverse templates, the developers have to check the whole log sets from each system to separately define rules or threshold for template generation, which is impractical in real-world log datasets~\cite{logpai} with hundreds of templates.

\subsubsection{Neural Network-based Approaches}  \label{ssec:nnmodel}

Neural network-based models~\cite{nulog,logppt} train neural network models (i.e. transformers) to predict each word type in logs. 
However, the training process demands large-scaled logs with annotated word type. 
The annotation of word type in large-scaled logs requires experienced software maintainer, which is expensive. 
Besides, the trained template generation model relies on patterns existing in training set with potential low generalization ability to unseen patterns.
For examples, if logs in training set only cover the template \texttt{[ADDRESS] Byte flume reports host available}, the trained model may falsely infer the template of log \texttt{[ADDRESS] Byte flume reports host available again} while an important keyword \texttt{again} is missed due to it is unseen in training set.
One intuitive solution to generalization problem is to re-train or update the neural network model based on the  logs~\cite{vldbtransfer1,vldbtransfer2}. However, 
it is also costly to re-train the transformer model~\cite{logppt,nulog,Divlog}.
As high cost for human-resources and low generalization problem for both heuristic-based approaches and neural network-based approaches, LLM-based template generation methods~\cite{Divlog} have been proposed.

\begin{table}[]
	
	\caption{Important Notations}
	\small
	\color{black}
		\begin{tabular}{l|l}
			\hline
			\textbf{Notation} & \textbf{Definition} \\ \hline
			$s, t$ &  A log $s \in S$ and  its template $t\in T$\\ \hline
			
			$\hat{t}$ &  Predicted template $\hat{t}$ of log $s$, $t \in \hat{T}$       \\ \hline
			
			$w_i^s $ &  The $i$-th word in the log $s$       \\ \hline
			
			$w_i^t $ &  The $i$-th type in template $t$       \\ \hline
			
			$\mathcal{K}$ &   Keyword set   \\ \hline

			$\mathcal{T}$ &  Candidate word type set    \\ \hline

			$G=(S,W)$ & Bipartite graph between log set $S$ and words $W$  \\ \hline
			
			$r$ &  The $r$-th round    \\ \hline

			$L_r$ & The annotated logs at the $r$-th round\\ \hline
			
			$L$ & Annotated logs\\ \hline
			
			$U$ & Unlabeled logs \\ \hline
			
			$cosine(\cdot)$ & {\color{black} Cosine similarity score}\\ \hline
			
			$SED(\cdot)$ & Semantic-based edit distance between two logs\\ \hline
			
			$I_{s_i}$ & Representative score of $s_i$ \\ \hline
			
			$ P(s_i,\hat{t}_i)$ & Average probability of predicted template $\hat{t}_i$\\ \hline 
			
			$\mathbb{I}(s_i,\hat{s}_i)$ & Prediction consistency indicator \\ \hline
			
			$C(s_i,\hat{t}_i, \hat{s}_i)$ & Confidence score \\ \hline

			$B_r$ & Annotation budget at the $r$ -th round \\ \hline
			
			$W_r$ & Identified words at the $r$ -th round \\ \hline

			$IS(\cdot)$ &  The informative score in Equation~\eqref{eq:multiroundannotation}\\ \hline
			
			$\lambda$ &  Trade-off parameter  in Equation~\eqref{eq:multiroundannotation} \\ \hline

			$D_s$ & Demonstrative logs of $s$ \\ \hline
			
			$UW(D_s)$ & Word set in $D_s$\\ \hline

		\end{tabular}
	\label{tab:notation}
\end{table}
\subsubsection{LLM-based Approaches} \label{sec:LLMrelatedwork}
Large language model pre-trained on massive corpora has obtained extraordinary natural language processing ability~\cite{li2024survey}. 
Motivated by success of LLM, researchers have investigated LLM-based log template generation that feed the log and candidate words labels to LLM with instructing LLM to label each word.
However, as logs contain various domain-specific words which is excluded in open-sourced corpora, it is difficult for LLM to understand these words~\cite{Divlog}. 
Based on recent studies, performance of pre-trained LLM can be training-freely augmented by prompting LLM with several examples with ground truth templates as contexts, calling ICL~\cite{Divlog, adaicl,vldbicl1,icl,icl2,icl1,vldbicl2}. 

Researchers have proposed a multi-round annotation framework for LLMs~\cite{adaicl,Ideal, votek, vldbhuman} with an annotation component and an ICL component.
As for annotation, they firstly compute the cosine similarity score between embedding of logs.
Based on the score, annotation process finds a representative subset as labeled logs within a budget $k_a$.
In ICL component, for each unlabeled log, they selects top-$k_c$ similar labeled logs which offer contextual knowledge for LLM to comprehend the correlation between each word and labels.
AdaICL~\cite{adaicl} uses each unlabeled log to represent its top-$k_a$ similar logs and proposes selecting $k_a$ unlabeled logs to represent as many logs as possible.
IDEAL~\cite{Ideal} employs an information diffusion process~\cite{Ideal} to select the top-$k_a$ most informative unlabeled logs.
DivLog~\cite{Divlog} selects the top-$k_a$ most diverse and informative unlabeled logs.
However, the current works achieve sub-optimal performance due to following three issues: 1) The embedding of logs overlooks the importance of keywords but emphasizes several useless words thus cannot identify the most informative contexts. 2) The annotation ignores that logs are unconfident for LLM also worthy for annotation. 3) Fixed top-$k_c$ similar demonstrations misguide LLM by irrelevant/scarced contexts.
 \section{Method} \label{sec:method}
In this section, we introduce our LLM-driven multi-round annotation and adaptive in-context learning framework in Figure~\ref{fig:framework}. 
\subsection{Framework Overview}
\noindent\textbf{Step 1: Log Similarity and Annotation.}  \label{ssec:step1}
This component aims to construct and update the annotated log set $ {L}$ at each round $r$ by selecting the most representative and challenging unlabeled logs from the dataset ${U}$. At the $r$-th round, we compute an {edit-distance-based metric} $\text{SED}$ between logs, which emphasizes important keywords to identify representative logs.
Based on SED, the framework selects a subset ${L}_r$ of size $B_r$ for human annotation by jointly considering two factors: representativeness and LLM prediction confidence.
More details can refer to Section~\ref{ssec:logdistance} and Section~\ref{ssec:multi_round}.

\noindent\textbf{Step 2. Adaptive Demonstration Selection.}
This component dynamically selects a {minimum number of labeled logs} from the labeled log set ${L}$ to serve as context for each unlabeled log $s_i \in {U}$ during LLM inference. This dynamic selection ensures that all words in the input log $s_i$ are covered while avoiding irrelevant or redundant context, which could otherwise degrade the quality of template generation.
More details can refer to Section~\ref{ssec:adaptive}.

\noindent\textbf{Step 3. LLM-Driven Template Generation.}
After constructing the adaptive context for each unlabeled log $s_i$, we generate the final template prediction. This process begins with {prompt construction}, which includes three key elements: an instruction for template generation, examples of labeled logs with their templates (retrieved from the context), and the input log $s_i$ with its identified words. 
Once the prompt is constructed, it is fed into the LLM for inference, resulting in the predicted template $\hat{t}_i$ for the input log $s_i$.

\begin{figure*}[t]
	\centering
	\includegraphics[width=\linewidth]{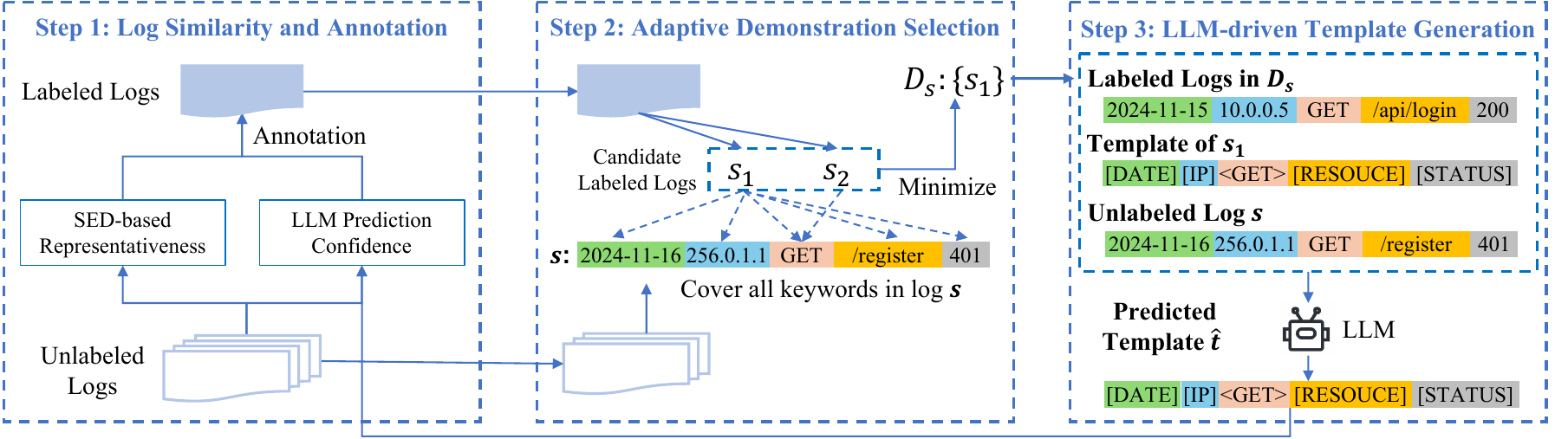}
	\vspace{-2em}
	\caption{ 
		Framework overview of
		LLMLog consists of three key components:
		(1) Log similarity and annotation:
		LLMLog employs a semantic-based edit-distance (SED) metric to assess log similarity.
		{\color{black} It selects a subset of logs for human annotation by identifying the most representative and challenging ones through a greedy algorithm across multiple rounds.}
		(2) Adaptive demonstration selection:
		LLMLog adaptively selects a minimal set of labeled logs that comprehensively cover all relevant words for each input log. This ensures that the contextual information provided to the LLM is both relevant and concise.
		(3) LLM-driven template generation:
		LLMLog constructs tailored prompts by incorporating the adaptive contexts and unlabeled logs. These prompts are then input into the LLM, which generates accurate templates for the unlabeled logs.}
	\label{fig:framework}
\end{figure*}

\subsection{Log Similarity} \label{ssec:logdistance}
As mentioned in Section~\ref{sec:LLMrelatedwork}, existing works compute the similarity~\cite{similarity} between two logs $s_i$ and $s_j$ using the cosine similarity of their embeddings~\cite{embedding1,embedding2,embedding3}, which can be defined as 
\begin{align}\label{eq:cosine}
sim(s_i,s_j)=cosine(\mathbf{s}_i, \mathbf{s}_j),
\end{align}
where  $\mathbf{s}_i = \frac{\sum_{k=1}^{|s_i|} \mathbf{w}^s_k}{|s_i|}$ is the embedding of log $s_i$, and $\mathbf{w}^s_k$ represents the embedding of the $k$-th word in log $s_i$.
However, the cosine similarity for log embedding is not suitable for representing log similarity in template generation, as it overlooks the importance of several keywords and favors irrelevant words with a large number of letters, such as timestamps or IP addresses.
We aim to find similar template logs as contexts for the target log in the ICL phase, where each contextual log should provide enough information to guide the LLM on how to process each word.

Suppose we have an unlabeled log \texttt{com.cse.ust.hk:8080 POST}, and two labeled logs, \texttt{proxy.cse.cuhk.edu.hk:5070 POST} and \texttt{com.cse.ust.hk:8080 GET}.
Intuitively, we should pick the first one to demonstrate how LLM converts the IP address to \texttt{[IP]} and tags \texttt{POST} with \texttt{$<$POST$>$}.
However, the cosine similarity in Equation~\eqref{eq:cosine}  emphasizes the longer words, such as \texttt{com.cse.ust.hk:8080}.
Therefore, the cosine similarity metric selects \texttt{com.cse.ust.hk:8080 GET} as the context, which has no information about the keyword \texttt{POST}.
Such a context may falsely guide the LLM to generate templates with an irrelevant word, \texttt{GET}.
To emphasize the important words, we can build a bipartite graph to model the similarity via connections between logs and the words they contain.

\begin{myDef}[Log-Word Bipartite Graph]
	\color{black}
	We build a bipartite graph $G=(S,W)$ where $S=\{s_i\}_{i=1}^{|S|}$ is the set of logs and $W =\cup_{s_i \in S}\cup_{w \in s_i}{w}$ is the set of words contained in $S$. 
	There is an edge connecting log $s_i$ and word $w$ if $w \in s_i$.
\end{myDef}

In the log-word bipartite graph, multiple logs can be linked through their shared words. 
This implies that a log $s_i$ can assist an LLM in understanding a similar log $s_j$ if they share similar or identical words. Consequently, the LLM can generate the correct template for the words in log $s_j$.
Here, we propose a novel semantic-based edit-distance in real sequence (SED) to determine the similarity among two logs, instead of using sentence embedding-based cosine similarity in Equation~\eqref{eq:cosine}. 

The core idea of our proposed metric $SED(s_i, s_j )$, is to compute the minimal number of operations (including deletion, insertion, and replacement) required to transform log $s_i$ into log $s_j$, while incorporating word semantics into the calculation.
Formally, given two logs $s_i = (w^{s_i}_1, \cdots, w^{s_i}_{|s_i|})$ and $s_j = (w^{s_j}_1, \cdots, w^{s_j}_{|s_j|})$ and all the labeled logs $L=\{(s_k,t_k)\}_{k=1}^{|L|}$, the semantic-based EDR similarity score $SED(s_i, s_j )$ between $s_i$ and $s_j$ is defined as follows:

\begin{align} \label{eq:SED}
		SED(s_i, s_j) = \begin{cases}
			|s_i^r|, if |s_j^r| = 0 \\ 
			|s_j^r|, if |s_i^r| = 0 \\ 
			min \begin{cases}
				SED(R(s_i^r), R(s_j^r))+c(s^r_i[0],s^r_j[0])\\ 
				SED(R(s_i^r), s_j^r)+1 \\ 
				SED(s_i^r, R(s_j^r))+1
			\end{cases}
		\end{cases}
\end{align}
\begin{align} \label{eq:cost}
	c(w_1,w_2) =  \begin{cases}
		1, \quad cosine(\mathbf{w}_{1}, \mathbf{w}_{2}) \geq 0 \\
		0, \quad cosine(\mathbf{w}_{1}, \mathbf{w}_{2})  < 0
			\end{cases}
\end{align}

where $s_i^r = \{w_i \mid w_i \in s_i \land w_i \notin s_k, \forall s_k \in L\}$ represents the remaining words in $s_i$ that are not identified in the labeled logs.
Also,  $R(s_i^r)$ represents the sub-sequence of words in $s_i^r$ with the current first word removed, Equation~\eqref{eq:SED} recursively enumerates all words in $s_i$ and $s_j$.
More specifically,
as SED is designed to measure the minimal distance between two logs, we use $min(\cdot)$ instead of $avg(\cdot)$ or $max(\cdot)$ to compute the minimal operation of replace/insertion/deletion. Instead of only computing the minimal operation for the first word in $s_i$ and $s_j$, the minimum distance for the two sequences is obtained by recursively iterating $R(s_i^r)$ and $R(s_j^r)$ in a dynamic programming manner.
The function $c(w_1,w_2)$ computes the cosine similarity between the two words based on cosine similarity and $\mathbf{w}_1$ is the embedding of word $w_1$~\cite{llm2vec}.
We define $c(w_1,w_2) = 1$ if the word similarity between $w_1$  and is greater than 0 otherwise the value $c(w_1,w_2) = 0$. 

Compared to cosine similarity in Equation~\eqref{eq:cosine}, SED is better at identifying logs with the same templates by emphasizing keywords.
By definition in Section~\ref{ssec:templategeneration}, logs under the same template share keywords. Thus, we can use edit distance to capture the common parts. The effect of other words can be further reduced by using an adaptive word set.
For example, consider the previous case where cosine similarity produces a false positive, identifying \texttt{com.cse.ust.hk:8080 GET} as a match for the unlabeled log \texttt{com.cse.ust.hk:8080 POST}, while ignoring the true positive log \texttt{proxy.cse.cuhk.edu.hk:5070 POST}.
As the adaptive word set in the SED metric removes irrelevant IP addresses as shown in labeled logs, the edit distance between input log and \texttt{proxy.cse.cuhk.edu.hk:5070 POST} becomes $0$, which is smaller than the distance between it and \texttt{com.cse.ust.hk:8080 GET}. Therefore, SED is more suitable for measuring the similarity between logs in the similar templates.
In particular, SED can be computed in  $O(|s_i|*|s_j|)$, where $|s_i|$ denotes the number of words in $s_i$.

\subsection{Multiple Round Log Annotation}\label{ssec:multi_round}
In this subsection, we introduce our LLM-driven multi-round annotation approach. 
Specifically, our framework for multi-round labeled log annotation follows a common setup where the total number of rounds is $n$, 
and each round $r$ is allocated a budget of $B_r$. In the $r$-th round, 
given the unlabeled logs  $U = \{s_i\}_{i=1}^{|U|}$,
the objective is to annotate up to $B_r$ logs,
 which are then added to the labeled log set $L = \{(s_j, t_j)\}_{j=1}^{|L|}$. 
This labeled log set $L$  is subsequently utilized as demonstration candidates for unlabeled logs $U\setminus L$ during LLM inference, improving its ability to generate accurate templates.

To achieve effective annotation, we propose two complementary metrics to guide the selection of the most representative and challenging unlabeled logs for annotation in each round. 
The first metric is the representative score, which evaluates how representative an unlabeled log is in relation to other unlabeled logs. 
The second metric is LLM confidence, which measures the LLM's confidence in generating a correct template for each unlabeled log.
By combining these two metrics, our approach ensures that the most representative and challenging logs are strategically selected for annotation in each round. The details of these two metrics are discussed in the following sections.

\subsubsection{Representative Score} \label{ssec:represenativescore}
We firstly present the metric of the reprensentativeness of labeled logs. 
As introduced in Section~\ref{ssec:logdistance}, 
similar words tend to have the same types, 
and logs with low SED scores are likely to share similar templates. 
Therefore, we aim to annotate logs that share similar words with many other logs, maximizing the representativeness of the labeled set. This approach ensures that the labeled logs capture a diverse and meaningful range of patterns present in the data.
As a result, we can select representative logs for labeling because they provide more contextual information for other logs during in-context learning (ICL).
 Formally, given two logs $s_i$ and $s_j$, we say that $s_j$ is represented by $s_i$ if $SED(s_i, s_j) \leq \delta*min(len(s_i), len(s_j))$.
We define the representative set for a log $s_i$ as:
\begin{align}\label{eq:repre_set}
I_{s_i} = \{s_j |SED(s_i, s_j) \leq \delta*min(len(s_i), len(s_j)) \}, \delta \in (0,1]
\end{align}
The threshold is set to $\delta*min(len(s_i), len(s_j))$, implying that the represented log $s_j$ must have at least one common word to $s_i$. 
In general, the size of the set $I_{s_i}$ reflects how representative $s_i$ is among all unlabeled logs.
$\delta$ is hyper-parameter to control the size of representative group. 
Intuitively, $\delta$ should not be too large as it will weaken the representativeness of each log, thus provide less informative contexts for unlabeled logs.
We give a parameter sensitivity experiment in Section~\ref{ssec:para} to investigate the effects of $\delta$.

\subsubsection{LLM Prediction Confidence Score} \label{ssec:LLMconfidence}

We introduce the confidence score of each unlabeled log based on the prediction of LLMs. 
Intuitively, if the LLM exhibits low confidence in generating an accurate template for a log, 
that log is prioritized for annotation. 
By focusing on these low-confidence logs, the annotation process targets the most challenging cases, 
thereby improving the LLM's overall performance on similar logs, which is a common practice in LLM annotation works~\cite{adaicl,votek,Ideal}.

In general, given an unlabeled log $s_i=(w^{s_i}_1,\dots, w^{s_i}_{|s_i|})$,
 the generated log template by a LLM can be denoted as $\hat{t}_{i}=\{\hat{w}^{t_i}_1,\cdots,\hat{w}^{t_i}_{|t_i|}\}$.
Firstly, we can use the average predicted word-probability to estimate the confidence by LLM, defined as follows.
\begin{align} \label{eq:tokenprob}
	 P(s_i,\hat{t}_i) = \frac{1}{|\hat{t}_i|}*\sum_{k=1}^{|\hat{t}_i|} p({w}^{\hat{t}_i}_k)
\end{align}
where $p({w}^{\hat{t}_i}_k)$ is the predicted probability of LLM for each word in the predicted template $\hat{t}_i$, i.e, ${w}^{\hat{t}_i}_k\in \hat{t}_i$.
If the average probability is low, it indicates the LLM's low confidence in its predictions, suggesting that the log is challenging to interpret and needs to be prioritized for annotation.
Suppose there is a log \texttt{com.cse.ust.hk:8080 DELETE} where the keyword \texttt{DELETE} does not exist in labeled set. LLM is less confident for \texttt{DELETE} than that other identified keywords, like \texttt{POST}. Therefore, the predicted probability of the log tends to be less than other logs. The average probability metric selects this log for annotation. 

However, a high predicted word-probability $ P(s_i,\hat{t}_i)$ cannot guarantee that the generated word type $w^{\hat{t}_i}_j \in \hat{t}_i$ really corresponds to $w_j^{s_i} \in s_i$.
The reason is that LLMs generate output words in a regressive manner, 
which cannot guarantee that the generated template word $w^{\hat{t}_i}_j \in \hat{t}_i$ corresponds accurately to $w_j^{s_i} \in s_i$. 
Additionally, current researchers~\cite{zhao2024chat2data, li2024llm,liu2024enhancing} have demonstrated that 
LLMs suffer from the hallucination problem, where the model may misunderstand the context and generate irrelevant or incorrect information.
For example, the correct template of a log \texttt{com.cse.ust.hk:8080 POST} is \texttt{[IP] $<$POST$>$}. 
However, the LLM might predict the template of this log as \texttt{[IP] [STATUS]} 
with high confidence, 
as it may mistakenly interpret \texttt{[IP] $<$POST$>$} as \texttt{[IP] [STATUS]} due to the hallucination problem.
Thus, the LLM produces an incorrect result even exhibiting high confidence in the predicted template.

Therefore, except for average predicted word-probability metric, inspired by the template generation requirements in Section~\ref{ssec:templategeneration}, we can derive a word-consistency metric for template generation task.
Specifically, in addition to the predicted template $\hat{t}_i$ of $s_i$
generated by the LLM, 
we enable the LLM to generate the corresponding words $\hat{s}_i = (w^{\hat{s}_i}_1,\cdots, w^{\hat{s}_i}_{|\hat{s}_i|})$  as well. 
We then compare the consistency between the words in the input log $s_i= (w^{{s}_i}_1,\cdots, w^{{s}_i}_{|{s}_i|})$ and  the generated words $\hat{s}_i = (w^{\hat{s}_i}_1,\cdots, w^{\hat{s}_i}_{|\hat{s}_i|})$, with consistency indicator $	\mathbb{I}(s_i,\hat{s}_i) $ defined as follows.

 \begin{align} \label{eq:consistency}
 	\mathbb{I}(s_i,\hat{s}_i) = \begin{cases}
 		0, {s}_i= \hat{s}_i \\ 
 		1, {s}_i \ne \hat{s}_i \\ 
 	\end{cases}
 \end{align}

The word-consistency indicator function evaluates whether the words in the log $\hat{s}_i$ generated by the LLM are consistent with the original words in the input log $s_i$. Specifically, it ensures that no new words or non-candidate labels are falsely generated by the LLM. Additionally, it verifies whether the word count $|\hat{s}_i|$ matches the word count $|s_i|$, ensuring that each word is correctly labeled.
Any templates that fail to meet these two conditions are considered incorrect and can be used to identify error cases.
To assess the LLM's performance in the log template generation task, we combine the average word-probability and word-consistency metrics.
Formally, given a log $s_i$ and the predicted log template $\hat{t_i}$ associated with the predicted words $\hat{s}_i$, we define the confidence score as $C(s_i,\hat{t}_i, \hat{s}_i)$ as follows.
 into a weighted sum as follows.
  \begin{align} \label{eq:LLMhardness}
 	C(s_i,\hat{t}_i, \hat{s}_i) = a*(1-P(s_i,\hat{t}_i))+ (1-a)*\mathbb{I}(s_i,\hat{s}_i)
 \end{align}
where  $a$ is a weight to balance the average probability  $P(s_i,\hat{t}_i)$ and consistency score $\mathbb{I}(s_i,\hat{s}_i)$, where we use $1-P(s_i,\hat{t}_i)$ to select logs with  low confidence in the predicted template $\hat{t}_i$.
A larger $C(s_i, \hat{t}_i, \hat{s}_i)$ indicates that the LLM has low confidence and potential inconsistency between $\hat{s}_i$ and $s_i$, suggesting that $s_i$ is challenging.

\subsubsection{LLM-driven Log Annotation Problem} \label{ssec:LLMLogannotation}
We introduce our LLM-driven log annotation problem by incorporating the proposed representative score in Section~\ref{ssec:represenativescore} and the LLM prediction confidence score in Section~\ref{ssec:LLMconfidence}. The formal definition is given as follows.

\begin{definition}[LLM-driven Log Annotation Problem] \label{def:multiroundannotation}
Given a budget $B_r$ and an unlabeled log set $U$ at round $r$, where the LLM predicts a template $\hat{t}_i$ along with the corresponding log $\hat{s}_i$ for each unlabeled log $s_i \in U$, the objective is to select a subset of informative unlabeled logs $L_r \subseteq U$ for annotation, such that $|L_r| \leq B_r$. 
The selection aims to maximize the following objective:
	\begin{align} \label{eq:multiroundannotation}
	& 	IS(L_r)=\max  \sum_{s_i \in L_r}  (1-\lambda) \frac{|\bigcup_{i=0}^{|L_r|} I_{s_i} |}{|U|} + \lambda  C(s_i,\hat{t}_i, \hat{s}_i)  \\
		& s.t. \quad |L_r|\leq B_r, L_r \subseteq U
	\end{align}
where $\lambda$ is a trade-off parameter, $ I_{s_i}$ is the unlabeled logs represented by $s_i$ defined in Equation~\eqref{eq:repre_set}, and $ C(s_i,\hat{t}_i, \hat{s}_i) $ is the confidence score defined in Equation~\eqref{eq:LLMhardness}.
	
\end{definition}
As shown in Equation~\eqref{eq:multiroundannotation}, in each round $r$, the first term, $\frac{|\bigcup_{i=0}^{|L_r|} I_{s_i}|}{|U|}$, prioritizes selecting representative logs that are likely to impact a larger number of unlabeled logs and guide the LLM to process more logs effectively. 
The second term, $C(s_i, \hat{t}_i, \hat{s}_i)$, ensures the selection of logs with the lowest LLM prediction confidence.
By combining these two scores, we can effectively select $B_r$ informative logs, $L_r \subseteq U$, for annotation.

\begin{theorem} \label{theo:multiroundannotation}
The LLM-driven log annotation problem is NP-hard.
\end{theorem}
	
\begin{proof}
	
	We prove Theorem~\ref{theo:adaptivedemon} by reduction from the Max Coverage problem~\cite{maxcover} to our problem. 
	Due to space limits, we put full proof in technique report~\cite{techniquereport} Appendix 6.1. 
\end{proof}

\subsubsection{Algorithm for Annotation Log Selection} \label{ssec:annotationroundalgorithm}
As demonstrated in Theorem~\ref{theo:multiroundannotation}, the LLM-driven log annotation problem under the budget $B_r$ is \textit{NP-hard}, 
indicating that it is unlikely to be solved optimally in polynomial time. 
To address this, we propose a greedy algorithm with a theoretical guarantee.
The basic idea is to greedily select the unlabeled log that can bring the maximum information gain into the selected annotation set until exceeding the log budget $B_r$. 
Specifically, given the unlabeled  set ${U}$, we first define the marginal information gain of $s$ for the selected set $L_r$ as follows:
\begin{equation}\label{eq:marginal_score}
	\triangle IS(s|L_r) = IS( L_r \cup \{s\})-IS(L_r)
\end{equation}
The details are provided in Algorithm~\ref{alg:annotationselection}. Specifically, we first initialize the selected log set ${L}_r$ as $\emptyset$ (line 1). 
Then, for each unlabeled log $s \in U$, we update $SED(s,s_i)$ among the other unlabeled logs $s_i \in U \setminus s$ (line 3). 
Also, we obtain the representative log set $I_s$ for each unlabeled log $s$ (line 4) and compute the confidence score $ C(s,\hat{t}, \hat{s})  $ based on the LLM model $f_\theta$ (line 5).
Next, we compute the informative score $IS(L_r \cup s)$ by incorporating each unlabeled log $s \in U$ into the selected log set $L_r$. 
We then select the log $s^*$ with the maximum $\triangle IS(s|L_r)$, as defined in Equation~\eqref{eq:marginal_score} (lines 7–10). Afterward, we add $s^*$ to ${S}_r$ and remove it from the set of unlabeled logs $U$ (lines 11–12). 
This selection procedure is repeated until $B_r$ logs have been selected (lines 6–13).

\noindent\textbf{Time complexity} Loop in line 2 enumerates $U$ to compute the representative score and LLM confidence score. For each $s \in U$, let $s_{max}$ be the log with maximum number of words, it takes at most $O(|s_{max}|^2)$ time to compute SED while scores in line 4 and line 5 can be computed in constant time, where the loop costs $O(|U||s_{max}|^2)$. As for while loop starting in line 6, it requires to enumerate each instance in $U$ in inner for loop at line 7 to select one log with maximized $\triangle IS(s|L_r)$, which roughly takes time complexity $O(B_r*|U|)$. Thus, the overall time complexity is $O((B_r+|s_{max}|^2)*|U|)$.
\begin{theorem} \label{theo:annotationselection}
	Algorithm~\ref{alg:annotationselection} has an approximation ratio of $1-\frac{1}{e}$. 
\end{theorem}

\begin{proof}
	Let $IS(L_r^*)$ denotes the optimal value of objective in Equation~\eqref{eq:multiroundannotation} within budget $B_r$. 
	We first prove the  $\triangle IS(s|D_s)$ in Equation~\eqref{eq:marginal_score} is monotone increasing and submodular.
	Then we prove $(1-(\frac{1}{B_r})^{B_r})IS(L_r^*) \leq IS(L_r)$.
	Due to space limits, we put the full proof in technique report~\cite{techniquereport} Appendix 6.2.
\end{proof}

\begin{algorithm}[t]
	\color{black}
	\KwIn{Annotation budget $B_r$, unlabelled logs $U$, previous LLM prediction $\hat{T}_{r-1}^U$ and LLM $f_\theta$ }
	\KwOut{Selected logs ${L}_r$ for annotation}
	
	$L_r \leftarrow \emptyset$ \\
	\For{$s \in U$}{
		$\forall s_i \in U\setminus s$, $SED(s,s_i) \gets$Equation~\eqref{eq:SED}\\
		$I_{s}\gets$ Equation~\eqref{eq:repre_set}\\
		$ C(s,\hat{t}, \hat{s})  \gets$ Equation~\eqref{eq:LLMhardness}\\
	}	
	\While {$|L_r| < B_r$} {
		\For{$s \in U$}{
			$IS(L_r \cup \{s\}) \gets$Equation~\eqref{eq:multiroundannotation}\\
			$	\triangle IS(s|L_r) = IS( L_r \cup \{s\})-IS(L_r)$
		}
		$s^* = argmax_{s \in U} \triangle IS(s|L_r)$ \\
		$L_r = L_r \cup s^*$ \\
		$U = U\setminus  s^*$ \\
	}
	\textbf{return}  $L_r$
	\caption{Annotation selection at the $r$-th round}
	
	\label{alg:annotationselection}
\end{algorithm}
\setlength{\textfloatsep}{0.2cm}
\setlength{\floatsep}{0.2cm}
\subsubsection{Algorithm for Multiple Round Log Annotation} \label{ssec:annotationalgorithm}
 
In the multi-round framework, annotation selection can be performed iteratively using Algorithm~\ref{alg:annotationselection} under a total budget $B$.
The overall performance of LLM will converge once $B$ is large enough to cover all the words. 
We give a parameter sensitivity experiment in Section~\ref{ssec:para} to investigate the effect of $B$.
To execute Algorithm~\ref{alg:annotationselection}, a strategy is required to determine the budget $B_r$ for each individual round.
However, since the LLM operates as a black box, it is not possible to predict the performance improvement of the LLM as the labeled log set is augmented.
Alternatively, we design an adaptive strategy base on the number of identified words.
\begin{align} \label{eq:avgbudget}
	B_r = B_{r-1} * (1-\frac{\triangle W_{r-1}}{W_{r-1}}) \\
	\triangle W_{r-1} = W_{r-1} - W_{r-2}
\end{align}
Instead of manually setting a hyper-parameter for the total number of annotation rounds, Equation~\eqref{eq:avgbudget} adaptively computes $B_r$ based on the previous round budget $B_{r-1}$.
Specifically, in the $r-1$-th round, the number of identified words is denoted as $W_{r-1}$. $\triangle W_{r-1}$ denotes the increment of words in the $r-1$-th round.
By this definition, if the number of words in the $r-1$-th round increases compared to its previous round, LLMLog will annotate fewer logs in the current round.
On the contrary, if the number of words in the $r-1$-th round decreases, LLMLog will annotate more logs in the current round.
Since the adaptive budget strategy requires annotation results from the previous two rounds, we intuitively set the annotation budget for the first two rounds.

As illustrated in Algorithm~\ref{alg:multiroundannotation},
in line 1, we initialize the annotated labeled log set using the Determinantal Point Process (DPP)\cite{dpp, Divlog}, following the approach in Divlog\cite{Divlog} to select the most diverse logs.
After the initialization process, we iteratively perform the following steps:
First, in lines 5 to 7, we execute Algorithm~\ref{alg:adaptivedemon} to retrieve the demonstration set $D_s$ for each $s \in U$.
Then, we feed $D_s$ and $s$ into $f_\theta$ to obtain the predicted template $\hat{t}$.
After obtaining the prediction result, we determine the annotation budget $B_r$ for the next round based on Equation~\ref{eq:avgbudget} (line 10-11).
Once the annotation budget $B_r$ is determined, we proceed to a new round of annotation at line 12.
The labeled log set $L$ and unlabeled log set $U$ are updated after the annotation process (lines 13-14).

\subsection{Adaptive Demonstration Selection}\label{ssec:adaptive}
After annotating the selected logs $L_r = \{(s_i, t_i)\}_{i=1}^{B_r}$ in the $r$-th round, we obtain all labeled logs from the first round to the $r$-th round as $L = \cup_{i=1}^r L_i$. 
For each unlabeled log $s \in U$, 
we select a set of demonstrations $D_s \subseteq L$ to provide contextual information for $s$. 
These demonstrations $D_s$ will help the LLM understand the words in each unlabeled log $s$ and generate the correct log template $t$.
To improve efficiency and prevent irrelevant information, it is important to limit the size of the demonstration set $D_s$. 
As mentioned in Section~\ref{ssec:log_models}, existing works~\cite{Divlog,adaicl,votek,Ideal} commonly define a fixed number $k$ and select $k$ demonstration logs for each $s$. 
However, a fixed number $k$ is not ideal for every $s$, as the number of words and the difficulty of each unlabeled log can vary significantly.

In this paper, we propose an adaptive demonstration selection approach that does not rely on a fixed number $k$. 
Instead, the size of the demonstration set is dynamically adjusted based on the characteristics of each unlabeled log $s$. Formally, we define the adaptive log demonstration selection problem as follows.

\begin{algorithm}[t]
	\color{black}
	\KwIn{Annotation budget $B$, LLM $f_\theta$, unlabelled logs $U$, total rounds $n$}
	\KwOut{Annotated logs $L$}
	
	Initialize $L_1$, $\mathcal{K}$, $\mathcal{V}$ by DPP \\
	
	$r \leftarrow 2$ \\
	\While {$B \geq 0$ } {
		$\hat{T}_{r}^U \gets \emptyset$\\
		\For{$s\in U$}{$D_s \leftarrow AdaptiveDemonSelection(s,L)$\\
			$\hat{t}  \leftarrow f_\theta(s, D_s)$\\
			$\hat{T}_{r}^U \leftarrow \hat{T}_{r}^U \cup \hat{t}_{s}$\\}
		$r \leftarrow r + 1$ \\
		$B_r\gets$ Equation~\eqref{eq:avgbudget}\\
		$B_r = min(B_r,  B)$, $B = B - B_r$ \\ \label{eq:totalbudget}
		$L_r \leftarrow AnnotationSelection(B_r, U,\hat{T}_{r}^U, f_\theta)$ \\ 
		$L \leftarrow L \cup L_r$ \\
		$U \leftarrow U \setminus L_r$ \\
	}
	\textbf{return}  $L$
	\caption{Multiple Round Log Annotation}
	
	\label{alg:multiroundannotation}
\end{algorithm}
\begin{myPro}[Adaptive Log Demonstration
	Selection Problem]  \label{def:adaptivedemon}
	Given an input log $s = (w_1^s, \cdots, w_{|s|}^s)$ and the annotated log set $L = \{(s_i, t_i)\}_{i=1}^{|L|}$, 
	the goal is to select a demonstration set $D_s \subseteq L$ from the annotated log set $L$ by minimizing the following objective.
	 \begin{align} \label{eq:adaptivedemon}
		&\quad \qquad\min |D_s| \\
		s.t. &\forall\  w^s_i \in s, \exists\  w^{s_j}_k \in s_j  \land s_j \in D_s, \nonumber
		\\
		& \text{such that } cosine(\mathbf{w}^s_i, \mathbf{w}^{s_j}_k) \geq 0 \label{eq:adaptiveconstraint}
	\end{align}
	 where  $\mathbf{w}^s_i \in \mathbb{R}^{d_w}$ is the word embedding of the word ${w}^s_i$ and $cosine(\cdot)$ is the cosine similarity measurement.
\end{myPro}

\begin{theorem} \label{theo:adaptivedemon}
	The adaptive log demonstration selection is NP-hard.
\end{theorem}

\begin{proof}
We prove Theorem~\ref{theo:adaptivedemon} by reducing our problem to the Set-cover problem~\cite{setcover}. Due to space limits, we provide the full proof in the technical report~\cite{techniquereport} Appendix 6.3.
	\end{proof}

\subsubsection{Adaptive Context Selection}
As demonstrated in Theorem~\ref{theo:adaptivedemon}, the adaptive log demonstration
selection problem is \textit{NP-hard}.
To address this, we propose a greedy algorithm with a theoretical guarantee.
As shown in Algorithm~\ref{alg:adaptivedemon}, 
the \textit{Adaptive Demonstration Selection Algorithm} aims to select a subset of labeled logs $D_s \subseteq L$ from a pool of labeled logs $L$ that are most relevant to a given unlabeled log $s$. 
The algorithm begins by initializing an empty set of selected logs $D_s$ and an empty set of union words  $UW(D_s)$ (line 1-2). 
It iteratively selects the most relevant labeled log from $L$ to include in $D_s$ until no additional contribution from labeled logs $L$.
\begin{algorithm}[t]
	\KwIn{Unlabeled log $s$ and all labeled logs $L$.}
	\KwOut{Selected demonstrated logs $D_s \subseteq L$ for $s$}
	
	$D_s = \emptyset$ \\
	
	$UW(D_s) = \emptyset$ \\
	
	\While {\textbf{True}} {
		
		\For{$s_i \in L$}{
			$UW(s_i|s) =\emptyset$\\
			\For{$w^s_j \in s$}{
				
				$w^{{s_i}*}_k = \arg\max_{w^{s_i}_k \in s_i}{cosine(\mathbf{w}^{s_i}_k, \mathbf{w}^s_j)}$\\
				\If{$cosine( \mathbf{w}^{{s_i}*}_k,\mathbf{w}^s_j)\ge 0$}{
					$UW(s_i|s) = UW(s_i|s) \cup \{w^{{s_i}*}_k \}$\\
				}
			}
			$UW(D_s \cup s_i) = UW(D_s) \cup UW(s_i|s)$ \\
		$\triangle UW(s_i|D_s)=UW(D_s \cup s_i) -UW(D_s) $ \\
		\If{$|\triangle UW(s_i|D_s)|$=0}{
			\textbf{break}
		}
	}
	\If{$\forall s_i \in L$, $|\triangle UW(s_i|D_s)|$=0}{
		\textbf{break}
	}
	$s_i^*=argmax_{s_i \in L}\triangle UW(s_i|D_s)$\\
	$D_s \leftarrow D_s \cup s_i^*$ \\
	$L \leftarrow L\setminus s_i^*$ \\
}
\textbf{return}  $D_s$
\caption{Adaptive Demonstration Selection}

\label{alg:adaptivedemon}
\end{algorithm}

Specifically, at each iteration, the algorithm evaluates every labeled log $s_i \in L$.
 For each word $w^s_j$ in the unlabeled log $s$, it identifies the most similar token $w^{s_i*}_k$ in the labeled log $s_i \in L$ using cosine similarity (line 6-7). 
 Words are considered similar if their cosine similarity score is greater than or equal to 0 (line 8).
 If the word in $s_i$ meets this threshold, it is added to the set of matched word set $UW(s_i|s)$ (line 9). 
 After processing all words in $s$, the algorithm merges current similar words $UW(D_s)$ in selected logs $D_s$ with the  words in $UW(s_i|s)$ (line 10). 
 The contribution of $s_i$ to $s$ regarding the selected log $D_s$ as $\triangle UW(s_i|D_s) = UW(D_s\cup s_i) -UW(D_s)$ (line 11).
If the larges contribution among logs $s_i \in L$ (i.e., $|\triangle UW(s_i|D_s)|$ is zero), the algorithm stops processing that log, as it adds no new information (line 12-13). 
If all remaining labeled logs $s_i \in L$  contributes no information to $s$, the algorithm terminates (line 14-15). 
Otherwise, we add $s_i^*$ to the selected set $D_s$, and remove it from the pool of labeled logs $L$ (line 16-18). 

\begin{theorem} \label{theo:adaptivedemonselection}
	Algorithm~\ref{alg:adaptivedemon} has an approximation ratio of $1+ ln(n)$. 
\end{theorem}
	\begin{proof}
We prove that $\triangle UW(s_i|D_s)$ is monotone increasing and submodular. Then, according to \cite{setcoverproof},the approximation ratio is $1 + \ln(n)$. We provide the full proof in \cite{techniquereport} Appendix 6.4.
\end{proof}

\section{Experiments}\label{sec:experiments}

\subsection{Experiment Setting}\label{ssec:exp_setting}

\subsubsection{Datasets} 
We use the widely-used log template benchmark over 16 domains provided by Log-PAI~\cite{logpai} with their statistics summarized in Table~\ref{tab:stat}.
In each domain, there are 2,000 logs labeled with ground-truth templates and a unique ID~\cite{logpai, Divlog}.

\subsubsection{Metrics} 
We use three metrics to evaluate the effectiveness of template generation from logs, message level accuracy (MLA), precision template accuracy (PTA) and recall template accuracy (RTA). 
	MLA is to measure the effectiveness of template generation in log level while PTA and RTA evaluate it at template levels. 

\begin{itemize}[leftmargin=*]

\item \textbf{Message Level Accuracy (MLA).} 
\textbf{MLA} is defined as the ratio of logs whose templates are \textit{correctly generated} to the total number of logs~\cite{metric}. Formally,
given the unlabeled logs $S$, \textbf{MLA} is defined as $\frac{\sum_{s_i \in S} \mathbb{I}(\hat{t}_i= t_i)}{|S|}$, where $ \mathbb{I}(  \hat{t}_i= t_i)=1$ if the predicted template $\hat{t}_i$ of each unlabeled log $s_i \in S$ is the same as its ground truth $t_i$.

\item \textbf{Precision Template Accuracy (PTA).} 
\textbf{PTA} is the ratio of \textit{correctly generated} templates to all generated templates, where \textit{correctly generated} refers to the template whose corresponding logs are all correctly predicted. 
Formally, given the generated templates $\hat{T}$, 
 \textbf{PTA} is $ \frac{\sum_{\hat{t} \in \hat{T}}  {f}\left(\bigwedge_{s_i \in \text{logs}(\hat{t})} \mathbb{I}(\hat{t}_i=\hat{t})\right)}{|\hat{T}|},
$
where $ {f}(\cdot)=1$
when the  predicted template $\hat{t}_i$ of
each log $s_i \in  \text{logs}(\hat{t})$ is  correctly predicted as $\hat{t}$.

\item \textbf{Recall Template Accuracy (RTA).}
 \textbf{RTA} is  the ratio of ground truth templates for which all corresponding logs are correctly predicted to the total number of ground truth templates.
  Formally, given ground truth templates $T$, \textbf{RTA} is $   \frac{\sum_{t \in T} {f}\left(\bigwedge_{s_i \in \text{logs}(t)} \mathbb{I}(\hat{t}_i = t)\right)}{|T|},$
where $ {f}(\cdot)=1$
when the  predicted template $\hat{t}_i$ of
each log $s_i \in  \text{logs}({t})$ is  correctly predicted as ${t}$.
	 
\end{itemize}
Accuracy in template level is tighter than MLA as they require all corresponding logs are correctly generated, which are suitable to evaluate the effectiveness for large-scaled system logs~\cite{Divlog}. 

\subsubsection{Baselines} 
We select Drain~\cite{drain} and LogPPT~\cite{logppt} as representative for heuristic-based methods and NN-based methods respectively. 
We also include the LLM-based method, Divlog~\cite{Divlog}, which is the SOTA approach on template generation from log. 
Besides existing template generation methods, we adopt existing multiple-round annotation algorithm, AdaICL~\cite{adaicl} to Divlog, forming a new baseline namely AdaICL. 
For \textit{apple-to-apple} comparison, we select GPT-4o~\cite{gpt4o} as the base LLM for all LLM-based baselines and our proposed framework.
 
\begin{table}[t]
	\centering
	
	\color{black}
	\caption{\color{black}Statistics for sixteen log datasets. 
	}
	\label{tab:stat}
		\begin{tabular}{c|c|c|c}
			\hline
			\textbf{Dataset}        & \textbf{Templates\#}  & \textbf{Logs\#} &  \textbf{Words\#}     \\ \hline
			\textbf{Android}      & 165  & 437        & 857                           \\ \hline
			\textbf{BGL}     & 120  & 1367        & 2008                                  \\ \hline
			\textbf{Hadoop} & 114  & 734          & 979                                                  \\ \hline
			\textbf{HDFS}     & 14  & 2000        & 2960                                               \\ \hline
			\textbf{Linux}     & 118  & 290          & 667                                        \\ \hline
			\textbf{Mac}     & 341  & 1185         & 3136                                        \\ \hline
			\textbf{Thunderbird}     & 149  & 339         & 676                                    \\ \hline
			\textbf{Zookeeper}     & 50  & 693         & 959                           \\ \hline
			\textbf{HealthApp}     & 75  & 1179        & 1682                           \\ \hline
			\textbf{Spark}     & 36  & 1699        & 1360                                     \\ \hline
			\textbf{Windows}  & 50  & 963          & 1185                                            \\ \hline
			
			\textbf{OpenSSH}    & 27  & 729          & 692                               \\ \hline
			\textbf{OpenStack}     & 43  & 1548         & 1484                               \\ \hline
			\textbf{Proxifier}    & 8  & 1056        & 2284                                           \\ \hline
			\textbf{HPC }    & 46  & 381       & 485                            \\ \hline
			\textbf{Apache}     & 5  & 886         & 907                              \\ \hline
		\end{tabular}
\end{table}

\subsubsection{Implementation and Hyperparameter Setting} \label{ssec:para_setting}
In our experiments, our proposed LLMLog and all baselines are implemented in Python 3.9. For LLMs, we use GPT-4o~\cite{gpt4o} and Qwen2.5-7B-Instruct~\cite{qwen2.5} as our LLM backbones to conduct experiments due to their superior capabilities.
Specifically,  for  our proposed model LLMLog and all  ICL-based baselines, including Divlog~\cite{Divlog} and AdaICL~\cite{adaicl},  we set the total budget $B=50$  on five datasets, including HDFS, Proxifier, Apache, HPC and Windows, since they have fewer templates and words as shown in Table~\ref{tab:stat}.
For other datasets with more templates and words, we set the total budget $B=200$ instead. 
Specifically, for single-round annotation method Divlog, we perform DPP~\cite{dpp} algorithm to select 50 or 200 labeled logs following~\cite{Divlog}.
For multiple-round with fixed budget method AdaICL, 
 the budget per round is $10$ for $B=50$ and $40$ for $B=200$, respectively~\cite{adaicl}. In terms of LLMLog, we need manually set startup two rounds for adaptive budget which is $B_0=10$, $B_1=10$ for $B=50$ and $B_0=50$, $B_1=25$ for $B=200$ respectively. 
Also,  for baselines~\cite{drain,logppt,Divlog}, we maintain the default settings following~\cite{Divlog} with  labeled demonstration log number $k_c = 5$. 
For our  LLMLog, we set $\lambda = 0.5$ in Equation~\eqref{eq:multiroundannotation} and $\delta=0.5$ for all dataset. 

\begin{table*}[h]
	\centering
	
	\caption{Effectiveness (accuracy) over 16 log datasets on GPT-4o. 
		The \textbf{bold number} indicates the best performance.
	}	
	\label{tab:mainexp}
	\setlength\tabcolsep{4.2pt}
	\begin{tabular}{c|ccc|ccc|ccc|ccc|ccc}
		\hline
		
		\multirow{2}{*}{\textbf{Dataset}} & \multicolumn{3}{c|}{\textbf{Drain}}                                   & \multicolumn{3}{c|}{\textbf{LogPPT}}      & \multicolumn{3}{c|}{\textbf{DivLog}} & \multicolumn{3}{c|}{\textbf{AdaICL}}                                   & \multicolumn{3}{c}{\textbf{LLMLog (Ours)}}     \\
		& \multicolumn{1}{c}{\textbf{MLA}} & \multicolumn{1}{c}{\textbf{PTA}} &  \multicolumn{1}{c|}{\textbf{RTA}}  & \multicolumn{1}{c}{\textbf{MLA}} & \multicolumn{1}{c}{\textbf{PTA}} &  \multicolumn{1}{c|}{\textbf{RTA}} &  \multicolumn{1}{c}{\textbf{MLA}} & \multicolumn{1}{c}{\textbf{PTA}} &  \multicolumn{1}{c|}{\textbf{RTA}}   	& \multicolumn{1}{c}{\textbf{MLA}} & \multicolumn{1}{c}{\textbf{PTA}} &  \multicolumn{1}{c|}{\textbf{RTA}}    & \multicolumn{1}{c}{\textbf{MLA}} & \multicolumn{1}{c}{\textbf{PTA}} &  \multicolumn{1}{c}{\textbf{RTA}}  \\ \hline
		\textbf{Android} &73.0 & 56.6 &  \multicolumn{1}{c|}{62.0}  &76.7 & 58.4  &  \multicolumn{1}{c|}{68.4}  &63.8& 58.9  &  \multicolumn{1}{c|}{68.4} &97.8 & 89.4 &  \multicolumn{1}{c|}{92.1}  &\textbf{99.6} & \textbf{94.6}  &  \multicolumn{1}{c}{\textbf{96.4}}   \\ 
		\textbf{BGL} &44.4 & 33.9 &  \multicolumn{1}{c|}{30.8}  &97.0 & 68.6 &  \multicolumn{1}{c|}{78.3}  &94.0 & 68.4  &  \multicolumn{1}{c|}{77.5} &99.4 & 93.5 &  \multicolumn{1}{c|}{95.8}  &\textbf{99.9} & \textbf{95.1}  &  \multicolumn{1}{c}{\textbf{98.3}}   \\ 
		\textbf{Hadoop} &43.9 & 36.8  &  \multicolumn{1}{c|}{34.2}  &89.5 & 54.0  &  \multicolumn{1}{c|}{58.8}   &89.0 & 69.3  &  \multicolumn{1}{c|}{85.1} &99.4 & 92.2  &  \multicolumn{1}{c|}{97.4}  &\textbf{100.0} & \textbf{100.0}  &  \multicolumn{1}{c}{\textbf{100.0}}   \\ 
		\textbf{HDFS} &95.9 & 81.3  &  \multicolumn{1}{c|}{92.9}  &90.2 & 85.7  &  \multicolumn{1}{c|}{85.7}    &\textbf{100.0} & \textbf{100.0}  &  \multicolumn{1}{c|}{\textbf{100.0}}  &99.9 & 86.7  &  \multicolumn{1}{c|}{92.9}  &\textbf{100.0} & \textbf{100.0}  &  \multicolumn{1}{c}{\textbf{100.0}}   \\ 
		\textbf{Linux} &19.4 & 43.4 &  \multicolumn{1}{c|}{42.2}  &94.9 & 47.5  &  \multicolumn{1}{c|}{49.1} &97.3& 92.4  &  \multicolumn{1}{c|}{93.2} &99.7 & 96.6 &  \multicolumn{1}{c|}{96.6}  &\textbf{99.8} & \textbf{96.6}  &  \multicolumn{1}{c}{\textbf{98.3}}   \\ 
		\textbf{Mac} &27.2& 21.2  &  \multicolumn{1}{c|}{24.9}  &67.3& 43.6  &  \multicolumn{1}{c|}{53.4}    &62.4& 48.3  &  \multicolumn{1}{c|}{64.5}  &93.2& 74.4  &  \multicolumn{1}{c|}{82.1}  &\textbf{96.0}& \textbf{77.1}  &  \multicolumn{1}{c}{\textbf{85.9}}   \\ 
		\textbf{Thunderbird} &19.1 & 29.9 &  \multicolumn{1}{c|}{36.9}  &92.6 & 50.6  &  \multicolumn{1}{c|}{59.1}    &88.9& 86.8  &  \multicolumn{1}{c|}{92.6} &98.9 & 83.3 &  \multicolumn{1}{c|}{90.6}  &\textbf{99.9} & \textbf{93.3}  &  \multicolumn{1}{c}{\textbf{98.7}}   \\ 
		\textbf{Zookeeper} &49.8 & 39.1 &  \multicolumn{1}{c|}{36.0}  &99.0 & 74.1  &  \multicolumn{1}{c|}{86.0}   &\textbf{100.0} & \textbf{100.0}  &  \multicolumn{1}{c|}{\textbf{100.0}}  &\textbf{100.0} & \textbf{100.0}  &  \multicolumn{1}{c|}{\textbf{100.0}}  &\textbf{100.0} & \textbf{100.0}  &  \multicolumn{1}{c}{\textbf{100.0}}    \\ 
		\textbf{HealthApp} &24.1 & 8.3 &  \multicolumn{1}{c|}{34.7}  &78.9 & 85.3  &  \multicolumn{1}{c|}{85.3}   &99.9 & 98.7 &  \multicolumn{1}{c|}{98.7} &99.9 & 98.7 &  \multicolumn{1}{c|}{98.7}   &\textbf{100.0} & \textbf{100.0}  &  \multicolumn{1}{c}{\textbf{100.0}}  \\ 
		\textbf{Spark} &37.6 & 50.0 &  \multicolumn{1}{c|}{41.7}  &99.1 & 60.0 &  \multicolumn{1}{c|}{58.3} &82.1 & 48.3 &  \multicolumn{1}{c|}{77.8} &99.9 & 97.2 &  \multicolumn{1}{c|}{97.2}   &\textbf{100.0} & \textbf{100.0}  &  \multicolumn{1}{c}{\textbf{100.0}}   \\ 
		\textbf{Windows} &69.6 & 46.3  &  \multicolumn{1}{c|}{50.0}  &98.3 & 55.4  &  \multicolumn{1}{c|}{72.0}    &97.6& 55.9  &  \multicolumn{1}{c|}{76.0}  &99.9 &92.3 &  \multicolumn{1}{c|}{96.0}  &\textbf{100.0} & \textbf{100.0}  &  \multicolumn{1}{c}{\textbf{100.0}}  \\ 
		\textbf{OpenSSH} &53.4 & 52.0  &  \multicolumn{1}{c|}{50.0}  &97.6 & 48.9 &  \multicolumn{1}{c|}{84.6}  &99.9& 96.3  &  \multicolumn{1}{c|}{96.3} &99.9  & 96.3  &  \multicolumn{1}{c|}{96.3}  &\textbf{100.0} & \textbf{100.0}  &  \multicolumn{1}{c}{\textbf{100.0}}   \\ 
		\textbf{OpenStack} &18.0 & 5.5 &  \multicolumn{1}{c|}{39.5}  &90.7 & 84.4  &  \multicolumn{1}{c|}{88.4}   &96.9& 74.0 &  \multicolumn{1}{c|}{88.4} &\textbf{100.0} & \textbf{100.0}  &  \multicolumn{1}{c|}{\textbf{100.0}}   &\textbf{100.0} & \textbf{100.0}  &  \multicolumn{1}{c}{\textbf{100.0}}   \\ 
		\textbf{Proxifier} &52.7 & 26.9  &  \multicolumn{1}{c|}{87.5}  &\textbf{100.0} & \textbf{100.0}  &  \multicolumn{1}{c|}{\textbf{100.0}}  &96.5 & 14.3  &  \multicolumn{1}{c|}{75.0} &99.9& 77.8  &  \multicolumn{1}{c|}{87.5}  &\textbf{100.0} & \textbf{100.0}  &  \multicolumn{1}{c}{\textbf{100.0}}   \\ 
		\textbf{HPC} &67.2 & 38.8 &  \multicolumn{1}{c|}{41.3}  &94.7 & 73.6  &  \multicolumn{1}{c|}{84.8}  &97.5 & 42.6&  \multicolumn{1}{c|}{87.0}&98.6 & 62.5&  \multicolumn{1}{c|}{97.8}  &\textbf{100.0} & \textbf{100.0}  &  \multicolumn{1}{c}{\textbf{100.0}}  \\ 
		\textbf{Apache} &\textbf{100.0} & \textbf{100.0}  &  \multicolumn{1}{c|}{\textbf{100.0}}   &99.4 & 83.3  &  \multicolumn{1}{c|}{83.3}    &\textbf{100.0} & \textbf{100.0}  &  \multicolumn{1}{c|}{\textbf{100.0}}&\textbf{100.0} & \textbf{100.0}  &  \multicolumn{1}{c|}{\textbf{100.0}}  &\textbf{100.0} & \textbf{100.0}  &  \multicolumn{1}{c}{\textbf{100.0}}    \\ 
		\hline
	\end{tabular}
	
\end{table*}

All experiments are conducted on CentOS 7 with a 20-core Intel(R) Xeon(R) Silver4210 CPU@2.20GHz, 8 NVIDIA GeForce RTX 2080 Ti GPUs (11G), and 92G of RAM.

\subsection{Main Results} \label{ssec:main_results}

For main experiments, we evaluate the effectiveness by reporting MLA, PTA and RTA over 16 datasets on GPT-4o in Table \ref{tab:mainexp}. 
In terms of efficiency, we report the average template generation time for unlabeled logs and API cost for LLM-based approaches. Under the same dataset, the bold number indicates the best performance. 

\subsubsection{Effectiveness Evaluation}  \label{ssec:effective_exp}

We compare our LLMLog with state-of-the-art baselines using three metrics: MLA, PTA, and RTA on sixteen datasets, as shown in Table~\ref{tab:mainexp}.
Regarding MLA, LLMLog and AdaICL outperform DivLog by leveraging the benefits of multi-round annotation, enabling more effective and accurate log processing.
Our LLMLog outperforms AdaICL in terms of MLA on all datasets.
The reason is that, for HDFS and Proxifier with fewer templates, SED in LLMLog outperforms traditional cosine similarity in AdaICL by generating a higher-quality labeled log set for annotation. These annotated logs serve as effective demonstrations for a large number of unlabeled logs. Additionally, the adaptive demonstration strategy in LLMLog ensures that each word in the unlabeled logs is accurately processed, achieving higher accuracy even with a limited budget.
For datasets with diverse templates and words, such as Mac, LLMLog reduces redundant contexts, providing clear and sufficient context for each unlabeled log compared to fixed top-$k_c$ demonstrations.

In terms of template-level accuracy, the PTA and RTA metrics tend to be lower than MLA since they require all logs belonging to a template to be correctly predicted.
Specifically, if more words are mistakenly generated in predicted templates, PTA will drop because the total number of predicted templates increases.
Therefore, the lower PTA of DivLog and AdaICL compared to our proposed LLMLog reflects the false generation of word types, indicating that they provide low-quality contexts to unlabeled logs.
LLMLog ensures higher-quality contexts through both the SED-based representative score and the adaptive demonstration strategy. Thus, LLMLog performs better than the baselines on the PTA metric.
As for RTA, it evaluates how many templates in a dataset are correctly predicted.
Though Drain and LogPPT achieve 100 percent accuracy on Apache and Proxifier, respectively, their low RTA on other datasets illustrates their low generalization abilities.
In summary, LLMLog outperforms AdaICL on PTA and RTA, proving its suitability for large-scale log datasets.

\subsubsection{Efficiency Evaluation} \label{ssec:efficiency_exp}

We compare the template generation efficiency for LLMLog with baseline LLM-based methods.
Template generation efficiency is measured by average LLM prediction time for single log on each dataset, summarized in Table~\ref{tab:effexp}.
The average LLM prediction time for LLMLog is less than $0.8$ seconds on all 16 datasets while the time for DivLog and AdaICL is larger than 1 second.
As all three methods are measured under the same API, the less time consumption reflects the fewer number of input tokens. 
The adaptive demonstration selection minimizes the number of example logs based on word coverage, which reduces the number of input tokens in contextual demonstration compared to fixed top-$k_c$ logs in DivLog and AdaICL.
On datasets with fewer words and templates like Apache, Windows and HPC, adaptive demonstration can cover all the words within 1 or 2 log. Compared to $k_c=5$ setting in DivLog and AdaICL, the generation time of our proposed LLMLog is significantly decreased.
On datasets with more words and templates like mac, it requires more example logs to cover the words. The difference of generation time between AdaICL and LLMLog is less than that on simpler datasets.
API Cost experiments also prove the claim.
\begin{table}[t]
	\centering
	\small
	\caption{Template generation time (time with seconds) and API Cost (in USD) across 16 log datasets. 
		The \textbf{bold number} indicates the most efficient and lower API cost results.
	}
	\setlength\tabcolsep{1pt}
	\label{tab:effexp}
	
	\begin{tabular}{c|ccc|ccc}
		\hline
		& \multicolumn{3}{c|}{\textbf{Generation Time (s)}}& \multicolumn{3}{c}{\textbf{API Cost (USD)}} \\
		\textbf{Dataset}  & \textbf{DivLog} & \textbf{AdaICL}  & \textbf{LLMLog}  & \textbf{DivLog}  & \textbf{AdaICL} & \textbf{LLMLog}    \\ \cline{1-7}
		\textbf{Android} &1.1 & 1.1  & \textbf{0.7}    & 3.5 & 3.5        &   \textbf{2.1}   \\ 
		\textbf{BGL} &1.4 & 1.1 & \textbf{0.8}   &  3.8  &  3.7        &    \textbf{2.8}  \\ 
		\textbf{Hadoop} &1.1 & 1.0  & \textbf{0.6}   &    3.7 &    3.7      &   \textbf{2.0}    \\ 
		\textbf{HDFS} &1.0& 1.0  &  \textbf{0.7}   &7.6   &7.6        &    \textbf{5.1} \\ 
		\textbf{Linux} &1.2 & 1.1 &  \textbf{0.7}    &   3.1 &   3.1        &    \textbf{2.0}  \\ 
		\textbf{Mac} &1.1 &1.1   &  \textbf{0.8}  &    5.5  &    5.4        &  \textbf{5.0}  \\ 
		\textbf{Thunderbird} &1.1 & 1.0 &  \textbf{0.6} &  3.3  &  3.2          &   \textbf{2.7}  \\
		\textbf{Zookeeper} &1.1 & 1.1 &  \textbf{0.3} &    3.1  &    3.1       &   \textbf{2.1}  \\
		\textbf{HealthApp}  &1.8 & 1.4  & \textbf{0.8}    &  4.4    &  3.7    &  \textbf{2.3}    \\ 
		\textbf{Spark} &2.1& 1.5 & \textbf{1.4}  &6.5    &5.7       &  \textbf{3.3}   \\ 
		\textbf{Windows} &2.2 & \textbf{0.7} & \textbf{0.7}  &   5.3 &   5.3       &   \textbf{2.6}  \\ 
		\textbf{OpenSSH }&1.1 & 1.1  &  \textbf{0.7}  &    7.8  &    7.8        &      \textbf{2.5}  \\ 
		\textbf{OpenStack}  &1.8 & 1.9 &  \textbf{1.0}   &   9.8  &   10.0          &   \textbf{4.2}   \\ 
		\textbf{Proxifier} &0.9 &0.9   &  \textbf{0.8} &   5.7  &   5.7    &   \textbf{3.5}       \\ 
		\textbf{HPC}&1.6& 0.5 &  \textbf{0.3}  & 2.6      & 1.8         &      \textbf{1.3}   \\
		\textbf{Apache} &1.8 & 0.5&  \textbf{0.4} &   2.7    &   2.4        & \textbf{1.6}       \\
		\hline
	\end{tabular}
\end{table}
\begin{table*}[t]
	\centering
	\color{black}
	\small
	\caption{\color{black}Ablation study on GPT-4o. }
	\label{tab:ablation}
	\setlength\tabcolsep{1pt}
	
	\renewcommand{\arraystretch}{1.5}
	\begin{tabular}{c|ccc|ccc|ccc|ccc}
		\hline
		\multirow{2}{*}{\diagbox{\textbf{Dataset}}{\textbf{Model}}} & \multicolumn{3}{c|}{\textbf{Mac}} & \multicolumn{3}{c|}{\textbf{BGL}} & \multicolumn{3}{c|}{\textbf{Hadoop}} & \multicolumn{3}{c}{\textbf{Proxifier}} \\ \cline{2-13}
		& \textbf{MLA}   & \textbf{PTA} &  \textbf{RTA}  & \textbf{MLA}   & \textbf{PTA} &  \textbf{RTA}  & \textbf{MLA}   & \textbf{PTA} &  \textbf{RTA}   & \textbf{MLA}   & \textbf{PTA} &  \textbf{RTA}    \\ \cline{1-13}

		\textbf{LLMLog\textbackslash SED}     & $93.7_{(-3.3)}$   & $65.6_{(-15.7)}$        & $77.1_{(-9.1)}$ & $87.2_{(-12.7)}$   & $35.3_{(-61.4)}$        & $52.5_{(-45.8)}$ & $95.5_{(-4.5)}$   & $85.4_{(-14.6)}$        & $92.1_{(-7.9)}$ & $75.7_{(-24.3)}$   & $37.5_{(-62.5)}$        &$37.5_{(-62.5)}$  \\
		\textbf{LLMLog\textbackslash RS}      & $93.9_{(-3.1)}$   & $66.5_{(-14.8)}$        & $77.7_{(-8.5)}$ & $94.5_{(-5.4)}$   & $83.0_{(-13.7)}$        & $89.2_{(-9.1)}$ & $94.1_{(-5.9)}$   & $82.9_{(-17.1)}$        & $97.4_{(-2.6)}$ & $99.4_{(-0.6)}$   & $37.5_{(-62.5)}$        &$37.5_{(-62.5)}$  \\ 
		\textbf{LLMLog\textbackslash PC} & $96.0_{(-1.0)}$   & $77.10_{(-4.2)}$        & $85.9_{(-0.3)}$  & $99.0_{(-0.9)}$  & $93.5_{(-3.2)}$   & $95.8_{(-2.5)}$ & $99.5_{(-0.5)}$   & $98.3_{(-1.7)}$        & $99.1_{(-0.9)}$ & $100.0_{(-0.0)}$   & $100.0_{(-0.0)}$       & $100.0_{(-0.0)}$\\
		\textbf{LLMLog\textbackslash AD}      & $93.1_{(-3.9)}$  & $67.6_{(-13.7)}$       & $76.3_{(-9.9)}$ & $97.3_{(-2.6)}$  & $93.5_{(-3.2)}$       & $96.7_{(-1.6)}$ & $99.5_{(-0.5)}$  & $92.6_{(-7.4)}$       & $98.2_{(-1.8)}$ & $100.0_{(-0.0)}$   & $100.0_{(-0.0)}$       & $100.0_{(-0.0)}$\\ 	
		\textbf{LLMLog\textbackslash AB}      & $96.9_{(-0.1)}$  & $80.2_{(-1.1)}$       & $86.2_{(-0.0)}$ & $97.6_{(-3.2)}$  & $90.6_{(-6.1)}$       & $96.7_{(-1.6)}$ & $100.0_{(-0.0)}$   & $100.0_{(-0.0)}$       & $100.0_{(-0.0)}$ & $100.0_{(-0.0)}$   & $100.0_{(-0.0)}$       & $100.0_{(-0.0)}$\\ \hline
				\textbf{LLMLog}     & \textbf{97.0} & \textbf{81.3}        &\textbf{86.2} & \textbf{99.9}  & \textbf{96.7}        &\textbf{98.3} & \textbf{100.0}  & \textbf{100.0}        & \textbf{100.0} & \textbf{100.0}  & \textbf{100.0}        & \textbf{100.0}\\ \hline
	\end{tabular}
\end{table*}

\subsubsection{API Cost Evaluation} \label{ssec:cost_exp}
{\color{black}
 We compare the total API monetary costs for LLMLog and the state-of-the-art baselines, DivLog~\cite{Divlog} and AdaICL~\cite{adaicl}, as summarized in Table~\ref{tab:effexp}.
 The cost of GPT-4o is approximately \$3.6 USD per one million tokens.
In Table~\ref{tab:effexp}, the cost of our LLMLog ranges from \$1 to \$5 USD per dataset, where the cost of processing each log is only \$0.0025-\$0.005 USD. Therefore, the cost of LLMLog is both cheap and practical.
 Moreover, LLMLog incurs less API cost than the state-of-the-art baselines, as it adaptively selects the number of labeled logs to use as demonstrations for each unlabeled log, whereas AdaICL relies on a fixed number of demonstrations. By eliminating unnecessary log demonstrations, LLMLog significantly reduces the input token length for LLMs, further lowering the computational cost.
 LLMLog can effectively reduce costs for simple datasets like HDFS and Proxifier since the words can be adequately covered by one or two logs.
 As for complex datasets like Mac, the amount of cost savings is smaller while still better than the baselines, as each unlabeled log requires more example logs to cover the words.}

\subsection{Ablation Study}\label{ssec:ablation}
{\color{black}
This subsection analyzes the impact of components in LLMLog. For the log annotation, it introduces the SED metric to calculate the representative score of logs. Logs are selected for annotation by optimizing a weighted combination of the LLM's prediction confidence and the representative score. In the adaptive demonstration selection component, Algorithm \ref{alg:adaptivedemon} is proposed to adaptively determine suitable contexts for each unlabeled log. Additionally, an adaptive budget strategy (Equation \eqref{eq:avgbudget}) dynamically allocates the annotation budget for each round.
To verify these designs,
we denote LLMLog with SED replaced by cosine similarity as \textbf{LLMLog\textbackslash SED}, without the representative score as \textbf{LLMLog\textbackslash RS}, without LLM prediction confidence as \textbf{LLMLog\textbackslash PC}, with Algorithm~\ref{alg:adaptivedemon} replaced by a fixed top-$k_c$ strategy as \textbf{LLMLog\textbackslash AD}, and with Equation~\eqref{eq:avgbudget} replaced by a fixed budget for each round as \textbf{LLMLog\textbackslash AB}. For \textbf{LLMLog\textbackslash AB}, we set the budget for each round to $10$ for Proxifier and $40$ for Mac, BGL, and Hadoop. We conduct experiments on 4 datasets with different template distributions, including Mac, BGL, Hadoop, and Proxifier.

As shown in Table~\ref{tab:ablation}, \textbf{LLMLog\textbackslash SED} suffers from a significant accuracy drop because SED eliminates redundant words in logs, treating logs with the same template as similar. Regarding the representative score, the performance decrease of \textbf{LLMLog\textbackslash RS} shows that logs similar to the majority provide useful contexts for generating templates. On the other hand, \textbf{LLMLog\textbackslash PC} performs slightly better by selecting challenging logs for the LLM, but these logs are less representative, making their overall impact smaller.
Regarding the adaptive demonstration strategy, \textbf{LLMLog\textbackslash AD} shows reduced accuracy as template complexity increases. While a fixed top-$k_c$ strategy works well for simpler datasets, it struggles on complex datasets like Mac, where insufficient context leads to irregular processing and more generated templates. This results in a larger accuracy drop compared to simpler datasets like PTA. For the adaptive budget strategy, \textbf{LLMLog\textbackslash AB} shows that both fixed and adaptive strategies perform well on simpler datasets. However, in complex datasets with more templates and words, the adaptive strategy limits annotations per round, selecting labeled logs that are more diverse in word count and template variety.}

\subsection{Parameter Sensitivity} \label{ssec:para}
	We  evaluate the effectiveness of different hyper-parameter settings of LLMLog over two datasets, Hadoop and Proxifier.

\subsubsection{Annotation Budget $B$ in Algorithm~\ref{alg:adaptivedemon} } \label{sssec:budget_para}
The annotation budget $B$ represents the total budget for human annotation.
We vary $B$ within $\{50, 100, 150, 200, 250\}$.
On the Hadoop dataset in Figure~\ref{fig:hyperparameter_sensitivity}~(a), the three accuracy metrics increase first and then stabilize at $B=200$, indicating $B=200$ is sufficient to cover most templates on Hadoop.
Besides, the increment of PTA is larger than other two metrics due to the LLM falsely generating incorrect word types under insufficient $B$.
As Proxifier is a much simpler dataset containing only $4$ templates in Figure~\ref{fig:hyperparameter_sensitivity}~(b), all three metrics stabilize at $B=50$ which can provide sufficient contextual information for each word in unlabled logs.
To sum up, the performance stabilizes after $B$ is large enough, rather than peaking, making tuning easier.

\subsubsection{The weight $\lambda$ in Equation~\eqref{eq:multiroundannotation}} \label{ssec:lambdapara}
$\lambda$ is a trade-off parameter between the representative score and the LLM prediction confidence. 
We vary $\lambda\in \{0, 0.25, 0.5, 0.75, 1\}$.
As demonstrated in Figure~\ref{fig:hyperparameter_sensitivity}(c), MLA, PTA, and RTA initially increase and subsequently decrease over Hadoop, peaking within $0.25$ to $0.75$.
Small value of $\lambda$ undervalues the impact of LLM confidence, resulting in the selection of annotation logs with redundant information. In contrast, large $\lambda$ prioritizes logs with low template generation confidence.
In Figure~\ref{fig:hyperparameter_sensitivity}~(d), most logs are easily identified using representative contextual information, achieving 100\% accuracy even when $\lambda = 0$ over Proxifier. However, relying solely on LLM prediction confidence ($\lambda = 1$) causes LLMLog to focus only on low-confidence logs which is not appropriate even over simple datasets.

\subsubsection{The threshold $\delta$ in Equation~\eqref{eq:repre_set}}

$\delta$ controls the threshold of representative score.
We vary it within $\{0, 0.25, 0.5, 1.0\}$.
As shown in Figure~\ref{fig:hyperparameter_sensitivity}~(e), MLA, PTA, and RTA exhibit a trend of rising and then falling, roughly peaking at $\delta=0.5$. A low threshold underestimates the informativeness of logs. The annotation focuses excessively on LLM confidence.
On the contrary, a high threshold causes only a subset of unlabeled logs obtaining enough context.

 \begin{figure*}[t]
	
	\color{black}
	\centering 	
	\subfloat[$B$ over Hadoop]	{\centering\includegraphics[width = 0.2\linewidth]{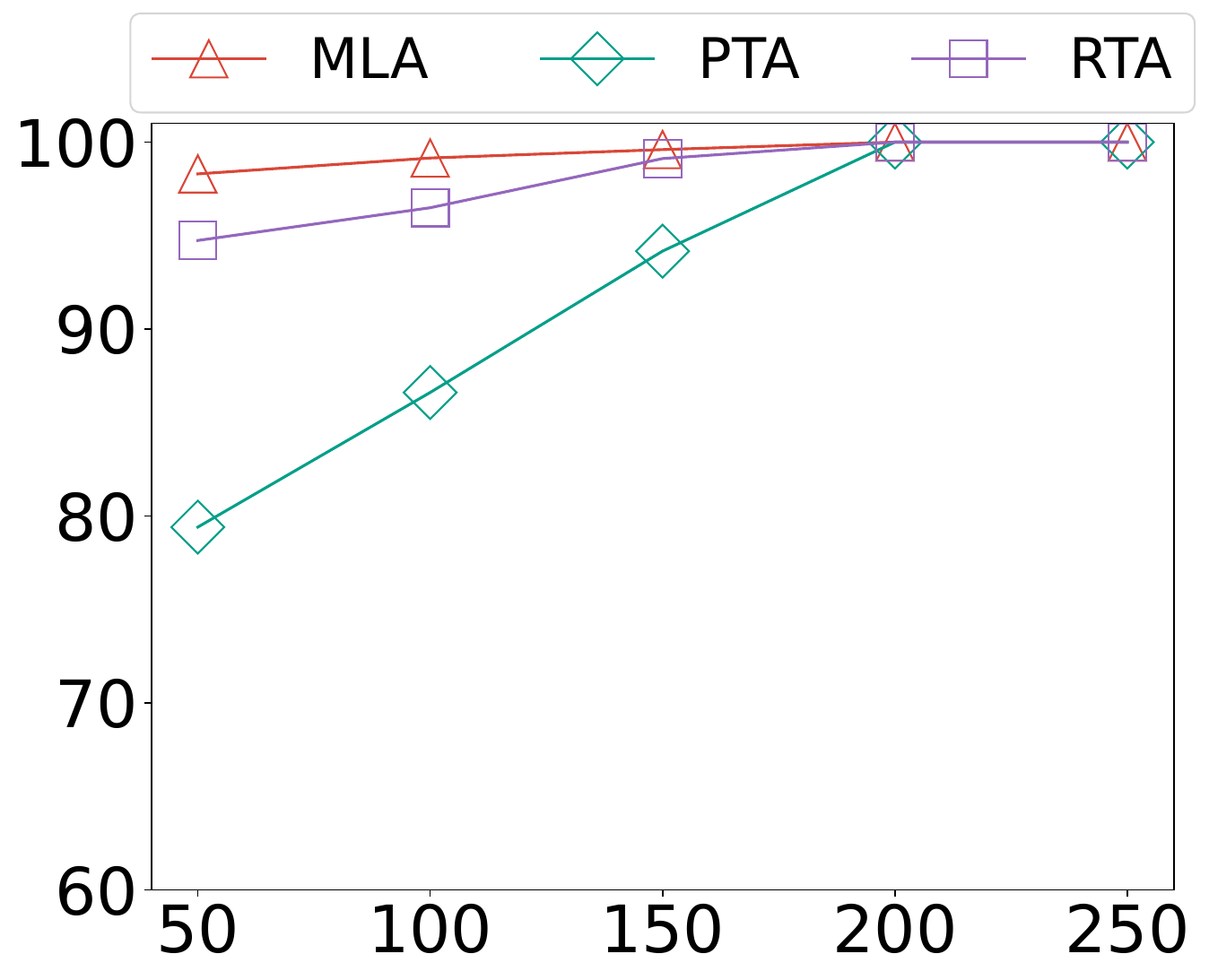}}
	\subfloat[$B$ over Proxifier]	{\centering\includegraphics[width = 0.2\linewidth]{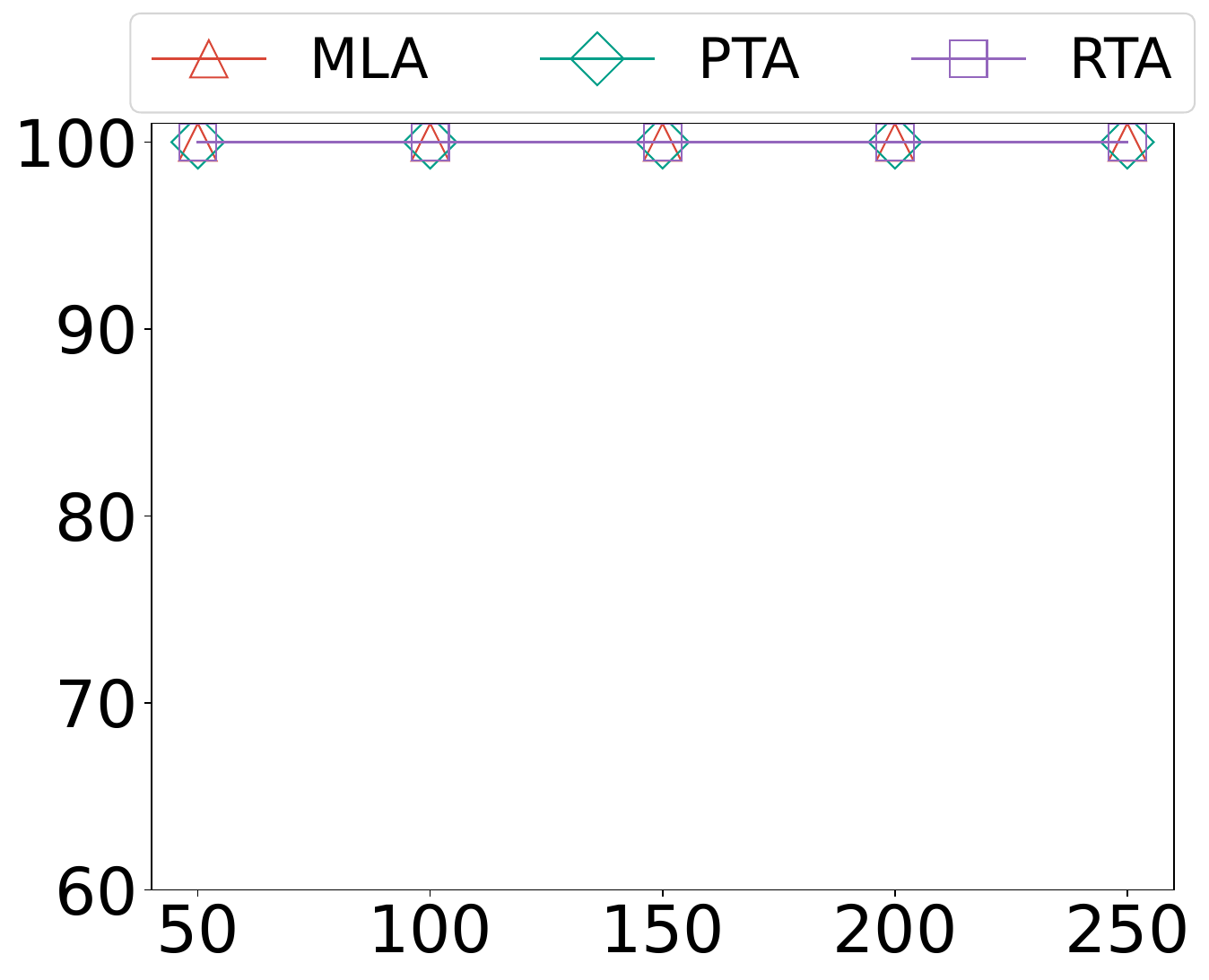}}
	\subfloat[$\lambda$ over Hadoop]	{\centering\includegraphics[width = 0.2\linewidth]{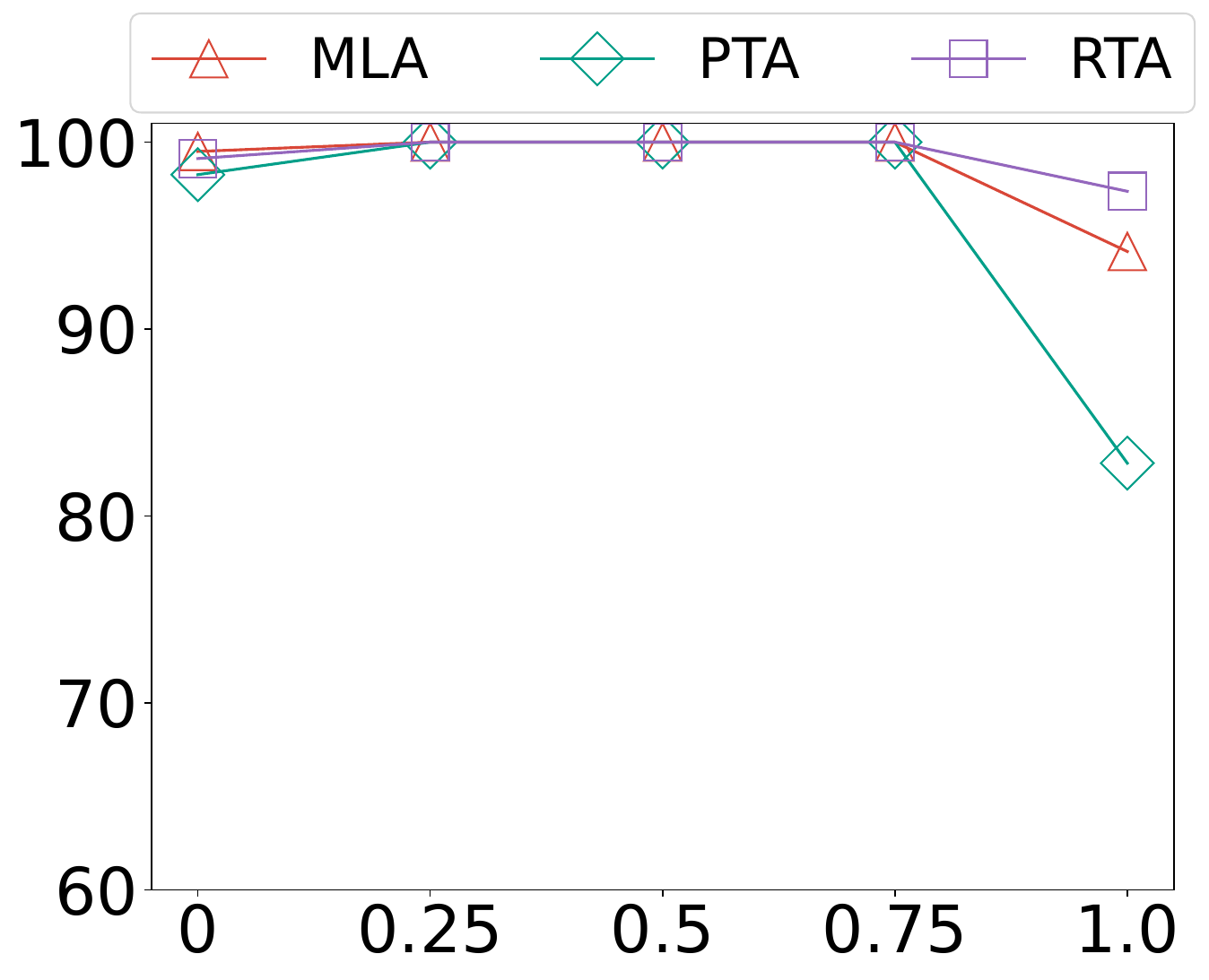}}
	\hfill
	\subfloat[$\lambda$ over Proxifier]	{\centering\includegraphics[width = 0.2\linewidth]{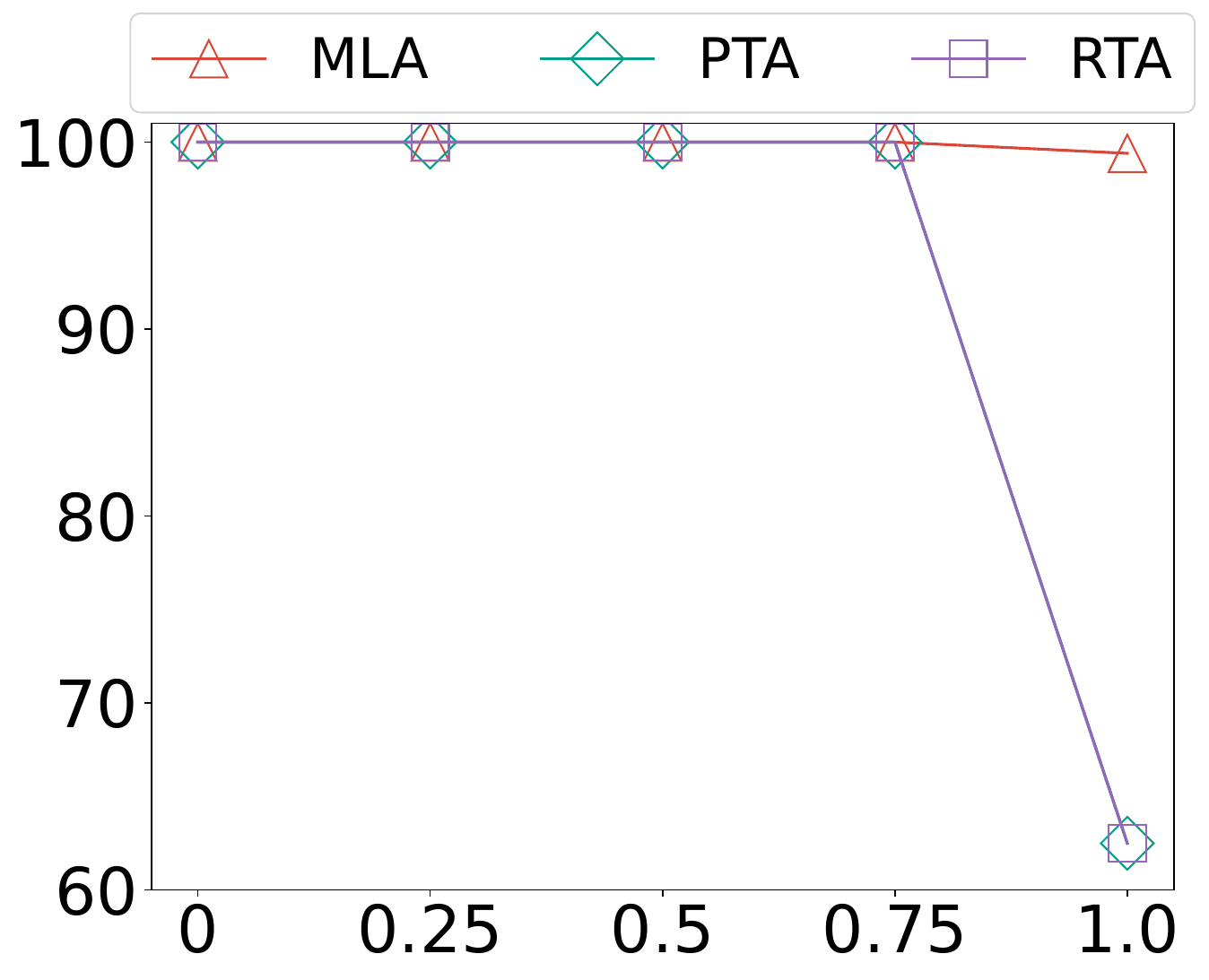}}
	\hfill
	\subfloat[$\delta$ over Hadoop]	{\centering\includegraphics[width = 0.2\linewidth]{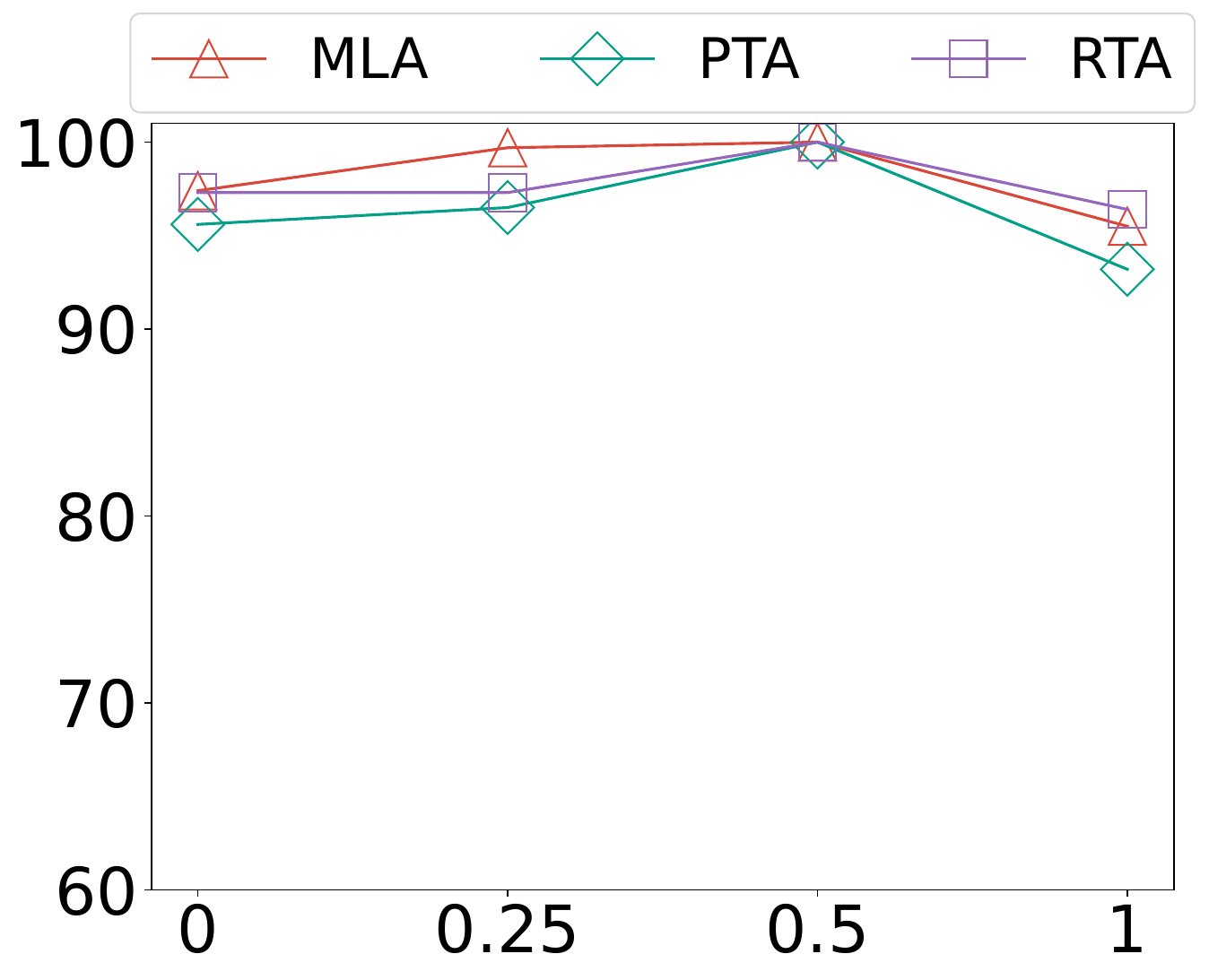}}
	\caption{Parameter sensitivity evaluations}
	\label{fig:hyperparameter_sensitivity}
\end{figure*}

 \begin{figure}[t]
	\centering 	
	\subfloat[Word similarity over Hadoop]	{\centering\includegraphics[width = 0.5\linewidth]{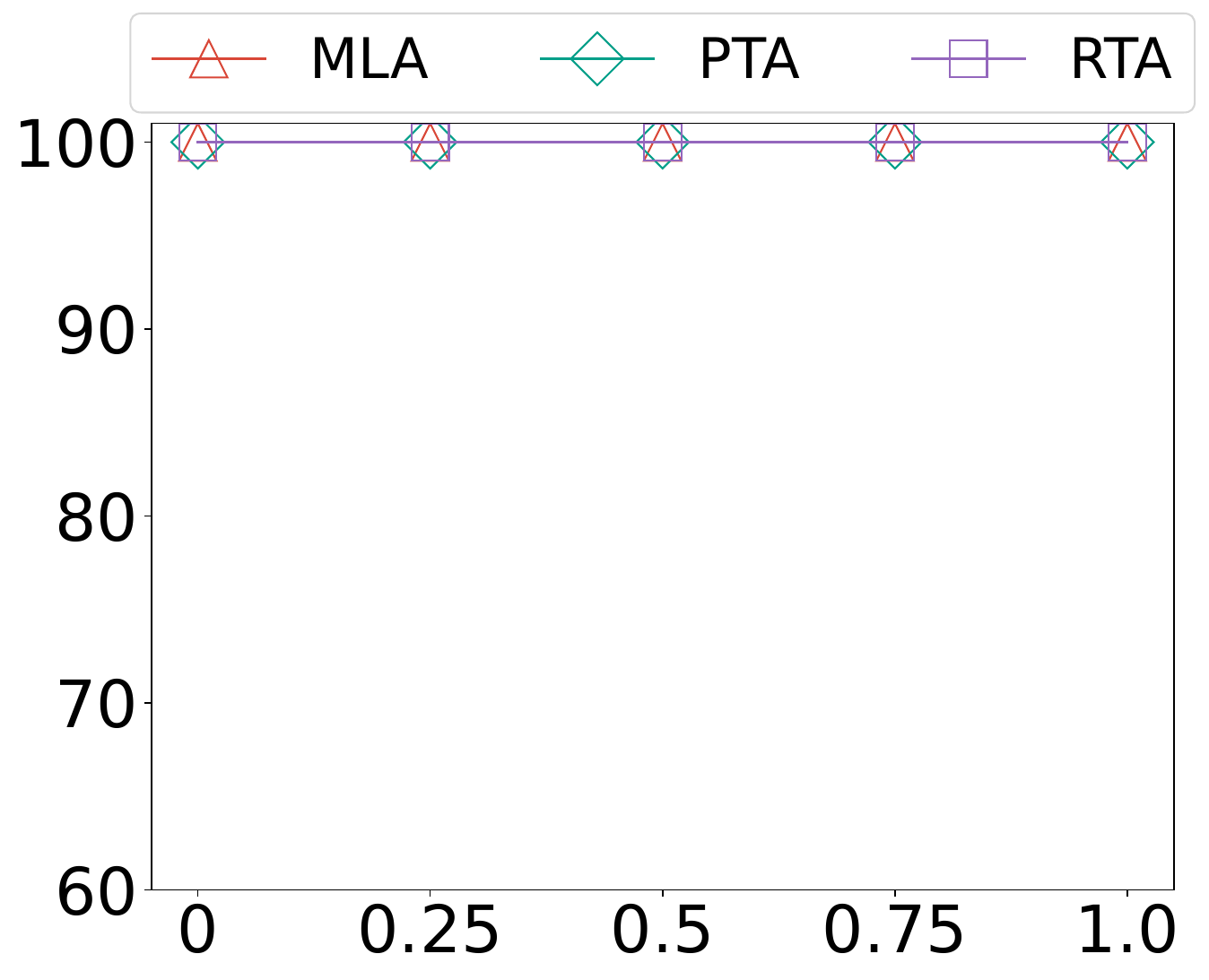}}
	\subfloat[Word similarity over Proxifier]	{\centering\includegraphics[width = 0.5\linewidth]{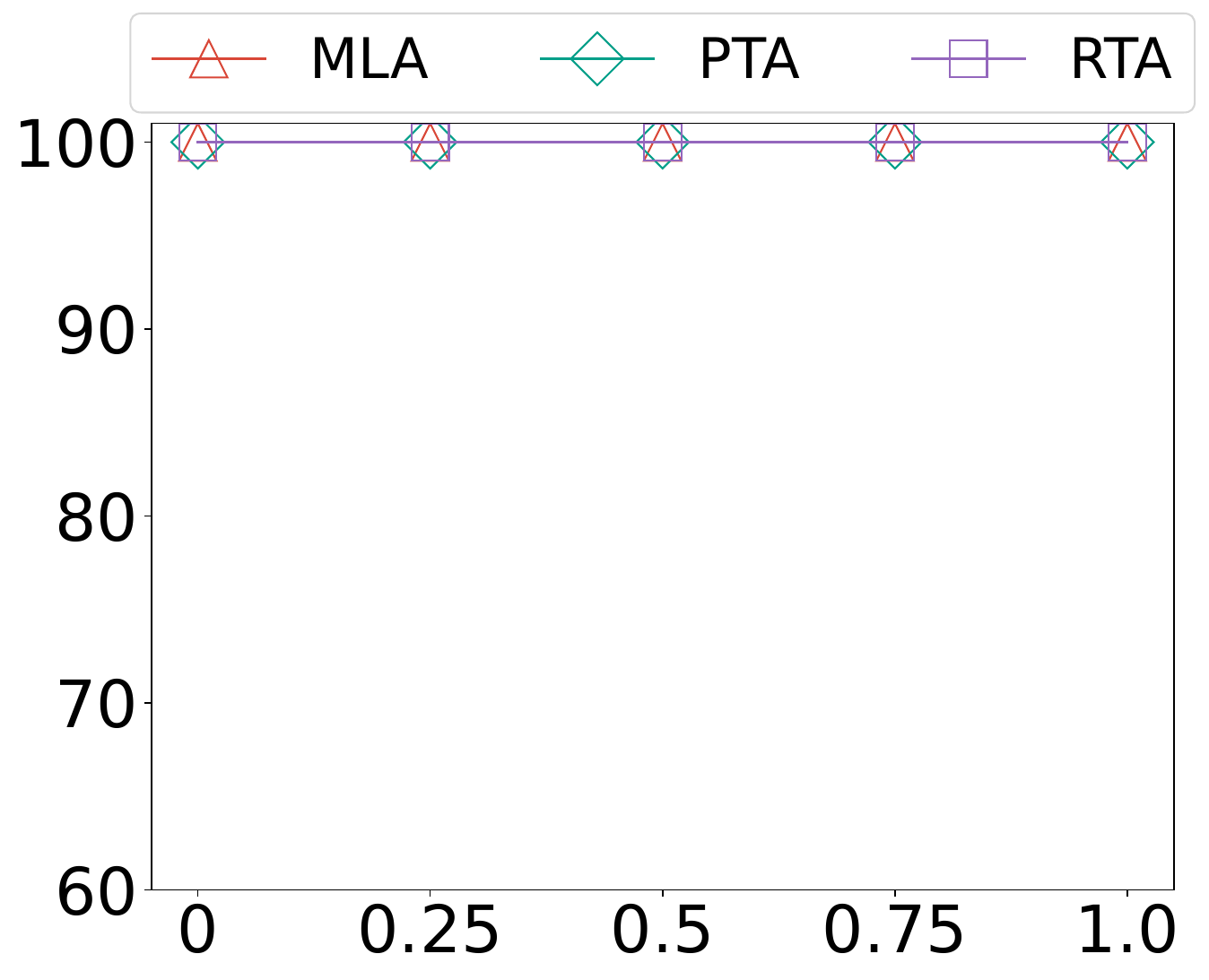}}
	\caption{Parameter sensitivity for word similarity}
	\label{fig:similarity_sensitivity}
 
\end{figure}

\subsubsection{Cosine Similarity Threshold in Equation~\eqref{eq:cost} and~\eqref{eq:adaptivedemon}}
Cosine similarity $cosine(\cdot)$ measures word similarity. Two words are considered similar if their embedding cosine similarity is greater than the threshold $0$.
We vary it within $\{0, 0.25, 0.5, 0.75, 1\}$. As shown in Figure~\ref{fig:similarity_sensitivity}~(a) and (b), the effect of varying threshold of cosine similarity has converged over two datasets.
This implies that LLMLog is robust to different settings of word similarity under a certain budget. Since the total number of distinct words is relatively small in system events, it is sufficient to distinguish words by $0$.

\subsection{Case Study} \label{ssec:case_study}
Hallucination is that LLM generates outputs without following the prompts, which is a common problem in LLM-related tasks~\cite{Divlog,adaicl,vldbicl1,vldbicl2}.
There are mainly two types of error caused by hallucination. 
The first is \textbf{generation error} that LLM falsely generates or deletes words in input logs. For instance, input log is \texttt{rts: kernel terminated for reason 1004} with the ground truth \texttt{rts: kernel terminated for reason [CODE]}. However, LLM may predict \texttt{rts: kernel terminated} where \texttt{for reason 1004} are falsely deleted. 
The second case is \textbf{word error} that even the type of a target word is included in prompt, LLM still makes wrong predictions.
 For example, input log is \texttt{rts: kernel terminated for reason 1004} and the prompt has instructed to replace \texttt{1004} to word type \texttt{[CODE]}, LLM still mistakenly remains \texttt{1004} in predicted template.
To investigate how confidence score in Equation~\eqref{eq:LLMhardness} help alleviate the hallucination issue, we vary two hyperparameters related to the confidence score. One is $\lambda$ in Equation~\eqref{eq:multiroundannotation}, the trade-off parameter for prediction confidence.
The other is $a$ in Equation~\eqref{eq:LLMhardness}, the trade-off parameter for the effect of token probability in the prediction confidence score.
First, we vary $\lambda \in \{0, 0.25, 0.5\}$.
As shown in Figure~\ref{fig:confidence}~(a), as $\lambda$ increases, both generation error and word error are reduced, implying that the confidence score can effectively select several "hard" logs, allowing human annotation to replace the LLM's hallucinated output.
On the other hand, we vary $a \in \{0.2, 0.5, 0.8\}$.
As shown in Figure~\ref{fig:confidence}~(b), as $a$ increases, logs with low prediction probability are included in the labeled log set. Thus, the labeled log set becomes effective at preventing word errors associated with low prediction probabilities. However, generation errors increase because several word-inconsistent error logs cannot be selected due to the decreasing weight of word consistency.

 \begin{figure}[t]
	\centering
	\subfloat[Error cases with $\lambda$]	{\centering\includegraphics[width = 0.5\linewidth]{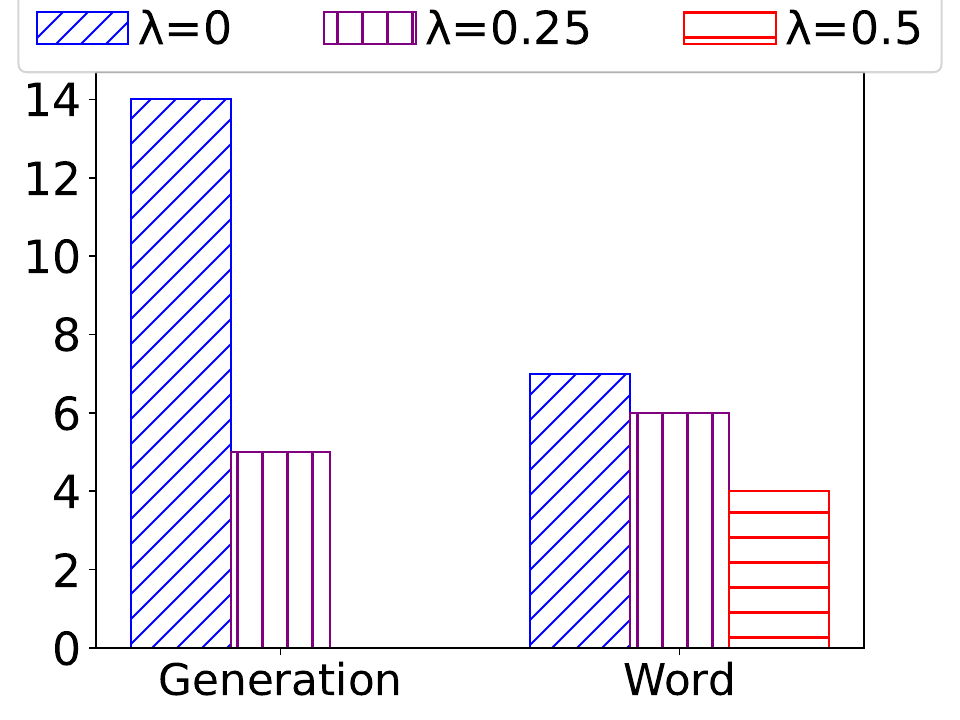}}
	\subfloat[Error cases with $a$]	{\centering\includegraphics[width = 0.5\linewidth]{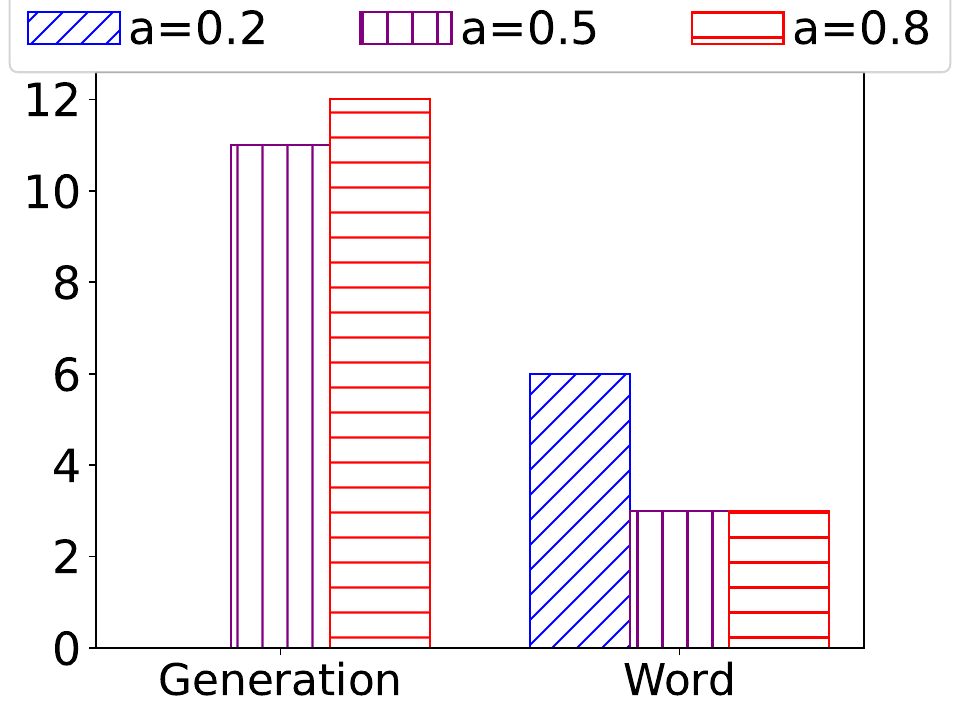}}
	\caption{Confidence score on hallucination errors}
	\label{fig:confidence}
\end{figure}

\section{Conclusion}\label{sec:conclusion}
In this paper, we present LLMLog, an LLM-driven multi-round annotation framework with adaptive in-context learning for log template generation. 
Firstly, we propose a distance metric to measure the log similarity, along with a confidence metric to assess the difficulty faced by LLM. 
Based on the two metrics, we identify the most valuable unlabeled logs for human annotation in each round. Second, we introduce an adaptive approach for selecting demonstrative contexts of each log to generate more accurate templates by LLM. Experimental results demonstrate that LLMLog achieves superior performance compared to the state-of-the-art baselines.

\begin{acks}
	Lei Chen's work is partially supported by National Key Research and Development Program of China Grant No. 2023YFF0725100, National Science Foundation of China (NSFC) under Grant No. U22B2060, Guangdong-Hong Kong Technology Innovation Joint Funding Scheme Project No. 2024A0505040012, the Hong Kong RGC GRF Project 16213620, RIF Project R6020-19, AOE Project AoE/E-603/18, Theme-based project TRS T41-603/20R, CRF Project C2004-21G, Key Areas Special Project of Guangdong Provincial  Universities 2024ZDZX1006,  Guangdong Province Science and Technology Plan Project 2023A0505030011, Guangzhou municipality big data intelligence key lab, 2023A03J0012, Hong Kong ITC ITF grants MHX/078/21 and PRP/004/22FX, Zhujiang scholar program 2021JC02X170, Microsoft Research Asia Collaborative Research Grant, HKUST-Webank joint research lab and 2023 HKUST Shenzhen-Hong Kong Collaborative Innovation Institute Green Sustainability Special Fund, from Shui On Xintiandi and the InnoSpace GBA.
\end{acks}
\balance
\bibliographystyle{ACM-Reference-Format}
\bibliography{sample}


\begin{thebibliography}{88}


\ifx \showCODEN    \undefined \def \showCODEN     #1{\unskip}     \fi
\ifx \showDOI      \undefined \def \showDOI       #1{#1}\fi
\ifx \showISBNx    \undefined \def \showISBNx     #1{\unskip}     \fi
\ifx \showISBNxiii \undefined \def \showISBNxiii  #1{\unskip}     \fi
\ifx \showISSN     \undefined \def \showISSN      #1{\unskip}     \fi
\ifx \showLCCN     \undefined \def \showLCCN      #1{\unskip}     \fi
\ifx \shownote     \undefined \def \shownote      #1{#1}          \fi
\ifx \showarticletitle \undefined \def \showarticletitle #1{#1}   \fi
\ifx \showURL      \undefined \def \showURL       {\relax}        \fi
\providecommand\bibfield[2]{#2}
\providecommand\bibinfo[2]{#2}
\providecommand\natexlab[1]{#1}
\providecommand\showeprint[2][]{arXiv:#2}

\bibitem[\protect\citeauthoryear{Amar and Rigby}{Amar and Rigby}{2019}]%
        {analysis1}
\bibfield{author}{\bibinfo{person}{Anunay Amar} {and} \bibinfo{person}{Peter~C.
  Rigby}.} \bibinfo{year}{2019}\natexlab{}.
\newblock \showarticletitle{Mining historical test logs to predict bugs and
  localize faults in the test logs}. In \bibinfo{booktitle}{\emph{Proceedings
  of the 41st International Conference on Software Engineering}} (Montreal,
  Quebec, Canada) \emph{(\bibinfo{series}{ICSE '19})}. \bibinfo{publisher}{IEEE
  Press}, \bibinfo{pages}{140–151}.
\newblock
\urldef\tempurl%
\url{https://doi.org/10.1109/ICSE.2019.00031}
\showDOI{\tempurl}


\bibitem[\protect\citeauthoryear{{Apache Software Foundation}}{{Apache Software
  Foundation}}{[n.d.]}]%
        {hadoop}
\bibfield{author}{\bibinfo{person}{{Apache Software Foundation}}.}
  \bibinfo{year}{[n.d.]}\natexlab{}.
\newblock \bibinfo{booktitle}{\emph{Hadoop}}.
\newblock
\urldef\tempurl%
\url{https://hadoop.apache.org}
\showURL{%
\tempurl}


\bibitem[\protect\citeauthoryear{Arpaci-Dusseau, Zhou, and Chen}{Arpaci-Dusseau
  et~al\mbox{.}}{2025}]%
        {vldbpattern1}
\bibfield{author}{\bibinfo{person}{Anna Arpaci-Dusseau},
  \bibinfo{person}{Zixiang Zhou}, {and} \bibinfo{person}{Xuhao Chen}.}
  \bibinfo{year}{2025}\natexlab{}.
\newblock \showarticletitle{Accurate and Fast Approximate Graph Pattern Mining
  at Scale}.
\newblock \bibinfo{journal}{\emph{Proc. VLDB Endow.}} \bibinfo{volume}{18},
  \bibinfo{number}{2} (\bibinfo{date}{Feb.} \bibinfo{year}{2025}),
  \bibinfo{pages}{93–107}.
\newblock
\showISSN{2150-8097}
\urldef\tempurl%
\url{https://doi.org/10.14778/3705829.3705831}
\showDOI{\tempurl}


\bibitem[\protect\citeauthoryear{Bagban and Kulkarni}{Bagban and
  Kulkarni}{2020}]%
        {vldbwebtemplate1}
\bibfield{author}{\bibinfo{person}{Tanveer~I. Bagban} {and}
  \bibinfo{person}{Prakash~J. Kulkarni}.} \bibinfo{year}{2020}\natexlab{}.
\newblock \showarticletitle{Template Based Clustering of Web Documents Using
  Locality Sensitive Hashing (LSH)}. In \bibinfo{booktitle}{\emph{Computing in
  Engineering and Technology}}, \bibfield{editor}{\bibinfo{person}{Brijesh
  Iyer}, \bibinfo{person}{P.~S. Deshpande}, \bibinfo{person}{S.~C. Sharma},
  {and} \bibinfo{person}{Ulhas Shiurkar}} (Eds.). \bibinfo{publisher}{Springer
  Singapore}, \bibinfo{address}{Singapore}, \bibinfo{pages}{567--584}.
\newblock
\showISBNx{978-981-32-9515-5}


\bibitem[\protect\citeauthoryear{BehnamGhader, Adlakha, Mosbach, Bahdanau,
  Chapados, and Reddy}{BehnamGhader et~al\mbox{.}}{2024}]%
        {llm2vec}
\bibfield{author}{\bibinfo{person}{Parishad BehnamGhader},
  \bibinfo{person}{Vaibhav Adlakha}, \bibinfo{person}{Marius Mosbach},
  \bibinfo{person}{Dzmitry Bahdanau}, \bibinfo{person}{Nicolas Chapados}, {and}
  \bibinfo{person}{Siva Reddy}.} \bibinfo{year}{2024}\natexlab{}.
\newblock \showarticletitle{{LLM}2Vec: Large Language Models Are Secretly
  Powerful Text Encoders}. In \bibinfo{booktitle}{\emph{First Conference on
  Language Modeling}}.
\newblock
\urldef\tempurl%
\url{https://openreview.net/forum?id=IW1PR7vEBf}
\showURL{%
\tempurl}


\bibitem[\protect\citeauthoryear{Bhowmick, Dragut, and Meng}{Bhowmick
  et~al\mbox{.}}{2023}]%
        {embedding4}
\bibfield{author}{\bibinfo{person}{Satadisha~Saha Bhowmick},
  \bibinfo{person}{Eduard~C. Dragut}, {and} \bibinfo{person}{Weiyi Meng}.}
  \bibinfo{year}{2023}\natexlab{}.
\newblock \showarticletitle{Globally Aware Contextual Embeddings for Named
  Entity Recognition in Social Media Streams}. In
  \bibinfo{booktitle}{\emph{2023 IEEE 39th International Conference on Data
  Engineering (ICDE)}}. \bibinfo{pages}{1544--1557}.
\newblock
\urldef\tempurl%
\url{https://doi.org/10.1109/ICDE55515.2023.00122}
\showDOI{\tempurl}


\bibitem[\protect\citeauthoryear{Bonifati, Martens, and Timm}{Bonifati
  et~al\mbox{.}}{2017}]%
        {vldblogrootcause6}
\bibfield{author}{\bibinfo{person}{Angela Bonifati}, \bibinfo{person}{Wim
  Martens}, {and} \bibinfo{person}{Thomas Timm}.}
  \bibinfo{year}{2017}\natexlab{}.
\newblock \showarticletitle{An analytical study of large SPARQL query logs}.
\newblock \bibinfo{journal}{\emph{Proc. VLDB Endow.}} \bibinfo{volume}{11},
  \bibinfo{number}{2} (\bibinfo{date}{Oct.} \bibinfo{year}{2017}),
  \bibinfo{pages}{149–161}.
\newblock
\showISSN{2150-8097}
\urldef\tempurl%
\url{https://doi.org/10.14778/3149193.3149196}
\showDOI{\tempurl}


\bibitem[\protect\citeauthoryear{Brown, Mann, Ryder, Subbiah, Kaplan, Dhariwal,
  Neelakantan, Shyam, Sastry, Askell, Agarwal, Herbert-Voss, Krueger, Henighan,
  Child, Ramesh, Ziegler, Wu, Winter, Hesse, Chen, Sigler, Litwin, Gray, Chess,
  Clark, Berner, McCandlish, Radford, Sutskever, and Amodei}{Brown
  et~al\mbox{.}}{2020}]%
        {icl2}
\bibfield{author}{\bibinfo{person}{Tom Brown}, \bibinfo{person}{Benjamin Mann},
  \bibinfo{person}{Nick Ryder}, \bibinfo{person}{Melanie Subbiah},
  \bibinfo{person}{Jared~D Kaplan}, \bibinfo{person}{Prafulla Dhariwal},
  \bibinfo{person}{Arvind Neelakantan}, \bibinfo{person}{Pranav Shyam},
  \bibinfo{person}{Girish Sastry}, \bibinfo{person}{Amanda Askell},
  \bibinfo{person}{Sandhini Agarwal}, \bibinfo{person}{Ariel Herbert-Voss},
  \bibinfo{person}{Gretchen Krueger}, \bibinfo{person}{Tom Henighan},
  \bibinfo{person}{Rewon Child}, \bibinfo{person}{Aditya Ramesh},
  \bibinfo{person}{Daniel Ziegler}, \bibinfo{person}{Jeffrey Wu},
  \bibinfo{person}{Clemens Winter}, \bibinfo{person}{Chris Hesse},
  \bibinfo{person}{Mark Chen}, \bibinfo{person}{Eric Sigler},
  \bibinfo{person}{Mateusz Litwin}, \bibinfo{person}{Scott Gray},
  \bibinfo{person}{Benjamin Chess}, \bibinfo{person}{Jack Clark},
  \bibinfo{person}{Christopher Berner}, \bibinfo{person}{Sam McCandlish},
  \bibinfo{person}{Alec Radford}, \bibinfo{person}{Ilya Sutskever}, {and}
  \bibinfo{person}{Dario Amodei}.} \bibinfo{year}{2020}\natexlab{}.
\newblock \showarticletitle{Language Models are Few-Shot Learners}. In
  \bibinfo{booktitle}{\emph{Advances in Neural Information Processing
  Systems}}, \bibfield{editor}{\bibinfo{person}{H.~Larochelle},
  \bibinfo{person}{M.~Ranzato}, \bibinfo{person}{R.~Hadsell},
  \bibinfo{person}{M.F. Balcan}, {and} \bibinfo{person}{H.~Lin}} (Eds.),
  Vol.~\bibinfo{volume}{33}. \bibinfo{publisher}{Curran Associates, Inc.},
  \bibinfo{pages}{1877--1901}.
\newblock
\urldef\tempurl%
\url{https://proceedings.neurips.cc/paper_files/paper/2020/file/1457c0d6bfcb4967418bfb8ac142f64a-Paper.pdf}
\showURL{%
\tempurl}


\bibitem[\protect\citeauthoryear{Chang, Kayed, Girgis, and Shaalan}{Chang
  et~al\mbox{.}}{2006}]%
        {tkdewebtemplate2}
\bibfield{author}{\bibinfo{person}{Chia-Hui Chang}, \bibinfo{person}{M. Kayed},
  \bibinfo{person}{M.R. Girgis}, {and} \bibinfo{person}{K.F. Shaalan}.}
  \bibinfo{year}{2006}\natexlab{}.
\newblock \showarticletitle{A Survey of Web Information Extraction Systems}.
\newblock \bibinfo{journal}{\emph{IEEE Transactions on Knowledge and Data
  Engineering}} \bibinfo{volume}{18}, \bibinfo{number}{10}
  (\bibinfo{year}{2006}), \bibinfo{pages}{1411--1428}.
\newblock
\urldef\tempurl%
\url{https://doi.org/10.1109/TKDE.2006.152}
\showDOI{\tempurl}


\bibitem[\protect\citeauthoryear{Chen, Zhang, and Zhou}{Chen
  et~al\mbox{.}}{2018}]%
        {dpp}
\bibfield{author}{\bibinfo{person}{Laming Chen}, \bibinfo{person}{Guoxin
  Zhang}, {and} \bibinfo{person}{Hanning Zhou}.}
  \bibinfo{year}{2018}\natexlab{}.
\newblock \showarticletitle{Fast greedy MAP inference for determinantal point
  process to improve recommendation diversity}. In
  \bibinfo{booktitle}{\emph{Proceedings of the 32nd International Conference on
  Neural Information Processing Systems}} (Montr\'{e}al, Canada)
  \emph{(\bibinfo{series}{NIPS'18})}. \bibinfo{publisher}{Curran Associates
  Inc.}, \bibinfo{address}{Red Hook, NY, USA}, \bibinfo{pages}{5627–5638}.
\newblock


\bibitem[\protect\citeauthoryear{Cheng, Liao, Huang, Yang, Zhou, Yuan, and
  Wang}{Cheng et~al\mbox{.}}{2024}]%
        {annotation3}
\bibfield{author}{\bibinfo{person}{Yurong Cheng}, \bibinfo{person}{Zhaohe
  Liao}, \bibinfo{person}{Xiaosong Huang}, \bibinfo{person}{Yi Yang},
  \bibinfo{person}{Xiangmin Zhou}, \bibinfo{person}{Ye Yuan}, {and}
  \bibinfo{person}{Guoren Wang}.} \bibinfo{year}{2024}\natexlab{}.
\newblock \showarticletitle{Cross Online Ride-Sharing for Multiple-Platform
  Cooperations in Spatial Crowdsourcing}. In \bibinfo{booktitle}{\emph{2024
  IEEE 40th International Conference on Data Engineering (ICDE)}}.
  \bibinfo{pages}{4140--4152}.
\newblock
\urldef\tempurl%
\url{https://doi.org/10.1109/ICDE60146.2024.00317}
\showDOI{\tempurl}


\bibitem[\protect\citeauthoryear{Christensen and Li}{Christensen and
  Li}{2013}]%
        {sigmodlogretrieval10}
\bibfield{author}{\bibinfo{person}{Robert Christensen} {and}
  \bibinfo{person}{Feifei Li}.} \bibinfo{year}{2013}\natexlab{}.
\newblock \showarticletitle{Adaptive log compression for massive log data}. In
  \bibinfo{booktitle}{\emph{Proceedings of the 2013 ACM SIGMOD International
  Conference on Management of Data}} (New York, New York, USA)
  \emph{(\bibinfo{series}{SIGMOD '13})}. \bibinfo{publisher}{Association for
  Computing Machinery}, \bibinfo{address}{New York, NY, USA},
  \bibinfo{pages}{1283–1284}.
\newblock
\showISBNx{9781450320375}
\urldef\tempurl%
\url{https://doi.org/10.1145/2463676.2465341}
\showDOI{\tempurl}


\bibitem[\protect\citeauthoryear{Chvatal}{Chvatal}{1979}]%
        {setcover}
\bibfield{author}{\bibinfo{person}{V. Chvatal}.}
  \bibinfo{year}{1979}\natexlab{}.
\newblock \showarticletitle{A Greedy Heuristic for the Set-Covering Problem}.
\newblock \bibinfo{journal}{\emph{Math. Oper. Res.}} \bibinfo{volume}{4},
  \bibinfo{number}{3} (\bibinfo{date}{Aug.} \bibinfo{year}{1979}),
  \bibinfo{pages}{233–235}.
\newblock
\showISSN{0364-765X}
\urldef\tempurl%
\url{https://doi.org/10.1287/moor.4.3.233}
\showDOI{\tempurl}


\bibitem[\protect\citeauthoryear{Du, Li, Zheng, and Srikumar}{Du
  et~al\mbox{.}}{2017}]%
        {anamoly1}
\bibfield{author}{\bibinfo{person}{Min Du}, \bibinfo{person}{Feifei Li},
  \bibinfo{person}{Guineng Zheng}, {and} \bibinfo{person}{Vivek Srikumar}.}
  \bibinfo{year}{2017}\natexlab{}.
\newblock \showarticletitle{DeepLog: Anomaly Detection and Diagnosis from
  System Logs through Deep Learning}. In \bibinfo{booktitle}{\emph{Proceedings
  of the 2017 ACM SIGSAC Conference on Computer and Communications Security}}
  (Dallas, Texas, USA) \emph{(\bibinfo{series}{CCS '17})}.
  \bibinfo{publisher}{Association for Computing Machinery},
  \bibinfo{address}{New York, NY, USA}, \bibinfo{pages}{1285–1298}.
\newblock
\showISBNx{9781450349468}
\urldef\tempurl%
\url{https://doi.org/10.1145/3133956.3134015}
\showDOI{\tempurl}


\bibitem[\protect\citeauthoryear{Fan, Li, and Zhou}{Fan et~al\mbox{.}}{2011}]%
        {icdetemplate3}
\bibfield{author}{\bibinfo{person}{Ju Fan}, \bibinfo{person}{Guoliang Li},
  {and} \bibinfo{person}{Lizhu Zhou}.} \bibinfo{year}{2011}\natexlab{}.
\newblock \showarticletitle{Interactive SQL query suggestion: Making databases
  user-friendly}. In \bibinfo{booktitle}{\emph{2011 IEEE 27th International
  Conference on Data Engineering}}. \bibinfo{pages}{351--362}.
\newblock
\urldef\tempurl%
\url{https://doi.org/10.1109/ICDE.2011.5767843}
\showDOI{\tempurl}


\bibitem[\protect\citeauthoryear{Fan, Han, Fan, Chai, Tang, Li, and Du}{Fan
  et~al\mbox{.}}{2024}]%
        {icl1}
\bibfield{author}{\bibinfo{person}{Meihao Fan}, \bibinfo{person}{Xiaoyue Han},
  \bibinfo{person}{Ju Fan}, \bibinfo{person}{Chengliang Chai},
  \bibinfo{person}{Nan Tang}, \bibinfo{person}{Guoliang Li}, {and}
  \bibinfo{person}{Xiaoyong Du}.} \bibinfo{year}{2024}\natexlab{}.
\newblock \showarticletitle{{ Cost-Effective In-Context Learning for Entity
  Resolution: A Design Space Exploration }}. In \bibinfo{booktitle}{\emph{2024
  IEEE 40th International Conference on Data Engineering (ICDE)}}.
  \bibinfo{publisher}{IEEE Computer Society}, \bibinfo{address}{Los Alamitos,
  CA, USA}, \bibinfo{pages}{3696--3709}.
\newblock
\urldef\tempurl%
\url{https://doi.org/10.1109/ICDE60146.2024.00284}
\showDOI{\tempurl}


\bibitem[\protect\citeauthoryear{Feng, Li, and Chen}{Feng
  et~al\mbox{.}}{2024}]%
        {icdetemplate5}
\bibfield{author}{\bibinfo{person}{Jieming Feng}, \bibinfo{person}{Zhanhuai
  Li}, {and} \bibinfo{person}{Qun Chen}.} \bibinfo{year}{2024}\natexlab{}.
\newblock \showarticletitle{Towards Exploratory Query Optimization for
  Template-Based SQL Workloads}. In \bibinfo{booktitle}{\emph{2024 IEEE 40th
  International Conference on Data Engineering (ICDE)}}.
  \bibinfo{pages}{151--164}.
\newblock
\urldef\tempurl%
\url{https://doi.org/10.1109/ICDE60146.2024.00019}
\showDOI{\tempurl}


\bibitem[\protect\citeauthoryear{Fernandez, Elmore, Franklin, Krishnan, and
  Tan}{Fernandez et~al\mbox{.}}{2023}]%
        {llm1}
\bibfield{author}{\bibinfo{person}{Raul~Castro Fernandez},
  \bibinfo{person}{Aaron~J. Elmore}, \bibinfo{person}{Michael~J. Franklin},
  \bibinfo{person}{Sanjay Krishnan}, {and} \bibinfo{person}{Chenhao Tan}.}
  \bibinfo{year}{2023}\natexlab{}.
\newblock \showarticletitle{How Large Language Models Will Disrupt Data
  Management}.
\newblock \bibinfo{journal}{\emph{Proc. VLDB Endow.}} \bibinfo{volume}{16},
  \bibinfo{number}{11} (\bibinfo{date}{July} \bibinfo{year}{2023}),
  \bibinfo{pages}{3302–3309}.
\newblock
\showISSN{2150-8097}
\urldef\tempurl%
\url{https://doi.org/10.14778/3611479.3611527}
\showDOI{\tempurl}


\bibitem[\protect\citeauthoryear{Feuer, Liu, Hegde, and Freire}{Feuer
  et~al\mbox{.}}{2024}]%
        {llm2}
\bibfield{author}{\bibinfo{person}{Benjamin Feuer}, \bibinfo{person}{Yurong
  Liu}, \bibinfo{person}{Chinmay Hegde}, {and} \bibinfo{person}{Juliana
  Freire}.} \bibinfo{year}{2024}\natexlab{}.
\newblock \showarticletitle{ArcheType: A Novel Framework for Open-Source Column
  Type Annotation Using Large Language Models}.
\newblock \bibinfo{journal}{\emph{Proc. VLDB Endow.}} \bibinfo{volume}{17},
  \bibinfo{number}{9} (\bibinfo{date}{Aug.} \bibinfo{year}{2024}),
  \bibinfo{pages}{2279–2292}.
\newblock
\showISSN{2150-8097}
\urldef\tempurl%
\url{https://doi.org/10.14778/3665844.3665857}
\showDOI{\tempurl}


\bibitem[\protect\citeauthoryear{Gao, Wang, Li, Sun, Qian, Ding, and Zhou}{Gao
  et~al\mbox{.}}{2024}]%
        {llm4}
\bibfield{author}{\bibinfo{person}{Dawei Gao}, \bibinfo{person}{Haibin Wang},
  \bibinfo{person}{Yaliang Li}, \bibinfo{person}{Xiuyu Sun},
  \bibinfo{person}{Yichen Qian}, \bibinfo{person}{Bolin Ding}, {and}
  \bibinfo{person}{Jingren Zhou}.} \bibinfo{year}{2024}\natexlab{}.
\newblock \showarticletitle{Text-to-SQL Empowered by Large Language Models: A
  Benchmark Evaluation}.
\newblock \bibinfo{journal}{\emph{Proc. VLDB Endow.}} \bibinfo{volume}{17},
  \bibinfo{number}{5} (\bibinfo{date}{May} \bibinfo{year}{2024}),
  \bibinfo{pages}{1132–1145}.
\newblock
\showISSN{2150-8097}
\urldef\tempurl%
\url{https://doi.org/10.14778/3641204.3641221}
\showDOI{\tempurl}


\bibitem[\protect\citeauthoryear{Golovin and Krause}{Golovin and
  Krause}{2011}]%
        {setcoverproof}
\bibfield{author}{\bibinfo{person}{Daniel Golovin} {and}
  \bibinfo{person}{Andreas Krause}.} \bibinfo{year}{2011}\natexlab{}.
\newblock \showarticletitle{Adaptive submodularity: theory and applications in
  active learning and stochastic optimization}.
\newblock \bibinfo{journal}{\emph{J. Artif. Int. Res.}} \bibinfo{volume}{42},
  \bibinfo{number}{1} (\bibinfo{date}{Sept.} \bibinfo{year}{2011}),
  \bibinfo{pages}{427–486}.
\newblock
\showISSN{1076-9757}


\bibitem[\protect\citeauthoryear{Hamooni, Debnath, Xu, Zhang, Jiang, and
  Mueen}{Hamooni et~al\mbox{.}}{2016}]%
        {logmine}
\bibfield{author}{\bibinfo{person}{Hossein Hamooni}, \bibinfo{person}{Biplob
  Debnath}, \bibinfo{person}{Jianwu Xu}, \bibinfo{person}{Hui Zhang},
  \bibinfo{person}{Guofei Jiang}, {and} \bibinfo{person}{Abdullah Mueen}.}
  \bibinfo{year}{2016}\natexlab{}.
\newblock \showarticletitle{LogMine: Fast Pattern Recognition for Log
  Analytics}. In \bibinfo{booktitle}{\emph{Proceedings of the 25th ACM
  International on Conference on Information and Knowledge Management}}
  (Indianapolis, Indiana, USA) \emph{(\bibinfo{series}{CIKM '16})}.
  \bibinfo{publisher}{Association for Computing Machinery},
  \bibinfo{address}{New York, NY, USA}, \bibinfo{pages}{1573–1582}.
\newblock
\showISBNx{9781450340731}
\urldef\tempurl%
\url{https://doi.org/10.1145/2983323.2983358}
\showDOI{\tempurl}


\bibitem[\protect\citeauthoryear{He, Zhu, Zheng, and Lyu}{He
  et~al\mbox{.}}{2017}]%
        {drain}
\bibfield{author}{\bibinfo{person}{Pinjia He}, \bibinfo{person}{Jieming Zhu},
  \bibinfo{person}{Zibin Zheng}, {and} \bibinfo{person}{Michael~R. Lyu}.}
  \bibinfo{year}{2017}\natexlab{}.
\newblock \showarticletitle{Drain: An Online Log Parsing Approach with Fixed
  Depth Tree}. In \bibinfo{booktitle}{\emph{2017 IEEE International Conference
  on Web Services (ICWS)}}. \bibinfo{pages}{33--40}.
\newblock
\urldef\tempurl%
\url{https://doi.org/10.1109/ICWS.2017.13}
\showDOI{\tempurl}


\bibitem[\protect\citeauthoryear{He, Lin, Lou, Zhang, Lyu, and Zhang}{He
  et~al\mbox{.}}{2018}]%
        {analysis2}
\bibfield{author}{\bibinfo{person}{Shilin He}, \bibinfo{person}{Qingwei Lin},
  \bibinfo{person}{Jian-Guang Lou}, \bibinfo{person}{Hongyu Zhang},
  \bibinfo{person}{Michael~R. Lyu}, {and} \bibinfo{person}{Dongmei Zhang}.}
  \bibinfo{year}{2018}\natexlab{}.
\newblock \showarticletitle{Identifying impactful service system problems via
  log analysis}. In \bibinfo{booktitle}{\emph{Proceedings of the 2018 26th ACM
  Joint Meeting on European Software Engineering Conference and Symposium on
  the Foundations of Software Engineering}} (Lake Buena Vista, FL, USA)
  \emph{(\bibinfo{series}{ESEC/FSE 2018})}. \bibinfo{publisher}{Association for
  Computing Machinery}, \bibinfo{address}{New York, NY, USA},
  \bibinfo{pages}{60–70}.
\newblock
\showISBNx{9781450355735}
\urldef\tempurl%
\url{https://doi.org/10.1145/3236024.3236083}
\showDOI{\tempurl}


\bibitem[\protect\citeauthoryear{He, Zhu, He, and Lyu}{He
  et~al\mbox{.}}{2016}]%
        {anomaly2}
\bibfield{author}{\bibinfo{person}{Shilin He}, \bibinfo{person}{Jieming Zhu},
  \bibinfo{person}{Pinjia He}, {and} \bibinfo{person}{Michael~R. Lyu}.}
  \bibinfo{year}{2016}\natexlab{}.
\newblock \showarticletitle{Experience Report: System Log Analysis for Anomaly
  Detection}. In \bibinfo{booktitle}{\emph{2016 IEEE 27th International
  Symposium on Software Reliability Engineering (ISSRE)}}.
  \bibinfo{pages}{207--218}.
\newblock
\urldef\tempurl%
\url{https://doi.org/10.1109/ISSRE.2016.21}
\showDOI{\tempurl}


\bibitem[\protect\citeauthoryear{Huang, He, Liang, Jiang, Xiao, and Chen}{Huang
  et~al\mbox{.}}{2024}]%
        {icl4}
\bibfield{author}{\bibinfo{person}{Yuncheng Huang}, \bibinfo{person}{Qianyu
  He}, \bibinfo{person}{Jiaqing Liang}, \bibinfo{person}{Sihang Jiang},
  \bibinfo{person}{Yanghua Xiao}, {and} \bibinfo{person}{Yunwen Chen}.}
  \bibinfo{year}{2024}\natexlab{}.
\newblock \showarticletitle{{ Enhancing Quantitative Reasoning Skills of Large
  Language Models through Dimension Perception }}. In
  \bibinfo{booktitle}{\emph{2024 IEEE 40th International Conference on Data
  Engineering (ICDE)}}. \bibinfo{publisher}{IEEE Computer Society},
  \bibinfo{address}{Los Alamitos, CA, USA}, \bibinfo{pages}{789--802}.
\newblock
\urldef\tempurl%
\url{https://doi.org/10.1109/ICDE60146.2024.00066}
\showDOI{\tempurl}


\bibitem[\protect\citeauthoryear{Jia, Wang, Zhao, Yuan, Tao, and Guan}{Jia
  et~al\mbox{.}}{2021}]%
        {icdelog18}
\bibfield{author}{\bibinfo{person}{Peng Jia}, \bibinfo{person}{Pinghui Wang},
  \bibinfo{person}{Junzhou Zhao}, \bibinfo{person}{Ye Yuan},
  \bibinfo{person}{Jing Tao}, {and} \bibinfo{person}{Xiaohong Guan}.}
  \bibinfo{year}{2021}\natexlab{}.
\newblock \showarticletitle{LogLog Filter: Filtering Cold Items within a Large
  Range over High Speed Data Streams}. In \bibinfo{booktitle}{\emph{2021 IEEE
  37th International Conference on Data Engineering (ICDE)}}.
  \bibinfo{pages}{804--815}.
\newblock
\urldef\tempurl%
\url{https://doi.org/10.1109/ICDE51399.2021.00075}
\showDOI{\tempurl}


\bibitem[\protect\citeauthoryear{Jiang, Zeller, Waleffe, Hoefler, and
  Alonso}{Jiang et~al\mbox{.}}{2024}]%
        {vldbicl2}
\bibfield{author}{\bibinfo{person}{Wenqi Jiang}, \bibinfo{person}{Marco
  Zeller}, \bibinfo{person}{Roger Waleffe}, \bibinfo{person}{Torsten Hoefler},
  {and} \bibinfo{person}{Gustavo Alonso}.} \bibinfo{year}{2024}\natexlab{}.
\newblock \showarticletitle{Chameleon: a Heterogeneous and Disaggregated
  Accelerator System for Retrieval-Augmented Language Models}.
\newblock \bibinfo{journal}{\emph{Proc. {VLDB} Endow.}} \bibinfo{volume}{18},
  \bibinfo{number}{1} (\bibinfo{year}{2024}), \bibinfo{pages}{42--52}.
\newblock
\urldef\tempurl%
\url{https://www.vldb.org/pvldb/vol18/p42-jiang.pdf}
\showURL{%
\tempurl}


\bibitem[\protect\citeauthoryear{Johnson, Pandis, Stoica, Athanassoulis, and
  Ailamaki}{Johnson et~al\mbox{.}}{2010}]%
        {sigmodlogretrieval12}
\bibfield{author}{\bibinfo{person}{Ryan Johnson}, \bibinfo{person}{Ippokratis
  Pandis}, \bibinfo{person}{Radu Stoica}, \bibinfo{person}{Manos
  Athanassoulis}, {and} \bibinfo{person}{Anastasia Ailamaki}.}
  \bibinfo{year}{2010}\natexlab{}.
\newblock \showarticletitle{Aether: a scalable approach to logging}.
\newblock \bibinfo{journal}{\emph{Proc. VLDB Endow.}} \bibinfo{volume}{3},
  \bibinfo{number}{1–2} (\bibinfo{date}{Sept.} \bibinfo{year}{2010}),
  \bibinfo{pages}{681–692}.
\newblock
\showISSN{2150-8097}
\urldef\tempurl%
\url{https://doi.org/10.14778/1920841.1920928}
\showDOI{\tempurl}


\bibitem[\protect\citeauthoryear{Kao, Lin, Ho, and Chen}{Kao
  et~al\mbox{.}}{2004}]%
        {tkdetemplate1}
\bibfield{author}{\bibinfo{person}{Hung-Yu Kao}, \bibinfo{person}{Shian-Hua
  Lin}, \bibinfo{person}{Jan-Ming Ho}, {and} \bibinfo{person}{Ming-Syan Chen}.}
  \bibinfo{year}{2004}\natexlab{}.
\newblock \showarticletitle{Mining Web informative structures and contents
  based on entropy analysis}.
\newblock \bibinfo{journal}{\emph{IEEE Transactions on Knowledge and Data
  Engineering}} \bibinfo{volume}{16}, \bibinfo{number}{1}
  (\bibinfo{year}{2004}), \bibinfo{pages}{41--55}.
\newblock
\urldef\tempurl%
\url{https://doi.org/10.1109/TKDE.2004.1264821}
\showDOI{\tempurl}


\bibitem[\protect\citeauthoryear{Khan, Shin, Bianculli, and Briand}{Khan
  et~al\mbox{.}}{2022}]%
        {metric}
\bibfield{author}{\bibinfo{person}{Zanis~Ali Khan}, \bibinfo{person}{Donghwan
  Shin}, \bibinfo{person}{Domenico Bianculli}, {and} \bibinfo{person}{Lionel
  Briand}.} \bibinfo{year}{2022}\natexlab{}.
\newblock \showarticletitle{Guidelines for assessing the accuracy of log
  message template identification techniques}. In
  \bibinfo{booktitle}{\emph{Proceedings of the 44th International Conference on
  Software Engineering}} (Pittsburgh, Pennsylvania)
  \emph{(\bibinfo{series}{ICSE '22})}. \bibinfo{publisher}{Association for
  Computing Machinery}, \bibinfo{address}{New York, NY, USA},
  \bibinfo{pages}{1095–1106}.
\newblock
\showISBNx{9781450392211}
\urldef\tempurl%
\url{https://doi.org/10.1145/3510003.3510101}
\showDOI{\tempurl}


\bibitem[\protect\citeauthoryear{Khuller, Moss, and Naor}{Khuller
  et~al\mbox{.}}{1999}]%
        {maxcover}
\bibfield{author}{\bibinfo{person}{Samir Khuller}, \bibinfo{person}{Anna Moss},
  {and} \bibinfo{person}{Joseph~(Seffi) Naor}.}
  \bibinfo{year}{1999}\natexlab{}.
\newblock \showarticletitle{The budgeted maximum coverage problem}.
\newblock \bibinfo{journal}{\emph{Inf. Process. Lett.}} \bibinfo{volume}{70},
  \bibinfo{number}{1} (\bibinfo{date}{April} \bibinfo{year}{1999}),
  \bibinfo{pages}{39–45}.
\newblock
\showISSN{0020-0190}
\urldef\tempurl%
\url{https://doi.org/10.1016/S0020-0190(99)00031-9}
\showDOI{\tempurl}


\bibitem[\protect\citeauthoryear{Le and Zhang}{Le and Zhang}{2023}]%
        {logppt}
\bibfield{author}{\bibinfo{person}{Van-Hoang Le} {and} \bibinfo{person}{Hongyu
  Zhang}.} \bibinfo{year}{2023}\natexlab{}.
\newblock \showarticletitle{Log Parsing with Prompt-Based Few-Shot Learning}.
  In \bibinfo{booktitle}{\emph{Proceedings of the 45th International Conference
  on Software Engineering}} (Melbourne, Victoria, Australia)
  \emph{(\bibinfo{series}{ICSE '23})}. \bibinfo{publisher}{IEEE Press},
  \bibinfo{pages}{2438–2449}.
\newblock
\showISBNx{9781665457019}
\urldef\tempurl%
\url{https://doi.org/10.1109/ICSE48619.2023.00204}
\showDOI{\tempurl}


\bibitem[\protect\citeauthoryear{Li, Zhou, and Zhao}{Li et~al\mbox{.}}{2024c}]%
        {li2024llm}
\bibfield{author}{\bibinfo{person}{Guoliang Li}, \bibinfo{person}{Xuanhe Zhou},
  {and} \bibinfo{person}{Xinyang Zhao}.} \bibinfo{year}{2024}\natexlab{c}.
\newblock \showarticletitle{LLM for Data Management}.
\newblock \bibinfo{journal}{\emph{Proceedings of the VLDB Endowment}}
  \bibinfo{volume}{17}, \bibinfo{number}{12} (\bibinfo{year}{2024}),
  \bibinfo{pages}{4213--4216}.
\newblock


\bibitem[\protect\citeauthoryear{Li, Li, Tian, Tang, Xu, Chen, Hu, Dong, Li,
  and Chen}{Li et~al\mbox{.}}{2024a}]%
        {li2024survey}
\bibfield{author}{\bibinfo{person}{Haoyang Li}, \bibinfo{person}{Yiming Li},
  \bibinfo{person}{Anxin Tian}, \bibinfo{person}{Tianhao Tang},
  \bibinfo{person}{Zhanchao Xu}, \bibinfo{person}{Xuejia Chen},
  \bibinfo{person}{Nicole Hu}, \bibinfo{person}{Wei Dong},
  \bibinfo{person}{Qing Li}, {and} \bibinfo{person}{Lei Chen}.}
  \bibinfo{year}{2024}\natexlab{a}.
\newblock \showarticletitle{A survey on large language model acceleration based
  on kv cache management}.
\newblock \bibinfo{journal}{\emph{arXiv preprint arXiv:2412.19442}}
  (\bibinfo{year}{2024}).
\newblock


\bibitem[\protect\citeauthoryear{Li, Chen, Jing, He, and Yu}{Li
  et~al\mbox{.}}{2020}]%
        {swisslog}
\bibfield{author}{\bibinfo{person}{Xiaoyun Li}, \bibinfo{person}{Pengfei Chen},
  \bibinfo{person}{Linxiao Jing}, \bibinfo{person}{Zilong He}, {and}
  \bibinfo{person}{Guangba Yu}.} \bibinfo{year}{2020}\natexlab{}.
\newblock \showarticletitle{SwissLog: Robust and Unified Deep Learning Based
  Log Anomaly Detection for Diverse Faults}. \bibinfo{pages}{92--103}.
\newblock
\urldef\tempurl%
\url{https://doi.org/10.1109/ISSRE5003.2020.00018}
\showDOI{\tempurl}


\bibitem[\protect\citeauthoryear{Li, Yuan, Wang, Cong, and Bing}{Li
  et~al\mbox{.}}{2024b}]%
        {vldbicl1}
\bibfield{author}{\bibinfo{person}{Zhaodonghui Li}, \bibinfo{person}{Haitao
  Yuan}, \bibinfo{person}{Huiming Wang}, \bibinfo{person}{Gao Cong}, {and}
  \bibinfo{person}{Lidong Bing}.} \bibinfo{year}{2024}\natexlab{b}.
\newblock \showarticletitle{{LLM-R2:} {A} Large Language Model Enhanced
  Rule-based Rewrite System for Boosting Query Efficiency}.
\newblock \bibinfo{journal}{\emph{Proc. {VLDB} Endow.}} \bibinfo{volume}{18},
  \bibinfo{number}{1} (\bibinfo{year}{2024}), \bibinfo{pages}{53--65}.
\newblock
\urldef\tempurl%
\url{https://www.vldb.org/pvldb/vol18/p53-yuan.pdf}
\showURL{%
\tempurl}


\bibitem[\protect\citeauthoryear{Lin, Zhang, Lou, Zhang, and Chen}{Lin
  et~al\mbox{.}}{2016}]%
        {analysis3}
\bibfield{author}{\bibinfo{person}{Qingwei Lin}, \bibinfo{person}{Hongyu
  Zhang}, \bibinfo{person}{Jian-Guang Lou}, \bibinfo{person}{Yu Zhang}, {and}
  \bibinfo{person}{Xuewei Chen}.} \bibinfo{year}{2016}\natexlab{}.
\newblock \showarticletitle{Log clustering based problem identification for
  online service systems}. In \bibinfo{booktitle}{\emph{Proceedings of the 38th
  International Conference on Software Engineering Companion}} (Austin, Texas)
  \emph{(\bibinfo{series}{ICSE '16})}. \bibinfo{publisher}{Association for
  Computing Machinery}, \bibinfo{address}{New York, NY, USA},
  \bibinfo{pages}{102–111}.
\newblock
\showISBNx{9781450342056}
\urldef\tempurl%
\url{https://doi.org/10.1145/2889160.2889232}
\showDOI{\tempurl}


\bibitem[\protect\citeauthoryear{Liu, Yuan, Fu, Jiang, Hayashi, and Neubig}{Liu
  et~al\mbox{.}}{2023}]%
        {icl}
\bibfield{author}{\bibinfo{person}{Pengfei Liu}, \bibinfo{person}{Weizhe Yuan},
  \bibinfo{person}{Jinlan Fu}, \bibinfo{person}{Zhengbao Jiang},
  \bibinfo{person}{Hiroaki Hayashi}, {and} \bibinfo{person}{Graham Neubig}.}
  \bibinfo{year}{2023}\natexlab{}.
\newblock \showarticletitle{Pre-train, Prompt, and Predict: A Systematic Survey
  of Prompting Methods in Natural Language Processing}.
\newblock \bibinfo{journal}{\emph{ACM Comput. Surv.}} \bibinfo{volume}{55},
  \bibinfo{number}{9}, Article \bibinfo{articleno}{195} (\bibinfo{date}{Jan.}
  \bibinfo{year}{2023}), \bibinfo{numpages}{35}~pages.
\newblock
\showISSN{0360-0300}
\urldef\tempurl%
\url{https://doi.org/10.1145/3560815}
\showDOI{\tempurl}


\bibitem[\protect\citeauthoryear{Liu, Wu, Zhou, Chen, Wang, Liu, and Wan}{Liu
  et~al\mbox{.}}{2024}]%
        {liu2024enhancing}
\bibfield{author}{\bibinfo{person}{Xinfu Liu}, \bibinfo{person}{Yirui Wu},
  \bibinfo{person}{Yuting Zhou}, \bibinfo{person}{Junyang Chen},
  \bibinfo{person}{Huan Wang}, \bibinfo{person}{Ye Liu}, {and}
  \bibinfo{person}{Shaohua Wan}.} \bibinfo{year}{2024}\natexlab{}.
\newblock \showarticletitle{Enhancing Large Language Models with Multimodality
  and Knowledge Graphs for Hallucination-free Open-set Object Recognition}.
\newblock \bibinfo{journal}{\emph{Proceedings of the VLDB Endowment. ISSN}}
  \bibinfo{volume}{2150} (\bibinfo{year}{2024}), \bibinfo{pages}{8097}.
\newblock


\bibitem[\protect\citeauthoryear{Liu, Zhang, He, Zhang, Li, Kang, Xu, Ma, Lin,
  Dang, Rajmohan, and Zhang}{Liu et~al\mbox{.}}{2022}]%
        {wwwlog}
\bibfield{author}{\bibinfo{person}{Yudong Liu}, \bibinfo{person}{Xu Zhang},
  \bibinfo{person}{Shilin He}, \bibinfo{person}{Hongyu Zhang},
  \bibinfo{person}{Liqun Li}, \bibinfo{person}{Yu Kang}, \bibinfo{person}{Yong
  Xu}, \bibinfo{person}{Minghua Ma}, \bibinfo{person}{Qingwei Lin},
  \bibinfo{person}{Yingnong Dang}, \bibinfo{person}{S. Rajmohan}, {and}
  \bibinfo{person}{Dongmei Zhang}.} \bibinfo{year}{2022}\natexlab{}.
\newblock \showarticletitle{UniParser: A Unified Log Parser for Heterogeneous
  Log Data}.
\newblock \bibinfo{journal}{\emph{Proceedings of the ACM Web Conference 2022}}
  (\bibinfo{year}{2022}).
\newblock
\urldef\tempurl%
\url{https://api.semanticscholar.org/CorpusID:246822534}
\showURL{%
\tempurl}


\bibitem[\protect\citeauthoryear{Makanju, Zincir-Heywood, and Milios}{Makanju
  et~al\mbox{.}}{2012}]%
        {tkdelog15}
\bibfield{author}{\bibinfo{person}{Adetokunbo Makanju}, \bibinfo{person}{A.~Nur
  Zincir-Heywood}, {and} \bibinfo{person}{Evangelos~E. Milios}.}
  \bibinfo{year}{2012}\natexlab{}.
\newblock \showarticletitle{A Lightweight Algorithm for Message Type Extraction
  in System Application Logs}.
\newblock \bibinfo{journal}{\emph{IEEE Transactions on Knowledge and Data
  Engineering}} \bibinfo{volume}{24}, \bibinfo{number}{11}
  (\bibinfo{year}{2012}), \bibinfo{pages}{1921--1936}.
\newblock
\urldef\tempurl%
\url{https://doi.org/10.1109/TKDE.2011.138}
\showDOI{\tempurl}


\bibitem[\protect\citeauthoryear{Markakis, Youngmann, Gao, Zhang, Shahout,
  Chen, Liu, Sabek, and Cafarella}{Markakis et~al\mbox{.}}{2025}]%
        {vldblog2}
\bibfield{author}{\bibinfo{person}{Markos Markakis}, \bibinfo{person}{Brit
  Youngmann}, \bibinfo{person}{Trinity Gao}, \bibinfo{person}{Ziyu Zhang},
  \bibinfo{person}{Rana Shahout}, \bibinfo{person}{Peter~Baile Chen},
  \bibinfo{person}{Chunwei Liu}, \bibinfo{person}{Ibrahim Sabek}, {and}
  \bibinfo{person}{Michael Cafarella}.} \bibinfo{year}{2025}\natexlab{}.
\newblock \showarticletitle{From Logs to Causal Inference: Diagnosing Large
  Systems}.
\newblock \bibinfo{journal}{\emph{Proc. VLDB Endow.}} \bibinfo{volume}{18},
  \bibinfo{number}{2} (\bibinfo{date}{Feb.} \bibinfo{year}{2025}),
  \bibinfo{pages}{158–172}.
\newblock
\showISSN{2150-8097}
\urldef\tempurl%
\url{https://doi.org/10.14778/3705829.3705836}
\showDOI{\tempurl}


\bibitem[\protect\citeauthoryear{Martens, Niewerth, Popp, Rojas, Vansummeren,
  and Vrgo\v{c}}{Martens et~al\mbox{.}}{2023}]%
        {vldbpattern2}
\bibfield{author}{\bibinfo{person}{Wim Martens}, \bibinfo{person}{Matthias
  Niewerth}, \bibinfo{person}{Tina Popp}, \bibinfo{person}{Carlos Rojas},
  \bibinfo{person}{Stijn Vansummeren}, {and} \bibinfo{person}{Domagoj
  Vrgo\v{c}}.} \bibinfo{year}{2023}\natexlab{}.
\newblock \showarticletitle{Representing Paths in Graph Database Pattern
  Matching}.
\newblock \bibinfo{journal}{\emph{Proc. VLDB Endow.}} \bibinfo{volume}{16},
  \bibinfo{number}{7} (\bibinfo{date}{March} \bibinfo{year}{2023}),
  \bibinfo{pages}{1790–1803}.
\newblock
\showISSN{2150-8097}
\urldef\tempurl%
\url{https://doi.org/10.14778/3587136.3587151}
\showDOI{\tempurl}


\bibitem[\protect\citeauthoryear{Mavromatis, Srinivasan, Shen, Zhang, Rangwala,
  Faloutsos, and Karypis}{Mavromatis et~al\mbox{.}}{2023}]%
        {adaicl}
\bibfield{author}{\bibinfo{person}{Costas Mavromatis},
  \bibinfo{person}{Balasubramaniam Srinivasan}, \bibinfo{person}{Zhengyuan
  Shen}, \bibinfo{person}{Jiani Zhang}, \bibinfo{person}{Huzefa Rangwala},
  \bibinfo{person}{Christos Faloutsos}, {and} \bibinfo{person}{George
  Karypis}.} \bibinfo{year}{2023}\natexlab{}.
\newblock \bibinfo{title}{Which Examples to Annotate for In-Context Learning?
  Towards Effective and Efficient Selection}.
\newblock
\newblock
\showeprint[arxiv]{2310.20046}~[cs.CL]
\urldef\tempurl%
\url{https://arxiv.org/abs/2310.20046}
\showURL{%
\tempurl}


\bibitem[\protect\citeauthoryear{Mei, Mao, and Wen}{Mei et~al\mbox{.}}{2024}]%
        {embedding1}
\bibfield{author}{\bibinfo{person}{Lang Mei}, \bibinfo{person}{Jiaxin Mao},
  {and} \bibinfo{person}{Ji-Rong Wen}.} \bibinfo{year}{2024}\natexlab{}.
\newblock \showarticletitle{{ Optimizing Probabilistic Box Embeddings with
  Distance Measures }}. In \bibinfo{booktitle}{\emph{2024 IEEE 40th
  International Conference on Data Engineering (ICDE)}}.
  \bibinfo{publisher}{IEEE Computer Society}, \bibinfo{address}{Los Alamitos,
  CA, USA}, \bibinfo{pages}{5088--5100}.
\newblock
\urldef\tempurl%
\url{https://doi.org/10.1109/ICDE60146.2024.00106}
\showDOI{\tempurl}


\bibitem[\protect\citeauthoryear{Mizutani}{Mizutani}{2013}]%
        {shio}
\bibfield{author}{\bibinfo{person}{Masayoshi Mizutani}.}
  \bibinfo{year}{2013}\natexlab{}.
\newblock \showarticletitle{Incremental Mining of System Log Format}. In
  \bibinfo{booktitle}{\emph{2013 IEEE International Conference on Services
  Computing}}. \bibinfo{pages}{595--602}.
\newblock
\urldef\tempurl%
\url{https://doi.org/10.1109/SCC.2013.73}
\showDOI{\tempurl}


\bibitem[\protect\citeauthoryear{Nandi, Mandal, Atreja, Dasgupta, and
  Bhattacharya}{Nandi et~al\mbox{.}}{2016}]%
        {kddtemplate2}
\bibfield{author}{\bibinfo{person}{Animesh Nandi}, \bibinfo{person}{Atri
  Mandal}, \bibinfo{person}{Shubham Atreja}, \bibinfo{person}{Gargi~B.
  Dasgupta}, {and} \bibinfo{person}{Subhrajit Bhattacharya}.}
  \bibinfo{year}{2016}\natexlab{}.
\newblock \showarticletitle{Anomaly Detection Using Program Control Flow Graph
  Mining From Execution Logs}. In \bibinfo{booktitle}{\emph{Proceedings of the
  22nd ACM SIGKDD International Conference on Knowledge Discovery and Data
  Mining}} (San Francisco, California, USA) \emph{(\bibinfo{series}{KDD '16})}.
  \bibinfo{publisher}{Association for Computing Machinery},
  \bibinfo{address}{New York, NY, USA}, \bibinfo{pages}{215–224}.
\newblock
\showISBNx{9781450342322}
\urldef\tempurl%
\url{https://doi.org/10.1145/2939672.2939712}
\showDOI{\tempurl}


\bibitem[\protect\citeauthoryear{Nedelkoski, Bogatinovski, Acker, Cardoso, and
  Kao}{Nedelkoski et~al\mbox{.}}{2021}]%
        {nulog}
\bibfield{author}{\bibinfo{person}{Sasho Nedelkoski}, \bibinfo{person}{Jasmin
  Bogatinovski}, \bibinfo{person}{Alexander Acker}, \bibinfo{person}{Jorge
  Cardoso}, {and} \bibinfo{person}{Odej Kao}.} \bibinfo{year}{2021}\natexlab{}.
\newblock \showarticletitle{Self-supervised Log Parsing}. In
  \bibinfo{booktitle}{\emph{Machine Learning and Knowledge Discovery in
  Databases: Applied Data Science Track}},
  \bibfield{editor}{\bibinfo{person}{Yuxiao Dong}, \bibinfo{person}{Dunja
  Mladeni{\'{c}}}, {and} \bibinfo{person}{Craig Saunders}} (Eds.).
  \bibinfo{publisher}{Springer International Publishing},
  \bibinfo{address}{Cham}, \bibinfo{pages}{122--138}.
\newblock
\showISBNx{978-3-030-67667-4}


\bibitem[\protect\citeauthoryear{Omar, Dhall, Kalnis, and Mansour}{Omar
  et~al\mbox{.}}{2023}]%
        {qa2}
\bibfield{author}{\bibinfo{person}{Reham Omar}, \bibinfo{person}{Ishika Dhall},
  \bibinfo{person}{Panos Kalnis}, {and} \bibinfo{person}{Essam Mansour}.}
  \bibinfo{year}{2023}\natexlab{}.
\newblock \showarticletitle{A Universal Question-Answering Platform for
  Knowledge Graphs}.
\newblock \bibinfo{journal}{\emph{Proceedings of the ACM on Management of
  Data}}  \bibinfo{volume}{1} (\bibinfo{year}{2023}), \bibinfo{pages}{1 -- 25}.
\newblock
\urldef\tempurl%
\url{https://api.semanticscholar.org/CorpusID:257254920}
\showURL{%
\tempurl}


\bibitem[\protect\citeauthoryear{OpenAI}{OpenAI}{2024}]%
        {gpt4o}
\bibfield{author}{\bibinfo{person}{OpenAI}.} \bibinfo{year}{2024}\natexlab{}.
\newblock \bibinfo{title}{Introducing gpt-4o: our fastest and most affordable
  flagship model.}
\newblock
\newblock
\urldef\tempurl%
\url{https://platform.openai.com/docs/guides/vision}
\showURL{%
\tempurl}
\newblock
\shownote{2024-11-07.}


\bibitem[\protect\citeauthoryear{Orogat and El-Roby}{Orogat and
  El-Roby}{2023}]%
        {qa1}
\bibfield{author}{\bibinfo{person}{Abdelghny Orogat} {and}
  \bibinfo{person}{Ahmed El-Roby}.} \bibinfo{year}{2023}\natexlab{}.
\newblock \showarticletitle{Maestro: Automatic Generation of Comprehensive
  Benchmarks for Question Answering Over Knowledge Graphs}.
\newblock \bibinfo{journal}{\emph{Proc. ACM Manag. Data}} \bibinfo{volume}{1},
  \bibinfo{number}{2}, Article \bibinfo{articleno}{177} (\bibinfo{date}{June}
  \bibinfo{year}{2023}), \bibinfo{numpages}{24}~pages.
\newblock
\urldef\tempurl%
\url{https://doi.org/10.1145/3589322}
\showDOI{\tempurl}


\bibitem[\protect\citeauthoryear{Peng, Wang, and Deng}{Peng
  et~al\mbox{.}}{2023}]%
        {nl1}
\bibfield{author}{\bibinfo{person}{Zhencan Peng}, \bibinfo{person}{Zhizhi
  Wang}, {and} \bibinfo{person}{Dong Deng}.} \bibinfo{year}{2023}\natexlab{}.
\newblock \showarticletitle{Near-Duplicate Sequence Search at Scale for Large
  Language Model Memorization Evaluation}.
\newblock \bibinfo{journal}{\emph{Proc. ACM Manag. Data}} \bibinfo{volume}{1},
  \bibinfo{number}{2}, Article \bibinfo{articleno}{179} (\bibinfo{date}{June}
  \bibinfo{year}{2023}), \bibinfo{numpages}{18}~pages.
\newblock
\urldef\tempurl%
\url{https://doi.org/10.1145/3589324}
\showDOI{\tempurl}


\bibitem[\protect\citeauthoryear{Ren, Fan, He, Huang, Dai, Huang, Jing, Zhang,
  Yang, and Wang}{Ren et~al\mbox{.}}{2024}]%
        {llm6}
\bibfield{author}{\bibinfo{person}{Tonghui Ren}, \bibinfo{person}{Yuankai Fan},
  \bibinfo{person}{Zhenying He}, \bibinfo{person}{Ren Huang},
  \bibinfo{person}{Jiaqi Dai}, \bibinfo{person}{Can Huang},
  \bibinfo{person}{Yinan Jing}, \bibinfo{person}{Kai Zhang},
  \bibinfo{person}{Yifan Yang}, {and} \bibinfo{person}{X.~Sean Wang}.}
  \bibinfo{year}{2024}\natexlab{}.
\newblock \showarticletitle{{ PURPLE: Making a Large Language Model a Better
  SQL Writer }}. In \bibinfo{booktitle}{\emph{2024 IEEE 40th International
  Conference on Data Engineering (ICDE)}}. \bibinfo{publisher}{IEEE Computer
  Society}, \bibinfo{address}{Los Alamitos, CA, USA}, \bibinfo{pages}{15--28}.
\newblock
\urldef\tempurl%
\url{https://doi.org/10.1109/ICDE60146.2024.00009}
\showDOI{\tempurl}


\bibitem[\protect\citeauthoryear{Renggli, Yao, Kolar, Rimanic, Klimovic, and
  Zhang}{Renggli et~al\mbox{.}}{2022}]%
        {vldbtransfer1}
\bibfield{author}{\bibinfo{person}{Cedric Renggli}, \bibinfo{person}{Xiaozhe
  Yao}, \bibinfo{person}{Luka Kolar}, \bibinfo{person}{Luka Rimanic},
  \bibinfo{person}{Ana Klimovic}, {and} \bibinfo{person}{Ce Zhang}.}
  \bibinfo{year}{2022}\natexlab{}.
\newblock \showarticletitle{SHiFT: an efficient, flexible search engine for
  transfer learning}.
\newblock \bibinfo{journal}{\emph{Proc. VLDB Endow.}} \bibinfo{volume}{16},
  \bibinfo{number}{2} (\bibinfo{date}{Oct.} \bibinfo{year}{2022}),
  \bibinfo{pages}{304–316}.
\newblock
\showISSN{2150-8097}
\urldef\tempurl%
\url{https://doi.org/10.14778/3565816.3565831}
\showDOI{\tempurl}


\bibitem[\protect\citeauthoryear{Shankar, Li, Asawa, Hulsebos, Lin,
  Zamfirescu-Pereira, Chase, Fu-Hinthorn, Parameswaran, and Wu}{Shankar
  et~al\mbox{.}}{2024}]%
        {llm5}
\bibfield{author}{\bibinfo{person}{Shreya Shankar}, \bibinfo{person}{Haotian
  Li}, \bibinfo{person}{Parth Asawa}, \bibinfo{person}{Madelon Hulsebos},
  \bibinfo{person}{Yiming Lin}, \bibinfo{person}{J.~D. Zamfirescu-Pereira},
  \bibinfo{person}{Harrison Chase}, \bibinfo{person}{Will Fu-Hinthorn},
  \bibinfo{person}{Aditya~G. Parameswaran}, {and} \bibinfo{person}{Eugene Wu}.}
  \bibinfo{year}{2024}\natexlab{}.
\newblock \showarticletitle{spade: Synthesizing Data Quality Assertions for
  Large Language Model Pipelines}.
\newblock \bibinfo{journal}{\emph{Proc. VLDB Endow.}} \bibinfo{volume}{17},
  \bibinfo{number}{12} (\bibinfo{date}{Nov.} \bibinfo{year}{2024}),
  \bibinfo{pages}{4173–4186}.
\newblock
\showISSN{2150-8097}
\urldef\tempurl%
\url{https://doi.org/10.14778/3685800.3685835}
\showDOI{\tempurl}


\bibitem[\protect\citeauthoryear{Song and He}{Song and He}{2021}]%
        {sigmodpattern2}
\bibfield{author}{\bibinfo{person}{Jie Song} {and} \bibinfo{person}{Yeye He}.}
  \bibinfo{year}{2021}\natexlab{}.
\newblock \showarticletitle{Auto-Validate: Unsupervised Data Validation Using
  Data-Domain Patterns Inferred from Data Lakes}. In
  \bibinfo{booktitle}{\emph{Proceedings of the 2021 International Conference on
  Management of Data}} (Virtual Event, China) \emph{(\bibinfo{series}{SIGMOD
  '21})}. \bibinfo{publisher}{Association for Computing Machinery},
  \bibinfo{address}{New York, NY, USA}, \bibinfo{pages}{1678–1691}.
\newblock
\showISBNx{9781450383431}
\urldef\tempurl%
\url{https://doi.org/10.1145/3448016.3457250}
\showDOI{\tempurl}


\bibitem[\protect\citeauthoryear{SU, Kasai, Wu, Shi, Wang, Xin, Zhang,
  Ostendorf, Zettlemoyer, Smith, and Yu}{SU et~al\mbox{.}}{2023}]%
        {votek}
\bibfield{author}{\bibinfo{person}{Hongjin SU}, \bibinfo{person}{Jungo Kasai},
  \bibinfo{person}{Chen~Henry Wu}, \bibinfo{person}{Weijia Shi},
  \bibinfo{person}{Tianlu Wang}, \bibinfo{person}{Jiayi Xin},
  \bibinfo{person}{Rui Zhang}, \bibinfo{person}{Mari Ostendorf},
  \bibinfo{person}{Luke Zettlemoyer}, \bibinfo{person}{Noah~A. Smith}, {and}
  \bibinfo{person}{Tao Yu}.} \bibinfo{year}{2023}\natexlab{}.
\newblock \showarticletitle{Selective Annotation Makes Language Models Better
  Few-Shot Learners}. In \bibinfo{booktitle}{\emph{The Eleventh International
  Conference on Learning Representations}}.
\newblock
\urldef\tempurl%
\url{https://openreview.net/forum?id=qY1hlv7gwg}
\showURL{%
\tempurl}


\bibitem[\protect\citeauthoryear{Suhara, Li, Li, Zhang, Demiralp, Chen, and
  Tan}{Suhara et~al\mbox{.}}{2022}]%
        {nl2}
\bibfield{author}{\bibinfo{person}{Yoshihiko Suhara}, \bibinfo{person}{Jinfeng
  Li}, \bibinfo{person}{Yuliang Li}, \bibinfo{person}{Dan Zhang},
  \bibinfo{person}{\c{C}a\u{g}atay Demiralp}, \bibinfo{person}{Chen Chen},
  {and} \bibinfo{person}{Wang-Chiew Tan}.} \bibinfo{year}{2022}\natexlab{}.
\newblock \showarticletitle{Annotating Columns with Pre-trained Language
  Models}. In \bibinfo{booktitle}{\emph{Proceedings of the 2022 International
  Conference on Management of Data}} (Philadelphia, PA, USA)
  \emph{(\bibinfo{series}{SIGMOD '22})}. \bibinfo{publisher}{Association for
  Computing Machinery}, \bibinfo{address}{New York, NY, USA},
  \bibinfo{pages}{1493–1503}.
\newblock
\showISBNx{9781450392495}
\urldef\tempurl%
\url{https://doi.org/10.1145/3514221.3517906}
\showDOI{\tempurl}


\bibitem[\protect\citeauthoryear{Sun, Jiachuan, Cheng, Zheng, Chen, and
  Yin}{Sun et~al\mbox{.}}{2024}]%
        {annotation2}
\bibfield{author}{\bibinfo{person}{Yushi Sun}, \bibinfo{person}{Wang Jiachuan},
  \bibinfo{person}{Peng Cheng}, \bibinfo{person}{Libin Zheng},
  \bibinfo{person}{Lei Chen}, {and} \bibinfo{person}{Jian Yin}.}
  \bibinfo{year}{2024}\natexlab{}.
\newblock \showarticletitle{Cross-Domain-Aware Worker Selection with Training
  for Crowdsourced Annotation}. \bibinfo{pages}{249--262}.
\newblock
\urldef\tempurl%
\url{https://doi.org/10.1109/ICDE60146.2024.00026}
\showDOI{\tempurl}


\bibitem[\protect\citeauthoryear{Sun, Xin, and Chen}{Sun et~al\mbox{.}}{2023}]%
        {annotation1}
\bibfield{author}{\bibinfo{person}{Yushi Sun}, \bibinfo{person}{Hao Xin}, {and}
  \bibinfo{person}{Lei Chen}.} \bibinfo{year}{2023}\natexlab{}.
\newblock \showarticletitle{RECA: Related Tables Enhanced Column Semantic Type
  Annotation Framework}.
\newblock \bibinfo{journal}{\emph{Proc. VLDB Endow.}} \bibinfo{volume}{16},
  \bibinfo{number}{6} (\bibinfo{date}{Feb.} \bibinfo{year}{2023}),
  \bibinfo{pages}{1319–1331}.
\newblock
\showISSN{2150-8097}
\urldef\tempurl%
\url{https://doi.org/10.14778/3583140.3583149}
\showDOI{\tempurl}


\bibitem[\protect\citeauthoryear{Tang, Li, and Perng}{Tang
  et~al\mbox{.}}{2011}]%
        {logsig}
\bibfield{author}{\bibinfo{person}{Liang Tang}, \bibinfo{person}{Tao Li}, {and}
  \bibinfo{person}{Chang-Shing Perng}.} \bibinfo{year}{2011}\natexlab{}.
\newblock \showarticletitle{LogSig: generating system events from raw textual
  logs}. In \bibinfo{booktitle}{\emph{Proceedings of the 20th ACM International
  Conference on Information and Knowledge Management}} (Glasgow, Scotland, UK)
  \emph{(\bibinfo{series}{CIKM '11})}. \bibinfo{publisher}{Association for
  Computing Machinery}, \bibinfo{address}{New York, NY, USA},
  \bibinfo{pages}{785–794}.
\newblock
\showISBNx{9781450307178}
\urldef\tempurl%
\url{https://doi.org/10.1145/2063576.2063690}
\showDOI{\tempurl}


\bibitem[\protect\citeauthoryear{Tang, Zhang, Zhao, Fang, Li, and Chen}{Tang
  et~al\mbox{.}}{2025}]%
        {vldbanomolydetection}
\bibfield{author}{\bibinfo{person}{Yanni Tang}, \bibinfo{person}{Zhuoxing
  Zhang}, \bibinfo{person}{Kaiqi Zhao}, \bibinfo{person}{Lanting Fang},
  \bibinfo{person}{Zhenhua Li}, {and} \bibinfo{person}{Wu Chen}.}
  \bibinfo{year}{2025}\natexlab{}.
\newblock \showarticletitle{Substructure-Aware Log Anomaly Detection}.
\newblock \bibinfo{journal}{\emph{Proc. VLDB Endow.}} \bibinfo{volume}{18},
  \bibinfo{number}{2} (\bibinfo{date}{Feb.} \bibinfo{year}{2025}),
  \bibinfo{pages}{213–225}.
\newblock
\showISSN{2150-8097}
\urldef\tempurl%
\url{https://doi.org/10.14778/3705829.3705840}
\showDOI{\tempurl}


\bibitem[\protect\citeauthoryear{Team}{Team}{2024}]%
        {qwen2.5}
\bibfield{author}{\bibinfo{person}{Qwen Team}.}
  \bibinfo{year}{2024}\natexlab{}.
\newblock \bibinfo{title}{Qwen2.5: A Party of Foundation Models}.
\newblock
\newblock
\urldef\tempurl%
\url{https://qwenlm.github.io/blog/qwen2.5/}
\showURL{%
\tempurl}


\bibitem[\protect\citeauthoryear{Teng, Li, and Chen}{Teng
  et~al\mbox{.}}{2025}]%
        {techniquereport}
\bibfield{author}{\bibinfo{person}{Fei Teng}, \bibinfo{person}{Haoyang Li},
  {and} \bibinfo{person}{Lei Chen}.} \bibinfo{year}{2025}\natexlab{}.
\newblock \showarticletitle{LLMLog: Advanced Log Template Generation via
  LLM-driven Multi-Round Annotation}.
\newblock \bibinfo{journal}{\emph{Online}} (\bibinfo{year}{2025}).
\newblock
\urldef\tempurl%
\url{https://github.com/XinTT/LLMLog}
\showURL{%
\tempurl}


\bibitem[\protect\citeauthoryear{{The Apache Software Foundation}}{{The Apache
  Software Foundation}}{2024}]%
        {spark}
\bibfield{author}{\bibinfo{person}{{The Apache Software Foundation}}.}
  \bibinfo{year}{2024}\natexlab{}.
\newblock \bibinfo{booktitle}{\emph{SparkR: R Front End for 'Apache Spark'}}.
\newblock
\urldef\tempurl%
\url{https://www.apache.org https://spark.apache.org}
\showURL{%
\tempurl}
\newblock
\shownote{R package version 3.5.1https://www.apache.org
  https://spark.apache.org.}


\bibitem[\protect\citeauthoryear{Tian, Sun, and Hu}{Tian et~al\mbox{.}}{2024}]%
        {embedding2}
\bibfield{author}{\bibinfo{person}{Xiaobin Tian}, \bibinfo{person}{Zequn Sun},
  {and} \bibinfo{person}{Wei Hu}.} \bibinfo{year}{2024}\natexlab{}.
\newblock \showarticletitle{{ Generating Explanations to Understand and Repair
  Embedding-Based Entity Alignment }}. In \bibinfo{booktitle}{\emph{2024 IEEE
  40th International Conference on Data Engineering (ICDE)}}.
  \bibinfo{publisher}{IEEE Computer Society}, \bibinfo{address}{Los Alamitos,
  CA, USA}, \bibinfo{pages}{2205--2217}.
\newblock
\urldef\tempurl%
\url{https://doi.org/10.1109/ICDE60146.2024.00175}
\showDOI{\tempurl}


\bibitem[\protect\citeauthoryear{Wang, Johnson, and Pandis}{Wang
  et~al\mbox{.}}{2017}]%
        {vldblogrootcause8}
\bibfield{author}{\bibinfo{person}{Tianzheng Wang}, \bibinfo{person}{Ryan
  Johnson}, {and} \bibinfo{person}{Ippokratis Pandis}.}
  \bibinfo{year}{2017}\natexlab{}.
\newblock \showarticletitle{Query fresh: log shipping on steroids}.
\newblock \bibinfo{journal}{\emph{Proc. VLDB Endow.}} \bibinfo{volume}{11},
  \bibinfo{number}{4} (\bibinfo{date}{Dec.} \bibinfo{year}{2017}),
  \bibinfo{pages}{406–419}.
\newblock
\showISSN{2150-8097}
\urldef\tempurl%
\url{https://doi.org/10.1145/3186728.3164137}
\showDOI{\tempurl}


\bibitem[\protect\citeauthoryear{Wang, Khan, Wu, Jin, and Yan}{Wang
  et~al\mbox{.}}{2020}]%
        {similarity}
\bibfield{author}{\bibinfo{person}{Yuxiang Wang}, \bibinfo{person}{Arijit
  Khan}, \bibinfo{person}{Tianxing Wu}, \bibinfo{person}{Jiahui Jin}, {and}
  \bibinfo{person}{Haijiang Yan}.} \bibinfo{year}{2020}\natexlab{}.
\newblock \showarticletitle{{ Semantic Guided and Response Times Bounded Top-k
  Similarity Search over Knowledge Graphs }}. In \bibinfo{booktitle}{\emph{2020
  IEEE 36th International Conference on Data Engineering (ICDE)}}.
  \bibinfo{publisher}{IEEE Computer Society}, \bibinfo{address}{Los Alamitos,
  CA, USA}, \bibinfo{pages}{445--456}.
\newblock
\urldef\tempurl%
\url{https://doi.org/10.1109/ICDE48307.2020.00045}
\showDOI{\tempurl}


\bibitem[\protect\citeauthoryear{Wang, Xin, and Chen}{Wang
  et~al\mbox{.}}{2024}]%
        {annotation4}
\bibfield{author}{\bibinfo{person}{Yubo Wang}, \bibinfo{person}{Hao Xin}, {and}
  \bibinfo{person}{Lei Chen}.} \bibinfo{year}{2024}\natexlab{}.
\newblock \showarticletitle{KGLink: A Column Type Annotation Method that
  Combines Knowledge Graph and Pre-Trained Language Model}.
\newblock \bibinfo{journal}{\emph{2024 IEEE 40th International Conference on
  Data Engineering (ICDE)}} (\bibinfo{year}{2024}),
  \bibinfo{pages}{1023--1035}.
\newblock
\urldef\tempurl%
\url{https://api.semanticscholar.org/CorpusID:270214355}
\showURL{%
\tempurl}


\bibitem[\protect\citeauthoryear{Wen, Qian, Qian, Yuan, Qin, Xuan, and
  Yuan}{Wen et~al\mbox{.}}{2024}]%
        {qa3}
\bibfield{author}{\bibinfo{person}{Zhenyu Wen}, \bibinfo{person}{Jiaxu Qian},
  \bibinfo{person}{Bin Qian}, \bibinfo{person}{Qin Yuan},
  \bibinfo{person}{Jianbin Qin}, \bibinfo{person}{Qi Xuan}, {and}
  \bibinfo{person}{Ye Yuan}.} \bibinfo{year}{2024}\natexlab{}.
\newblock \showarticletitle{Across Images and Graphs for Question Answering}.
  In \bibinfo{booktitle}{\emph{2024 IEEE 40th International Conference on Data
  Engineering (ICDE)}}. \bibinfo{pages}{1366--1379}.
\newblock
\urldef\tempurl%
\url{https://doi.org/10.1109/ICDE60146.2024.00112}
\showDOI{\tempurl}


\bibitem[\protect\citeauthoryear{Wu, Yuan, Li, Ma, and Zhang}{Wu
  et~al\mbox{.}}{2024}]%
        {embedding3}
\bibfield{author}{\bibinfo{person}{Anbiao Wu}, \bibinfo{person}{Ye Yuan},
  \bibinfo{person}{Changsheng Li}, \bibinfo{person}{Yuliang Ma}, {and}
  \bibinfo{person}{Hao Zhang}.} \bibinfo{year}{2024}\natexlab{}.
\newblock \showarticletitle{Attributed Network Embedding in Streaming Style}.
  In \bibinfo{booktitle}{\emph{2024 IEEE 40th International Conference on Data
  Engineering (ICDE)}}. \bibinfo{pages}{3138--3150}.
\newblock
\urldef\tempurl%
\url{https://doi.org/10.1109/ICDE60146.2024.00243}
\showDOI{\tempurl}


\bibitem[\protect\citeauthoryear{Xiong, Liu, Fang, Qu, Dong, Zhu, Tang, and
  Zhou}{Xiong et~al\mbox{.}}{2024}]%
        {tkdetemplate2}
\bibfield{author}{\bibinfo{person}{Zeyu Xiong}, \bibinfo{person}{Daizong Liu},
  \bibinfo{person}{Xiang Fang}, \bibinfo{person}{Xiaoye Qu},
  \bibinfo{person}{Jianfeng Dong}, \bibinfo{person}{Jiahao Zhu},
  \bibinfo{person}{Keke Tang}, {and} \bibinfo{person}{Pan Zhou}.}
  \bibinfo{year}{2024}\natexlab{}.
\newblock \showarticletitle{Rethinking Video Sentence Grounding From a Tracking
  Perspective With Memory Network and Masked Attention}.
\newblock \bibinfo{journal}{\emph{IEEE Transactions on Multimedia}}
  \bibinfo{volume}{26} (\bibinfo{year}{2024}), \bibinfo{pages}{11204--11218}.
\newblock
\urldef\tempurl%
\url{https://doi.org/10.1109/TMM.2024.3453062}
\showDOI{\tempurl}


\bibitem[\protect\citeauthoryear{Xu, Yang, Huo, Zhang, and He}{Xu
  et~al\mbox{.}}{2024}]%
        {Divlog}
\bibfield{author}{\bibinfo{person}{Junjielong Xu}, \bibinfo{person}{Ruichun
  Yang}, \bibinfo{person}{Yintong Huo}, \bibinfo{person}{Chengyu Zhang}, {and}
  \bibinfo{person}{Pinjia He}.} \bibinfo{year}{2024}\natexlab{}.
\newblock \showarticletitle{DivLog: Log Parsing with Prompt Enhanced In-Context
  Learning}. \bibinfo{pages}{1--12}.
\newblock
\urldef\tempurl%
\url{https://doi.org/10.1145/3597503.3639155}
\showDOI{\tempurl}


\bibitem[\protect\citeauthoryear{Xu, Choi, Peng, Xu, and S~Bhowmick}{Xu
  et~al\mbox{.}}{2023}]%
        {sigmodpattern1}
\bibfield{author}{\bibinfo{person}{Lyu Xu}, \bibinfo{person}{Byron Choi},
  \bibinfo{person}{Yun Peng}, \bibinfo{person}{Jianliang Xu}, {and}
  \bibinfo{person}{Sourav S~Bhowmick}.} \bibinfo{year}{2023}\natexlab{}.
\newblock \showarticletitle{A Framework for Privacy Preserving Localized Graph
  Pattern Query Processing}.
\newblock \bibinfo{journal}{\emph{Proc. ACM Manag. Data}} \bibinfo{volume}{1},
  \bibinfo{number}{2}, Article \bibinfo{articleno}{129} (\bibinfo{date}{June}
  \bibinfo{year}{2023}), \bibinfo{numpages}{27}~pages.
\newblock
\urldef\tempurl%
\url{https://doi.org/10.1145/3589274}
\showDOI{\tempurl}


\bibitem[\protect\citeauthoryear{Yang, Chen, Wang, Shang, Mao, and Zhang}{Yang
  et~al\mbox{.}}{2023}]%
        {tkdelog14}
\bibfield{author}{\bibinfo{person}{Chengcheng Yang}, \bibinfo{person}{Lisi
  Chen}, \bibinfo{person}{Hao Wang}, \bibinfo{person}{Shuo Shang},
  \bibinfo{person}{Rui Mao}, {and} \bibinfo{person}{Xiangliang Zhang}.}
  \bibinfo{year}{2023}\natexlab{}.
\newblock \showarticletitle{Dynamic Set Similarity Join: An Update Log Based
  Approach}.
\newblock \bibinfo{journal}{\emph{IEEE Transactions on Knowledge and Data
  Engineering}} \bibinfo{volume}{35}, \bibinfo{number}{4}
  (\bibinfo{year}{2023}), \bibinfo{pages}{3727--3741}.
\newblock
\urldef\tempurl%
\url{https://doi.org/10.1109/TKDE.2021.3126631}
\showDOI{\tempurl}


\bibitem[\protect\citeauthoryear{Yu, Lin, Sun, Zhou, Jiang, Huang, and
  Zhang}{Yu et~al\mbox{.}}{2022}]%
        {vldblog3}
\bibfield{author}{\bibinfo{person}{Muzhi Yu}, \bibinfo{person}{Zhaoxiang Lin},
  \bibinfo{person}{Jinan Sun}, \bibinfo{person}{Runyun Zhou},
  \bibinfo{person}{Guoqiang Jiang}, \bibinfo{person}{Hua Huang}, {and}
  \bibinfo{person}{Shikun Zhang}.} \bibinfo{year}{2022}\natexlab{}.
\newblock \showarticletitle{TencentCLS: the cloud log service with high query
  performances}.
\newblock \bibinfo{journal}{\emph{Proc. VLDB Endow.}} \bibinfo{volume}{15},
  \bibinfo{number}{12} (\bibinfo{date}{Aug.} \bibinfo{year}{2022}),
  \bibinfo{pages}{3472–3482}.
\newblock
\showISSN{2150-8097}
\urldef\tempurl%
\url{https://doi.org/10.14778/3554821.3554837}
\showDOI{\tempurl}


\bibitem[\protect\citeauthoryear{Zhang, Zhang, Lei, and Lin}{Zhang
  et~al\mbox{.}}{2018}]%
        {icdelog17}
\bibfield{author}{\bibinfo{person}{Chen Zhang}, \bibinfo{person}{Sen Zhang},
  \bibinfo{person}{Chen Lei}, {and} \bibinfo{person}{Peiguang Lin}.}
  \bibinfo{year}{2018}\natexlab{}.
\newblock \showarticletitle{Burstiness in Query Log: Web Search Analysis by
  Combining Global and Local Evidences}. In \bibinfo{booktitle}{\emph{2018 IEEE
  34th International Conference on Data Engineering (ICDE)}}.
  \bibinfo{pages}{1388--1391}.
\newblock
\urldef\tempurl%
\url{https://doi.org/10.1109/ICDE.2018.00157}
\showDOI{\tempurl}


\bibitem[\protect\citeauthoryear{Zhang, Xia, Wang, Chen, Liu, Wu, and
  Liu}{Zhang et~al\mbox{.}}{2024}]%
        {Ideal}
\bibfield{author}{\bibinfo{person}{Shaokun Zhang}, \bibinfo{person}{Xiaobo
  Xia}, \bibinfo{person}{Zhaoqing Wang}, \bibinfo{person}{Ling-Hao Chen},
  \bibinfo{person}{Jiale Liu}, \bibinfo{person}{Qingyun Wu}, {and}
  \bibinfo{person}{Tongliang Liu}.} \bibinfo{year}{2024}\natexlab{}.
\newblock \showarticletitle{{IDEAL}: Influence-Driven Selective Annotations
  Empower In-Context Learners in Large Language Models}. In
  \bibinfo{booktitle}{\emph{The Twelfth International Conference on Learning
  Representations}}.
\newblock
\urldef\tempurl%
\url{https://openreview.net/forum?id=Spp2i1hKwV}
\showURL{%
\tempurl}


\bibitem[\protect\citeauthoryear{Zhang, Qiu, Castellano, Rifai, Chen, and
  Pianese}{Zhang et~al\mbox{.}}{2023a}]%
        {tkdelog13}
\bibfield{author}{\bibinfo{person}{Tianzhu Zhang}, \bibinfo{person}{Han Qiu},
  \bibinfo{person}{Gabriele Castellano}, \bibinfo{person}{Myriana Rifai},
  \bibinfo{person}{Chung~Shue Chen}, {and} \bibinfo{person}{Fabio Pianese}.}
  \bibinfo{year}{2023}\natexlab{a}.
\newblock \showarticletitle{System Log Parsing: A Survey}.
\newblock \bibinfo{journal}{\emph{IEEE Transactions on Knowledge and Data
  Engineering}} \bibinfo{volume}{35}, \bibinfo{number}{8}
  (\bibinfo{year}{2023}), \bibinfo{pages}{8596--8614}.
\newblock
\urldef\tempurl%
\url{https://doi.org/10.1109/TKDE.2022.3222417}
\showDOI{\tempurl}


\bibitem[\protect\citeauthoryear{Zhang, Wu, Li, Tang, Tan, Li, and Cui}{Zhang
  et~al\mbox{.}}{2023b}]%
        {vldbtransfer2}
\bibfield{author}{\bibinfo{person}{Xinyi Zhang}, \bibinfo{person}{Hong Wu},
  \bibinfo{person}{Yang Li}, \bibinfo{person}{Zhengju Tang},
  \bibinfo{person}{Jian Tan}, \bibinfo{person}{Feifei Li}, {and}
  \bibinfo{person}{Bin Cui}.} \bibinfo{year}{2023}\natexlab{b}.
\newblock \showarticletitle{An Efficient Transfer Learning Based Configuration
  Adviser for Database Tuning}.
\newblock \bibinfo{journal}{\emph{Proc. {VLDB} Endow.}} \bibinfo{volume}{17},
  \bibinfo{number}{3} (\bibinfo{year}{2023}), \bibinfo{pages}{539--552}.
\newblock
\urldef\tempurl%
\url{https://doi.org/10.14778/3632093.3632114}
\showDOI{\tempurl}


\bibitem[\protect\citeauthoryear{Zhang, Xu, Lin, Qiao, Zhang, Dang, Xie, Yang,
  Cheng, Li, Chen, He, Yao, Lou, Chintalapati, Shen, and Zhang}{Zhang
  et~al\mbox{.}}{2019}]%
        {anomaly3}
\bibfield{author}{\bibinfo{person}{Xu Zhang}, \bibinfo{person}{Yong Xu},
  \bibinfo{person}{Qingwei Lin}, \bibinfo{person}{Bo Qiao},
  \bibinfo{person}{Hongyu Zhang}, \bibinfo{person}{Yingnong Dang},
  \bibinfo{person}{Chunyu Xie}, \bibinfo{person}{Xinsheng Yang},
  \bibinfo{person}{Qian Cheng}, \bibinfo{person}{Ze Li},
  \bibinfo{person}{Junjie Chen}, \bibinfo{person}{Xiaoting He},
  \bibinfo{person}{Randolph Yao}, \bibinfo{person}{Jian-Guang Lou},
  \bibinfo{person}{Murali Chintalapati}, \bibinfo{person}{Furao Shen}, {and}
  \bibinfo{person}{Dongmei Zhang}.} \bibinfo{year}{2019}\natexlab{}.
\newblock \showarticletitle{Robust log-based anomaly detection on unstable log
  data}. In \bibinfo{booktitle}{\emph{Proceedings of the 2019 27th ACM Joint
  Meeting on European Software Engineering Conference and Symposium on the
  Foundations of Software Engineering}} (Tallinn, Estonia)
  \emph{(\bibinfo{series}{ESEC/FSE 2019})}. \bibinfo{publisher}{Association for
  Computing Machinery}, \bibinfo{address}{New York, NY, USA},
  \bibinfo{pages}{807–817}.
\newblock
\showISBNx{9781450355728}
\urldef\tempurl%
\url{https://doi.org/10.1145/3338906.3338931}
\showDOI{\tempurl}


\bibitem[\protect\citeauthoryear{Zhao, Zhou, and Li}{Zhao
  et~al\mbox{.}}{2024}]%
        {zhao2024chat2data}
\bibfield{author}{\bibinfo{person}{Xinyang Zhao}, \bibinfo{person}{Xuanhe
  Zhou}, {and} \bibinfo{person}{Guoliang Li}.} \bibinfo{year}{2024}\natexlab{}.
\newblock \showarticletitle{Chat2Data: An Interactive Data Analysis System with
  RAG, Vector Databases and LLMs}.
\newblock \bibinfo{journal}{\emph{Proceedings of the VLDB Endowment}}
  \bibinfo{volume}{17}, \bibinfo{number}{12} (\bibinfo{year}{2024}),
  \bibinfo{pages}{4481--4484}.
\newblock


\bibitem[\protect\citeauthoryear{Zheng, Zou, Lian, Yu, Song, and Zhao}{Zheng
  et~al\mbox{.}}{2015}]%
        {sigmodtemplate4}
\bibfield{author}{\bibinfo{person}{Weiguo Zheng}, \bibinfo{person}{Lei Zou},
  \bibinfo{person}{Xiang Lian}, \bibinfo{person}{Jeffrey~Xu Yu},
  \bibinfo{person}{Shaoxu Song}, {and} \bibinfo{person}{Dongyan Zhao}.}
  \bibinfo{year}{2015}\natexlab{}.
\newblock \showarticletitle{How to Build Templates for RDF Question/Answering:
  An Uncertain Graph Similarity Join Approach}. In
  \bibinfo{booktitle}{\emph{Proceedings of the 2015 ACM SIGMOD International
  Conference on Management of Data}} (Melbourne, Victoria, Australia)
  \emph{(\bibinfo{series}{SIGMOD '15})}. \bibinfo{publisher}{Association for
  Computing Machinery}, \bibinfo{address}{New York, NY, USA},
  \bibinfo{pages}{1809–1824}.
\newblock
\showISBNx{9781450327589}
\urldef\tempurl%
\url{https://doi.org/10.1145/2723372.2747648}
\showDOI{\tempurl}


\bibitem[\protect\citeauthoryear{Zhou, He, Chen, and Zhang}{Zhou
  et~al\mbox{.}}{2024}]%
        {vldbhuman}
\bibfield{author}{\bibinfo{person}{Xiangmin Zhou}, \bibinfo{person}{Chengkun
  He}, \bibinfo{person}{Xi Chen}, {and} \bibinfo{person}{Yanchun Zhang}.}
  \bibinfo{year}{2024}\natexlab{}.
\newblock \showarticletitle{HSAP: A Human-in-the-Loop Social Media-Based
  Situation Awareness Platform}.
\newblock \bibinfo{journal}{\emph{Proc. VLDB Endow.}} \bibinfo{volume}{17},
  \bibinfo{number}{12} (\bibinfo{date}{Aug.} \bibinfo{year}{2024}),
  \bibinfo{pages}{4493–4496}.
\newblock
\showISSN{2150-8097}
\urldef\tempurl%
\url{https://doi.org/10.14778/3685800.3685908}
\showDOI{\tempurl}


\bibitem[\protect\citeauthoryear{Zhu, Huang, and Chaudhuri}{Zhu
  et~al\mbox{.}}{2023}]%
        {vldbpattern3}
\bibfield{author}{\bibinfo{person}{Erkang Zhu}, \bibinfo{person}{Silu Huang},
  {and} \bibinfo{person}{Surajit Chaudhuri}.} \bibinfo{year}{2023}\natexlab{}.
\newblock \showarticletitle{High-Performance Row Pattern Recognition Using
  Joins}.
\newblock \bibinfo{journal}{\emph{Proc. VLDB Endow.}} \bibinfo{volume}{16},
  \bibinfo{number}{5} (\bibinfo{date}{Jan.} \bibinfo{year}{2023}),
  \bibinfo{pages}{1181–1195}.
\newblock
\showISSN{2150-8097}
\urldef\tempurl%
\url{https://doi.org/10.14778/3579075.3579090}
\showDOI{\tempurl}


\bibitem[\protect\citeauthoryear{Zhu, He, Liu, He, Xie, Zheng, and Lyu}{Zhu
  et~al\mbox{.}}{2019}]%
        {logpai}
\bibfield{author}{\bibinfo{person}{Jieming Zhu}, \bibinfo{person}{Shilin He},
  \bibinfo{person}{Jinyang Liu}, \bibinfo{person}{Pinjia He},
  \bibinfo{person}{Qi Xie}, \bibinfo{person}{Zibin Zheng}, {and}
  \bibinfo{person}{Michael~R. Lyu}.} \bibinfo{year}{2019}\natexlab{}.
\newblock \showarticletitle{Tools and benchmarks for automated log parsing}. In
  \bibinfo{booktitle}{\emph{Proceedings of the 41st International Conference on
  Software Engineering: Software Engineering in Practice}} (Montreal, Quebec,
  Canada) \emph{(\bibinfo{series}{ICSE-SEIP '19})}. \bibinfo{publisher}{IEEE
  Press}, \bibinfo{pages}{121–130}.
\newblock
\urldef\tempurl%
\url{https://doi.org/10.1109/ICSE-SEIP.2019.00021}
\showDOI{\tempurl}


\bibitem[\protect\citeauthoryear{Zhu, Wang, Yang, Long, Wu, Tang, Zhao, and
  Wang}{Zhu et~al\mbox{.}}{2024}]%
        {llm3}
\bibfield{author}{\bibinfo{person}{Zhen Zhu}, \bibinfo{person}{Yibo Wang},
  \bibinfo{person}{Shouqing Yang}, \bibinfo{person}{Lin Long},
  \bibinfo{person}{Runze Wu}, \bibinfo{person}{Xiu Tang},
  \bibinfo{person}{Junbo Zhao}, {and} \bibinfo{person}{Haobo Wang}.}
  \bibinfo{year}{2024}\natexlab{}.
\newblock \showarticletitle{CORAL: Collaborative Automatic Labeling System
  Based on Large Language Models}.
\newblock \bibinfo{journal}{\emph{Proc. VLDB Endow.}} \bibinfo{volume}{17},
  \bibinfo{number}{12} (\bibinfo{date}{Nov.} \bibinfo{year}{2024}),
  \bibinfo{pages}{4401–4404}.
\newblock
\showISSN{2150-8097}
\urldef\tempurl%
\url{https://doi.org/10.14778/3685800.3685885}
\showDOI{\tempurl}


\end{thebibliography}
\clearpage

\section{Appendix}

\renewcommand{\proofname}{Proof}
\subsection{Proof for Theorem~\ref{theo:multiroundannotation}} \label{appendix:multiroundannotation}

\begin{proof}
	We prove Theorem~\ref{theo:multiroundannotation} by reduction from the Max Coverage problem~\cite{maxcover} to labeled log annotation problem. 
	Formally, an instance of max coverage problem is defined as $m$ subsets, i.e. $S_1$, ..., $S_m$ and a budget $k$.
	The goal is to choose at most $k$ subsets to maximize the number of covered elements $|\bigcup_{i \in J} S_i |$, where $J$ is the index of chosed subsets. 
	We can construct an instance of multiple-round labeled log annotation with $\lambda = 0$. The objective now focuses solely on maximizing the total cardinality of $|\bigcup_{i=0}^{k} I_{s_i} |$.
	Regarding $I_{s_i}$ as $S_i$ in max coverage problem, we can transform an instance of multiple-round labeled log annotation to an instance for max coverage problem. 
	If we can solve multiple-round labeled log annotation problem, we can directly find a solution to max coverage problem.
	Since max coverage problem is known as an NP-hard problem, multiple-round labeled log annotation problem is also NP-hard, which completes the proof.
\end{proof}

\subsection{Proof for Theorem~\ref{theo:annotationselection}} \label{appendix:annotationselection}

\begin{proof}
	We first prove the following lemma. 
	\begin{myLem}\label{lemma:is}
		The  $\triangle IS(s|L_r)$ in Equation~\ref{eq:marginal_score} is monotone  increasing and submodular.
	\end{myLem}
	
	\begin{proof}
		We prove the the  $\triangle IS(s|L_r)$ is monotone and submodular. as follows. 
		\begin{itemize}[leftmargin=10pt]
			\item {\textbf{Monotone increasing}: } 
			Given any $s \in U$ and selected logs $L_r$, we can obtain $ IS( L_r \cup \{s\})-IS(L_r)\geq 0$, since representative score in Equation~\eqref{eq:repre_set} and the LLM prediction confidence score in Equation~\eqref{eq:LLMhardness} is all non-negative for a new $s$ for selected log set $L_r$.
			Thus, $\triangle IS(s|L_r)$  is
			monotone increasing.
			\item {\textbf{Submodularity}: }
			Given  ${\tilde{L}}'_r \subset {L}'_r$, and $s_i \in U$ and $s_i \notin {\tilde{L}}'_r, {L}'_r$, 
			we can obtain: (omit weight for simplicity)
			\begin{align}
				\triangle  IS(s_i|{L}'_r)= C(s_i,\hat{t}_i,\hat{s}_i) + \frac{|\bigcup_{l=0}^{|{L}'_r|} I_{l} \cup I_{s_i} |}{|U|} \\
				\triangle  IS(s_i|{\tilde{L}}'_r)=  C(s_i,\hat{t}_i,\hat{s}_i) + \frac{|\bigcup_{l=0}^{|{\tilde{L}}'_r|} I_{l} \cup I_{s_i} |}{|U|}
			\end{align}
			Since  the marginal score satisfies the inequality $ \triangle IS(s_i|{L}'_r) - \triangle IS(s_i|{\tilde{L}}'_r) =   \frac{|\bigcup_{l=0}^{{L}'_r} I_{l} \cup I_{s_i} |-|\bigcup_{l=0}^{{\tilde{L}}'_r} I_{l} \cup I_{s_i} |}{|U|}\le 0$,  we can obtain the inequality as:
			
			\begin{small}
				\begin{align}
					IS(\tilde{L}'_r \cup {s_i}) -  IS(\tilde{L}'_r) \ge  	IS({L}'_r \cup {s_i}) -  IS({L}'_r).
				\end{align}
			\end{small}
			
			It demonstrates that $\triangle IS(s|L_r).$ in Equation.~\eqref{eq:marginal_score} is submodular.
		\end{itemize}
	\end{proof}
	Lemma~\ref{lemma:is} indicates that $\triangle IS(s|L_r) $ is monotone increasing and submodular.

	Suppose $IS(L_r^*)$ denotes the optimal value of objective for multiple round log annotation within budget $B_r$, then we can derive:
	\begin{equation} \label{eq:upperbound}
		IS(L_r^*) \leq IS(L_r) + B_r*\triangle IS(s|L_r)
	\end{equation}
	
	The Equation~\ref{eq:upperbound} provides the upper bound for $IS(L_r^*)$. By inductive hypothesis,
	\begin{equation} \label{eq:approximation}
		(1-(\frac{1}{B_r})^{B_r})IS(L_r^*) \leq IS(L_r)
	\end{equation}
	The fraction $(\frac{1}{B_r})^{B_r})$ approaches to $\frac{1}{e}$ as budget $B_r$ grows. Thus, the approximation ratio for multiple round log annotation is $1- \frac{1}{e}$.

\end{proof}

\subsection{Proof for Theorem~\ref{theo:adaptivedemon}} \label{appendix:adaptivedemon}

\begin{proof}
	We prove Theorem~\ref{theo:adaptivedemon} by reducing adaptive log demonstration
	selection problem to the Set-cover problem~\cite{setcover}.  Formally, one instance $I$ of the set cover problem is defined as: given a universe $U$, a family $S$ of subset of $U$ and the number $k$, the target is to minimize the size of a subset $C \subseteq S$ where union of $C$ is $U$. Here, we can transform $I$ into one instance $D_s$ of adaptive log demonstration selection problem: We set universe $U$ as words in log $s$, where each element $u \in U$ is a word $w_i^s \in s$. We set family $S$ as the word set $UW(D_s)$ in $D_s$, where $UW(D_s)$ is a subset containing words in universe $s$. Therefore, if we can find a minimum size of $D_s$, we can solve set cover problem.
	Since the  set cover problem is known as an NP-hard problem, adaptive context selection problem is also NP-hard, which completes the proof.
\end{proof}

\subsection{Proof for Theorem~\ref{theo:adaptivedemonselection}} \label{appendix:adaptivedemonselection}

\begin{proof}
	We first prove the following lemma. 
	\begin{myLem}\label{lemma:ws}
		The  $\triangle UW(s_i|D_s)$ in Algorithm~\ref{alg:adaptivedemon} is monotone  increasing and submodular.
	\end{myLem}
	
	\begin{proof}
		We prove the $\triangle UW(s_i|D_s)$ is monotone and submodular as follows. 
		\begin{itemize}[leftmargin=10pt]
			\item {\textbf{Monotone increasing}: } 
			Given any $s_i \in U$ and selected demonstrations $D_s$, we can obtain $ UW( D_s \cup \{s_i\})-UW(D_s)\geq 0$, since similar words are non-negative for a new $s_i$ for selected demonstrations $D_s$.
			Thus, $\triangle UW(s_i|D_s)$  is
			monotone increasing.
			\item {\textbf{Submodularity}: }
			Given  ${\tilde{D}}'_s \subset {D}'_s$, and $s_i \in U$ and $s_i \notin {\tilde{D}}'_s, {D}'_s$, 
			Because words in $UW({D}'_s) \setminus UW(\tilde{D}'_s)$ may have intersection with words in $s_i$, the marginal score satisfies the inequality $ \triangle UW(s_i|{D}'_s) - \triangle UW(s_i|{\tilde{D}}'_s) \le 0$. Thus, we can obtain the inequality as:
			
			\begin{small}
				\begin{align}
					UW(\tilde{D}'_s \cup {s_i}) -  UW(\tilde{D}'_s) \ge  	UW({D}'_s \cup {s_i}) -  UW({D}'_s).
				\end{align}
			\end{small}
			
			It demonstrates that $\triangle UW(s_i|D_s)$ in Equation.~\eqref{eq:marginal_score} is submodular.
		\end{itemize}
	\end{proof}
	As $\triangle UW(s_i|D_s)$ is monotone and submodular, according to theorems in \cite{setcoverproof}, greedy solution over words coverage has approximation ratio of $1+ln(n)$.

\end{proof}

\begin{table*}[h]
	\centering

	\caption{Effectiveness (accuracy) over 16 log datasets on top of Qwen2.5-7B-Instruct. 
		We exclude DivLog since AdaICL is an augmented version of it. 
		The \textbf{bold number} indicates the best performance.
	}
	\label{tab:qwenmainexp}
	\begin{tabular}{c|ccc|ccc|ccc|ccc}
		\hline
		
		\multirow{2}{*}{\textbf{Dataset}} & \multicolumn{3}{c|}{\textbf{Drain}}                                   & \multicolumn{3}{c|}{\textbf{LogPPT}}   & \multicolumn{3}{c|}{\textbf{AdaICL}}                                   & \multicolumn{3}{c}{\textbf{LLMLog}}     \\
		& \multicolumn{1}{c}{\textbf{MLA}} & \multicolumn{1}{c}{\textbf{PTA}} &  \multicolumn{1}{c|}{\textbf{RTA}}  & \multicolumn{1}{c}{\textbf{MLA}} & \multicolumn{1}{c}{\textbf{PTA}} &  \multicolumn{1}{c|}{\textbf{RTA}} & \multicolumn{1}{c}{\textbf{MLA}} & \multicolumn{1}{c}{\textbf{PTA}} &  \multicolumn{1}{c|}{\textbf{RTA}}    & \multicolumn{1}{c}{\textbf{MLA}} & \multicolumn{1}{c}{\textbf{PTA}} &  \multicolumn{1}{c}{\textbf{RTA}}    \\ \cline{1-13}
		\textbf{Android} &73.0 & 56.6 &  \multicolumn{1}{c|}{62.0}  &76.7 & 58.4  &  \multicolumn{1}{c|}{68.4}   &94.7 & 78.2 &  \multicolumn{1}{c|}{84.8}  &\textbf{99.1} & \textbf{90.1}  &  \multicolumn{1}{c}{\textbf{93.9}}   \\ 
		\textbf{BGL} &44.4 & 33.9 &  \multicolumn{1}{c|}{30.8}  &97.0 & 68.6 &  \multicolumn{1}{c|}{78.3} &89.9 & 73.1 &  \multicolumn{1}{c|}{88.3}  &\textbf{97.2} & \textbf{84.5}  &  \multicolumn{1}{c}{\textbf{91.7}}   \\ 
		\textbf{Hadoop} &43.9 & 36.8  &  \multicolumn{1}{c|}{34.2}  &89.5 & 54.0  &  \multicolumn{1}{c|}{58.8}    &99.1 &92.2  &  \multicolumn{1}{c|}{93.9}  &\textbf{99.8} & \textbf{94.8}  &  \multicolumn{1}{c}{\textbf{97.4}}   \\ 
		\textbf{HDFS} &95.9 & 81.3  &  \multicolumn{1}{c|}{92.9}  &90.0 & 85.7  &  \multicolumn{1}{c|}{85.7}   &91.7 & 15.8 &  \multicolumn{1}{c|}{64.3}  &\textbf{100.0} & \textbf{100.0}  &  \multicolumn{1}{c}{\textbf{100.0}}  \\ 
		\textbf{Linux} &19.4 & 43.4 &  \multicolumn{1}{c|}{42.2}  &94.9 & 47.5  &  \multicolumn{1}{c|}{49.1}   &97.9 &90.0 &  \multicolumn{1}{c|}{91.5}  &\textbf{99.7} & \textbf{95.7}  &  \multicolumn{1}{c}{\textbf{94.9}}   \\ 
		\textbf{Mac} &27.2& 21.2  &  \multicolumn{1}{c|}{24.9}  &67.3& 43.6  &  \multicolumn{1}{c|}{53.4}   &78.2& 46.6  &  \multicolumn{1}{c|}{63.1}  &\textbf{87.7}& \textbf{58.5}  &  \multicolumn{1}{c}{\textbf{75.1}}   \\ 
		\textbf{Thunderbird} &19.1 & 29.9 &  \multicolumn{1}{c|}{36.9}  &92.6 & 50.6  &  \multicolumn{1}{c|}{59.1}  &94.8 &76.7&  \multicolumn{1}{c|}{84.6}  &\textbf{98.3} & \textbf{77.2}  &  \multicolumn{1}{c}{\textbf{85.2}}   \\ 
		\textbf{Zookeeper} &49.8 & 39.1 &  \multicolumn{1}{c|}{36.0}  &99.0 & 74.1 &  \multicolumn{1}{c|}{86.0}  &99.8 & 92.5  &  \multicolumn{1}{c|}{99.8}  &\textbf{100.0} & \textbf{100.0}  &  \multicolumn{1}{c}{\textbf{100.0}}    \\ 
		\textbf{HealthApp} &24.1 & 8.3 &  \multicolumn{1}{c|}{34.7}  &78.9 & 85.3  &  \multicolumn{1}{c|}{85.3}   &99.3 & 87.5 &  \multicolumn{1}{c|}{93.3}   &\textbf{100.0} & \textbf{100.0}  &  \multicolumn{1}{c}{\textbf{100.0}}  \\ 
		\textbf{Spark} &37.6 & 50.0 &  \multicolumn{1}{c|}{41.7}  &99.1 & 60.0 &  \multicolumn{1}{c|}{58.3} &99.9 & 97.2&  \multicolumn{1}{c|}{97.2}   &\textbf{100.0} & \textbf{100.0}  &  \multicolumn{1}{c}{\textbf{100.0}}   \\ 
		\textbf{Windows} &69.6 & 46.3  &  \multicolumn{1}{c|}{50.0}  &98.3 & 55.4  &  \multicolumn{1}{c|}{72.0}    &71.8 & 50.3  &  \multicolumn{1}{c|}{72.0}  &\textbf{99.5} & \textbf{98.3}  &  \multicolumn{1}{c}{\textbf{99.1}}   \\ 
		\textbf{OpenSSH} &53.4 & 52.0  &  \multicolumn{1}{c|}{50.0}  &97.6 & 48.9&  \multicolumn{1}{c|}{84.6}   &93.8 & 28.3  &  \multicolumn{1}{c|}{55.6}  &\textbf{100.0} & \textbf{100.0}  &  \multicolumn{1}{c}{\textbf{100.0}}   \\ 
		\textbf{OpenStack} &18.7 & 5.5 &  \multicolumn{1}{c|}{39.5}  &90.7 & 84.4  &  \multicolumn{1}{c|}{88.4}   &48.4& 20.7 &  \multicolumn{1}{c|}{48.8}   &\textbf{100.0} & \textbf{100.0}  &  \multicolumn{1}{c}{\textbf{100.0}}   \\ 
		\textbf{Proxifier} &52.7 & 26.9  &  \multicolumn{1}{c|}{87.5}  &\textbf{100.0} & \textbf{100.0}  &  \multicolumn{1}{c|}{\textbf{100.0}}   &34.7& 11.9  &  \multicolumn{1}{c|}{12.5}  &\textbf{100.0} & \textbf{100.0}  &  \multicolumn{1}{c}{\textbf{100.0}}   \\ 
		\textbf{HPC} &67.2 & 38.8 &  \multicolumn{1}{c|}{41.3}  &94.7 & 73.6  &  \multicolumn{1}{c|}{84.8}  &96.1 & 29.9&  \multicolumn{1}{c|}{69.5}  &\textbf{100.0} & \textbf{100.0}  &  \multicolumn{1}{c}{\textbf{100.0}}  \\ 
		\textbf{Apache} &\textbf{100.0} & \textbf{100.0}  &  \multicolumn{1}{c|}{\textbf{100.0}}   &99.4 & 83.3  &  \multicolumn{1}{c|}{83.3}  &\textbf{100.0} & \textbf{100.0}  &  \multicolumn{1}{c|}{\textbf{100.0}}  &\textbf{100.0} & \textbf{100.0}  &  \multicolumn{1}{c}{\textbf{100.0}}    \\ \hline
	\end{tabular}
\end{table*}

\begin{table*}[t]
	\centering
	\small
	\caption{Ablation study on Qwen2.5-7B-Instruct }
	\label{tab:qwenablation}
	\setlength\tabcolsep{1pt}
	\vspace{-1em}
	\renewcommand{\arraystretch}{1.5}
	\begin{tabular}{c|ccc|ccc|ccc|ccc}
		\hline
		\multirow{2}{*}{\diagbox{\textbf{Dataset}}{\textbf{Model}}} & \multicolumn{3}{c|}{\textbf{Mac}} & \multicolumn{3}{c|}{\textbf{BGL}} & \multicolumn{3}{c|}{\textbf{Hadoop}} & \multicolumn{3}{c}{\textbf{Proxifier}} \\ \cline{2-13}
		& \textbf{MLA}   & \textbf{PTA} &  \textbf{RTA}  & \textbf{MLA}   & \textbf{PTA} &  \textbf{RTA}  & \textbf{MLA}   & \textbf{PTA} &  \textbf{RTA}   & \textbf{MLA}   & \textbf{PTA} &  \textbf{RTA}    \\ \cline{1-13}
		
		\textbf{LLMLog\textbackslash SED}    & $63.8_{(-23.9)}$   & $37.6_{(-20.9)}$        & $58.1_{(-17.0)}$ & $88.5_{(-8.7)}$  & $47.3_{(-37.2)}$   & $67.5_{(-24.2)}$ & $91.6_{(-8.2)}$  & $85.0_{(-9.8)}$        & $89.5_{(-7.9)}$  & $97.1_{(-2.9)}$   & $3.6_{(-96.4)}$        &$25.0_{(-62.5)}$  \\
		\textbf{LLMLog\textbackslash RS}      & $86.3_{(-1.4)}$   & $54.3_{(-4.2)}$        & $70.9_{(-4.2}$ & $88.8_{(-8.4)}$   & $46.8_{(-37.7)}$        & $65.9_{(-25.8)}$ & $94.0_{(-5.8)}$   & $90.2_{(-4.6)}$        & $96.5_{(-0.9)}$ & $100.0_{(-0.0)}$   & $100.0_{(-0.0)}$       & $100.0_{(-0.0)}$ \\ 
		\textbf{LLMLog\textbackslash PC}  & $85.5_{(-2.2)}$   & $52.2_{(-6.3)}$        & $68.9_{(-6.2)}$  & $95.6_{(-1.6)}$   & $68.5_{(-16.0)}$        & $90.9_{(-0.8)}$ & $97.2_{(-2.6)}$   & $90.9_{(-3.9)}$        & $96.5_{(-0.9)}$ & $100.0_{(-0.0)}$   & $100.0_{(-0.0)}$       & $100.0_{(-0.0)}$\\
		\textbf{LLMLog\textbackslash AD}     & $68.9_{(-20.8)}$  & $26.4_{(-32.1)}$       & $51.9_{(-23.2)}$ & $70.7_{(-26.5)}$  & $16.6_{(-67.9)}$       & $50.0_{(-41.7)}$ & $58.8_{(-41.0)}$  & $15.5_{(-79.3)}$       & $64.9_{(-32.5)}$ & $20.2_{(-79.8)}$   & $9.6_{(-90.4)}$       & $12.5_{(-87.5)}$\\ 
		\textbf{LLMLog\textbackslash AB}       & $86.6_{(-1.1)}$  & $55.1_{(-3.4)}$       & $72.9_{(-2.2)}$ & $89.0_{(-8.2)}$  & $67.1_{(-17.4)}$       & $81.7_{(-10.0)}$ & $99.6_{(-0.2)}$   & $93.9_{(-0.9)}$       & $97.4_{(-0.0)}$ & $100.0_{(-0.0)}$   & $100.0_{(-0.0)}$       & $100.0_{(-0.0)}$\\  \hline
		\textbf{LLMLog}     & \textbf{87.7} & \textbf{58.5}        &\textbf{75.1} & \textbf{97.2} & \textbf{84.5}     &\textbf{91.7} & \textbf{99.8}  & \textbf{94.8}        &\textbf{97.4} & \textbf{100.0}  & \textbf{100.0}        & \textbf{100.0}\\ \hline
	\end{tabular}
\end{table*}

\subsection{Additional Experiments}

In this section, we firstly present the detailed description of included baselines. Then we present the effectiveness results and ablation study on Qwen2.5-7B-Instruct model.

\subsubsection{Detailed Baselines}
The baselines in experiment are listed as follows:
\begin{itemize}[leftmargin=10pt]
	\item \textbf{Drain} Drain is the most effective unsupervised template generation methods for logs. 
	Given a log, it process the log on a fixed tree structure, where each path on the tree stands for a heuristic and each leaf node represents a specific template.
	The template will be generated once the log reaches a leaf node.

	\item \textbf{LogPPT} 
	LogPPT is the state-of-the-art (SOTA) ML-based method for the template generation task.
	It retrieves several of the most representative logs by sampling.
	Then, it fine-tunes a pre-trained language model to predict keyword labels in the sampled training set.
	After training, the fine-tuned language model can infer templates for each log.

	\item \textbf{DivLog} DivLog~\cite{Divlog} is the first DivLog~\cite{Divlog} is the first LLM-based method for the template generation task.
	It applies the DPP~\cite{dpp} algorithm to select diverse logs for annotation.
	For each unlabeled log, it retrieves the top-$k_c$ most similar labeled logs using the cosine similarity of embeddings as demonstrative contexts, which are fed into the LLM to guide template generation.
	
	\item \textbf{AdaICL} 
We combine DivLog~\cite{Divlog} with the existing multi-round annotation framework, AdaICL~\cite{adaicl}.
Specifically, in each round, AdaICL~\cite{adaicl} selects labeled logs by retrieving the top-$k_n$ most similar logs for each unlabeled log based on the representativeness score.
It then iteratively finds the log set that maximizes the representativeness score within a pre-defined budget.
For the in-context learning (ICL) component, the baseline maintains the same settings as DivLog~\cite{Divlog}.
\end{itemize}

\subsubsection{Effectiveness Evaluation on Qwen2.5-7B-Instruct}  \label{appendix:effective_exp}
We report the MLA, PTA, and RTA metrics on 16 datasets using the Qwen2.5-7B-Instruct model.
Because LLMs lack the ability to correctly generate templates without in-context learning (ICL), we exclude the pure LLM baseline from our evaluation.
Since AdaICL~\cite{adaicl} is an extension of DivLog~\cite{Divlog} with better accuracy, we include only AdaICL~\cite{adaicl} as the in-context learning baseline.
The results in Table~\ref{tab:qwenmainexp} illustrate the robustness of the proposed components across multiple types of LLMs.
As Qwen2.5-7B-Instruct has fewer parameters than GPT-4o, the overall performance of both LLMLog and AdaICL is downgraded.
On datasets with limited budgets, the accuracy of AdaICL degrades severely, and even the neural network-based method, LogPPT, outperforms it.
However, our model remains accurate across all 16 datasets, achieving $100\%$ accuracy on several datasets such as HDFS, Proxifier, HPC, and Apache under a limited budget. This is because the word-based annotation and demonstration selection provide clearer instructions for the LLM to process each log precisely.

\subsubsection{Ablation Study on Qwen2.5-7B-Instruct}\label{appendix:ablation}
For further verifying our analysis, we also summarize experiments on top of Qwen2.5-7B-Instruct model in Table~\ref{tab:qwenablation}. 
The variants and datasets maintain the setting in Table~\ref{tab:ablation}.
The results on QWEN show the similar trend for SED, representative score, LLM confidence score and adaptive budget strategy. 
However, the removal of adaptive context strategy severely affects the performance of LLMLog. 
Since Qwen2.5-7B-Instruct has fewer parameters than GPT-4o, simply piling fixed top-$k_c$ similar example logs may not provide enough contexts, illustrating that our adaptive demonstration strategy is more robust on smaller LLMs. 
This point can be also supported by the more significant decrease of PTA, resulting from the falsely processed words. 
As there are only 4 templates in Proxifier, the wrong prediction for several logs results in large perturbation of template level accuracy, PTA and RTA.
 
\end{document}